\documentclass[a4paper,11pt]{article}
\usepackage{jheppub} % for details on the use of the package, please see the JINST-author-manual

\usepackage{graphicx} % Required forinserting images
\usepackage{url}
\usepackage{amsmath}
\usepackage{subcaption}
\usepackage{enumitem}
\usepackage{amssymb}
\usepackage{outlines}
\usepackage{lipsum}
\usepackage{multirow}
\usepackage{mathrsfs}
\usepackage{hyperref}

\usepackage{color}  
\usepackage{tikz}
\usetikzlibrary{trees}
\usepackage{comment}
\usepackage{amsthm}
\usepackage{bbm}
\usepackage{dsfont}
\usepackage{float}
\theoremstyle{plain}
\newtheorem{thm}{Theorem}[section]

\newtheorem{lem}[thm]{Lemma}
\newtheorem{prop}[thm]{Proposition}
\newtheorem{cor}{Corollary}

\theoremstyle{definition}
\newtheorem{defn}{Definition}[section]

\theoremstyle{remark}
\newtheorem*{rem}{Remark}

\newcommand{\bs}[1]{\boldsymbol{#1}}

\hypersetup{
    colorlinks=true, %set true if you want colored links
    linktoc=all,     %set to all if you want both sections and subsections linked
    linkcolor=blue,  %choose some color if you want links to stand out
}

\makeatletter
\def\@fpheader{\relax}
\makeatother

\title{Maximally heavy dynamics in the causal diamond}
\author{David Poland, Gordon Rogelberg}
\affiliation{Department of Physics, Yale University, 217 Prospect St, New Haven, CT 06520, USA}
\emailAdd{david.poland@yale.edu, gordon.rogelberg@yale.edu}

\abstract{Correlation functions of CFT operators with infinite scaling dimension are rich, multifaceted objects that describe physics ranging across classical holography, black hole dynamics, and flat-space scattering amplitudes. In this work, we provide a rigorous framework for characterizing the space of four-point functions of identical operators with infinite dimension in terms of well-defined ``maximally heavy observables,'' which are akin to intrinsic quantities describing statistical systems in the thermodynamic limit. These observables are highly constrained by crossing symmetry and unitarity, and give novel insights into the locality of bulk states through the emergence of dynamical phase transitions. In certain cases, these results connect directly to the more familiar picture of torus partition functions at large central charge. We apply our framework to a number of illustrative examples including generalized free fields, chiral product correlators, and maximal giant gravitons in planar $\mathcal{N} =4$ SYM.}

\newcommand{\bh}{{\bar{h}}}

\newcommand{\rhob}{{\bar{\rho}}}

\newcommand{\cO}{\mathcal O}

\newcommand{\be}{\begin{equation}}
\newcommand{\ee}{\end{equation}}
\newcommand{\bea}{\begin{eqnarray}}
\newcommand{\eea}{\end{eqnarray}}
\newcommand{\ba}{\begin{equation} \begin{aligned}}
\newcommand{\ea}{\end{aligned} \end{equation}}

\newcommand{\mycomment}[1]{}

\newcommand{\mbf}{\mathbf}

\newcommand{\Df}{\Delta_\phi}

\newcommand{\De}{\Delta}

\newcommand{\U}{\textrm{U}}

\DeclareMathOperator{\sech}{sech}
\DeclareMathOperator{\csch}{csch}

\begin{document}

\maketitle

\pagenumbering{roman}
\setcounter{page}{2}
\newpage
\pagenumbering{arabic}
\setcounter{page}{1}

\section{Introduction}

The AdS/CFT correspondence provides a natural language for describing quantum gravity in asymptotically anti-de Sitter (AdS) space in terms of a conformal field theory (CFT) in one dimension lower~\cite{Maldacena:1997re,Gubser:1998bc,Witten:1998qj}. This language has been highly developed for correlation functions of operators with small scaling dimensions, which can be computed using weakly-coupled bulk perturbation theory~\cite{Arutyunov:2003ad,Chicherin_2016,Chicherin:2018avq}. In recent years, there has been an increased interest in studying correlation functions of CFT operators with large scaling dimensions, which are holographically dual to heavy states in the bulk~\cite{Fitzpatrick:2015zha,Asplund:2014coa,Collier:2019weq}. These states can take a variety of forms, including macroscopic objects such as stars, black holes, or compact binaries~\cite{Horowitz:2000fm,Hoyos:2021uff,Abajian2023-xw}. Heavy four-point correlators~\cite{Kim2015-sg,Poland:2024hvb} can thus encode the rich physics of black hole formation and thermalization, gravitational phase transitions, and the emergence of classical spacetime geometry. Experimental breakthroughs in strong-field gravity, such as the detection of gravitational wave ringdowns by LIGO~\cite{LIGOScientific:2016aoc} and the imaging of supermassive black holes by the Event Horizon Telescope~\cite{EventHorizonTelescope:2019dse}, have enabled direct observation of these phenomena and underscored the urgency of understanding them from first principles. 

Heavy states can most generally be defined as local CFT operators with scaling dimension $\De \gg d$, where $d$ is the spacetime dimension of the theory. In a holographic CFT with large central charge, $c_T \propto \langle T T\rangle$, and a parametrically large gap in the spectrum of single-trace operators, the characterization may be further refined by comparing the quantum numbers of operators with different powers of the central charge~\cite{Hartman:2014oaa, Asplund:2014coa, Ghosh:2019rcj,Pal:2023cgk,Alkalaev:2024knk}. For scalar operators labeled by dimension $\De$, we propose three main categories:
\ba
\De \ll \sqrt{c_T} \quad & \Rightarrow \quad  \text{light operators},& \\
 \sqrt{c_T} \lesssim \De \ll c_T \quad & \Rightarrow \quad \text{heavy operators},\\
 \De \gtrsim c_T \quad & \Rightarrow \quad \text{huge operators}.
\ea
The threshold between light and heavy operators is motivated by how the operator couples to leading multi-stress tensor families, which are dual to composite graviton states in the bulk. Namely, we can compute the three-point couplings between two identical operators $\cO$ with $\mathrm{dim}\; \cO = \De$ and the leading twist multi-stress tensor families, schematically written as $T^n$~\cite{Karlsson:2021duj}. At large $c_T$, these couplings scale as
\ba
\frac{\langle \cO \cO T^n\rangle}{\sqrt{\langle T^n T^n \rangle }} \sim \left(\frac{\De}{\sqrt{c_T} }\right)^n.
\ea
For $ \De \ll \sqrt{c_T}$, these couplings are suppressed at large central charge, telling us that light operators are weakly coupled to gravity in the semi-classical bulk theory. 

Light states furnish many descriptions. In 4d $\mathcal{N} = 4$ super Yang-Mills (SYM) with a rank $N$ gauge group and 't Hooft coupling $\lambda = g_{\mathrm{YM}}^2N$, the simplest examples are single-trace 1/2-BPS chiral primaries $\mathrm{Tr}(\Phi_{I_1} \Phi_{I_2} ... \Phi_{I_J})$ with protected dimension $\De = J$, dual to Kaluza-Klein supergravity modes on $S^5$~\cite{Arutyunov:2003ad}. The lightest unprotected single-trace primary is the Konishi operator $\mathcal{K} \sim \mathrm{Tr}(\Phi^I \Phi_I)$, whose dimension interpolates from $\De = 2 + O(\lambda)$ at weak coupling to $\De \sim 2\lambda^{1/4}$ at strong coupling~\cite{Eden:2012fe,Gromov:2013pga}. Operators with parametrically large dimension arise as semiclassical string states whose worldsheets trace out classical solutions in AdS${}_5 \times S^5$. Notable examples include BMN operators~\cite{Berenstein:2002adc} with parametrically large R-charge $Q\sim \sqrt{N}$, and giant magnons dual to open string solitons on $\mathbb{R} \times S^2$ with $\De - Q \sim \sqrt{\lambda}$~\cite{Hofman:2006xt}. Despite being allowed to have parametrically large dimensions, all light operators share the property that they decouple from the stress tensor families in the planar limit. For many, their correlation functions have been computed perturbatively with Witten diagrams or integrability~\cite{Arutyunov:2003ad, Beisert:2010jr, Chicherin_2016, Chicherin:2018avq}. 

For $\sqrt{c_T} \lesssim \De \ll c_T$, the heavy state is strongly coupled to gravity but does not fully backreact on the AdS bulk. The canonical examples in $\mathcal{N} = 4$ SYM are giant gravitons with $\De \sim N$, which are semi-classically described by D3-branes wrapping an $S^3 \subset S^5$ in AdS${}_5 \times S^5$, stabilized against collapse by their large angular momentum~\cite{McGreevy:2000cw}.  In the dual gauge theory, these states are described by Schur polynomial operators $\chi_R(y \cdot\Phi )$, where $y \cdot \Phi  = y^I \Phi_I$ is the linear combination of the six real adjoint scalars $\Phi^I$ specified by a null polarization vector $y^I$ on $S^5$, and $R$ is a Young diagram representation of $S_m$ with $m = \De$ boxes~\cite{Corley:2001zk, deMelloKoch:2024sdf}. These operators provide a complete orthogonal basis for the 1/2-BPS sector at finite $N$~\cite{Corley:2001zk}. The Cayley-Hamilton theorem constrains the Young diagram to have at most $N$ rows, and in the totally antisymmetric representation with exactly $N$ rows, the Schur polynomial reduces to the determinant operator $\det(y\cdot\Phi )$, known as a maximal giant graviton~\cite{McGreevy:2000cw, Balasubramanian:2001nh}. Maximal giant gravitons carry the largest angular momentum on $S^5$ and become fully extended within the compact subspace while remaining point-like in AdS${}_5$. These states can therefore be understood as defect systems, localizing to a boundary-to-boundary geodesic in the ambient AdS spacetime~\cite{Chen:2025yxg,Chen:2026ium}. 

Operators with $\De \gtrsim c_T \sim N^2$ are called ``huge" and fully backreact on the AdS geometry, effectively filling the bulk to produce a different metric entirely. Holographically, these states can be constructed by assembling stacks of $O(N)$ D-branes, whose collective gravitational backreaction generates a new asymptotically AdS geometry. In $\mathcal{N} = 4$ SYM, the simplest class of huge operators are symmetric Schur polynomials with $\De \sim N^2$, which are dual to systems of \textit{dual} giant gravitons: D3-branes wrapping an $S^3 \subset \mathrm{AdS}_5$ rather than the compact subspace~\cite{Grisaru:2000zn}. In the supergravity limit, these configurations are captured by the ``bubbling'' 1/2-BPS geometries of~\cite{Lin:2004nb}, which provide a complete geometric classification of the 1/2-BPS sector at large $N$. Studying light correlators in a huge dual giant graviton background provides a novel approach to probing the physics of a gapped QFT within the framework of a CFT~\cite{Ivanovskiy:2024vel}. Low-point correlators of huge operators have been computed in~\cite{Abajian2023-xw, Abajian:2023bqv,Kazakov:2024ald}, where it was shown that reproducing the CFT correlator holographically requires an extended horizon contribution in the form of a Gibbons-Hawking-York term in the bulk effective action, reflecting the fact that these states source horizons.
 
The eigenstate thermalization hypothesis (ETH) provides a key physical distinction between heavy and huge operators~\cite{Karlsson:2021duj}. States with $\De \gtrsim c_T$ are so excited that they behave as thermal baths, and we may equate the expectation value of some light operator $\phi$ in this $|\cO_{\text{huge}}\rangle$ state with its thermal expectation value in the micro-canonical ensemble
\ba
\langle \cO_{\text{huge}} | \phi |\cO_{\text{huge}} \rangle \sim \langle \phi\rangle_\beta.
\ea
The effective temperature scales as $T \sim (\De / c_T)^{1/d}$~\cite{Benjamin2023-qo,Asplund:2014coa}, so huge operators with $\De \gtrsim c_T$ source finite-temperature backgrounds while heavy operators have zero temperature. Evidence for this thermality was presented in~\cite{Asplund:2014coa} and subsequently established rigorously in 2d CFT using modular bootstrap techniques~\cite{Kraus:2016nwo,Das:2017vej} and asymptotics of semiclassical Virasoro vacuum blocks~\cite{Das:2020jhy,Fitzpatrick:2016ive}.

In the planar limit $N \to \infty$, one can construct light, heavy, and huge operators which all have strictly infinite scaling dimensions, while admitting completely different semi-classical descriptions. Even in a non-holographic CFT, the spectrum of global primary operators is unbounded \cite{Qiao:2017xif,Mukhametzhanov:2018zja}, so it remains possible to take a limit where one obtains a space of operators with strictly infinite scaling dimension. We call this space of operators in a generic CFT \textit{maximally heavy}, and we call a four-point function of maximally heavy operators a \textit{maximally heavy correlator}. Another example of correlators with strictly infinite external scaling dimensions arises in the study of the flat-space limit of rigid QFT in AdS space. The theory on the asymptotic conformal boundary of these rigid QFTs is given by a conformal theory (CT), which is distinguished from a CFT by its lack of a conserved stress tensor \cite{Paulos2016-mh}. Due to the scalar mass-dimension relation of $\De (\De - d) = m^2 R^2$, taking the AdS radius $R \to \infty$ with $m$ finite results in the dimensions of all non-identity operators in the CT becoming infinite. In turn, a particular subset of the space of maximally heavy correlators can also be mapped to the space of S-matrices of rigid flat-space QFT \cite{vanRees:2022zmr,Paulos2016-mh,vanRees:2023fcf,Penedones:2010ue}.

Despite the significance of maximally heavy correlators in the settings of classical gravity in asymptotically AdS space, black hole dynamics, thermalization, and flat-space QFT amplitudes, they suffer from an extremely poor formal understanding both perturbatively and in the context of the non-perturbative conformal bootstrap. The reasons for this lack of understanding are numerous. First, these correlators are generically unbounded in any neighborhood of the t-channel OPE limit of $v \to 0$. This is because the t-channel identity contribution is $(u/v)^\De$, and taking $\De \to \infty$ for any $u >v$ gives an exponentially divergent contribution to the OPE decomposition on an open subset of the domain where both the s and t-channel converge. Thus, a formal understanding of these correlators requires some type of regularization to produce a well-behaved function that is amenable to analysis. Second, unlike correlation functions of light single-trace operators, large-$N$ factorization is known to break down for huge operators, due to the fact that they delocalize away from bulk geodesics \cite{Abajian2023-xw,Abajian:2023bqv,Lin:2004nb}. This makes direct computation of these correlators extremely complicated even at tree level, and the prospect of obtaining higher order perturbative corrections at both weak and strong coupling becomes increasingly intractable. Third, since maximally heavy correlators do not satisfy the property of light-state dominance observed in \cite{Kraus:2018pax}, traditional numerical bootstrap techniques do an extremely poor job of constraining the space of these correlators, due to the fact that they are not described well by a finite number of low-lying states. Indeed, the physics of maximally heavy correlators may only be captured by studying the collective behavior of a large number $(\sim\sqrt{\De})$ of operators, which similarly becomes infinite as we take $\De \to \infty$. Therefore, a bootstrap treatment of maximally heavy correlators should not rely on computing the CFT data associated with any finite number of individual operator contributions to the OPE, and should instead elucidate the \textit{emergent} structure of these contributions in order to accurately understand maximally heavy dynamics. 

In~\cite{Poland:2024hvb}, we initiated a bootstrap study of maximally heavy correlators in terms of OPE moments: integrals of powers of exchanged quantum numbers against a positive measure encoding the conformal block decomposition of the correlator at the self-dual point $u = v = 1/4$. Crossing symmetry forces odd central moments to vanish in the heavy limit, and combining this with unitarity yields two-sided bounds on all moments. Recently, these bounds were rigorously strengthened and extended to the light-correlator regime using numerical semi-definite programming in \cite{Chiang:2026nmd}. We also analyzed the properties of saddle-point structures in the OPE measure for free multi-particle correlators. This analysis suggested that emergent structures in the OPE spectrum are the key data for characterizing heavy dynamics, rather than any finite number of exchanged operator dimensions and OPE coefficients. However, several key limitations remained, such as the fact that odd moment constraints alone are insufficient to prove the symmetry of the OPE measure around its mean~\cite{churchill1}. Moreover, the entire analysis was restricted to the self-dual point, with no control over the globally ill-defined nature of the correlator.

In this work, we develop a systematic framework for studying maximally heavy correlators of identical scalar operators in a Lorentzian kinematic regime called the \textit{causal diamond}, where both the s and t-channel OPE decompositions converge absolutely and the correlator is manifestly positive. Our approach is built on a close analogy with statistical mechanics. We treat the correlator as a partition function, the external scaling dimension $\De$ as the system size, and the cross ratios as control parameters analogous to temperature. The problem of the $\De \to \infty$ limit then becomes thermodynamic in nature: rather than studying the divergent partition function itself, one should extract intensive quantities that remain well-defined and physically meaningful.

We make this precise by introducing a framework of correlator sequences and regularizations: \textit{heavying sequences} of correlators $\{\mathcal{G}_n\}$ are those whose external dimensions $\{\De_n\}$ form a monotonically increasing and unbounded positive sequence, and \textit{good regularizations} map this sequence of correlators into topological target spaces such that accumulation points are guaranteed to exist. The set of these accumulation points defines a set of \textit{maximally heavy observables} associated to a heavying sequence in the image of a good regularization. Different good regularizations can yield distinct sets of maximally heavy observables from the same heavying sequence, each offering a unique characterization of the underlying dynamics. 

In this work, we study three good regularizations, briefly summarized as follows: The \textit{dynamical free energy density}, defined as the logarithm of the correlator divided by the external dimension, plays the role of a free energy per unit volume. We prove that a heavying sequence of correlators in the image of this regularization is uniformly bounded and equicontinuous, so that there exist maximally heavy observables called \textit{rate functions} (proposition \ref{preaa}). These observables are locally uniformly bounded and Lipschitz continuous functions on the causal diamond that encode the global phase structure of the correlator. Additionally, we derive uniform two-sided bounds for rate functions as a result of crossing symmetry and unitarity (proposition \ref{universalratefunctionbound}). The \textit{rescaled OPE measure}, obtained by uniformly rescaling the spectral parameters of the OPE by a factor of $\De$, captures the local structure of the OPE near the self-dual point. We prove that the mass of this sequence of measures does not escape to infinity in the heavy limit, so it admits maximally heavy observables called \textit{classical measures} (proposition \ref{tightnessofmeasuresequence}). These measures are compactly supported probability measures equipped with an exact discrete group symmetry determined by the crossing and chiral symmetry of the correlator (theorem \ref{thm:maxheavy}), revealing a setting for which the approximate reflection symmetry described in \cite{Kim2015-sg} becomes exact. The third maximally heavy observable is the \textit{$\alpha$-local rate function}, which probes deviations of size $\sim 1/\De^\alpha$ around the self-dual point and interpolates between the standard rate function and the cumulant generating function of the classical measure as $\alpha$ varies from $0$ to $1$.

In order to understand the global implications of the local picture provided by the classical measure, we introduce the \textit{coherent state decomposition}: a canonical procedure for building exact solutions to the s-t crossing equation at any finite external dimension from a classical measure. These coherent states are approximate minimal wavepackets on a suitably defined phase-space that narrow into delta masses in the classical limit of $\hbar \equiv 1/\De \to 0$. We provide conditions under which these solutions accurately reproduce the true correlator within a neighborhood of size $\sim 1/\sqrt{\De}$ around self-duality (theorem \ref{cohstatedecompthm}), and show that the rate function of a coherent state correlator takes a simple form depending only on the convex hull of the support of the classical measure (proposition \ref{ratefunctioncompute}). A key result connecting the local and global pictures is a \textit{matching theorem} (theorem \ref{matchinglocalratefunctions}), which identifies the precise range of $\alpha$ for which the $\alpha$-local rate function is determined by the classical measure alone.

This framework reveals a variety of phase structures that characterize the dynamics of maximally heavy operators. For correlators that are sub-exponentially bounded at the self-dual point in the heavy limit, the rate function along the diagonal $\bar{z} = z$ exhibits a universal phase transition between a phase dominated by the s-channel identity and one dominated by the t-channel identity, occurring precisely at the location of the mutual information transition studied in \cite{Ceyhan:2025qrj}. The physical origin of this transition is the localization of heavy operators to free geodesics in the bulk, made evident by studying the world-line effective action of the disconnected correlator. Away from the diagonal limit, large spin saddles may dominate the OPE decomposition to produce non-universal classical ``vortex" phases. These non-universal dynamics are characterized by the convex hull of the classical measure, and can be bounded under the assumption of a classical twist gap constraining the locations of non-identity saddles. Using a change of variables, we map these phase transitions directly to the Hawking-Page transition observed in torus partition functions of 2d CFT at large central charge. We then compare our results to the optimal bounds on the locations of non-universal regions conjectured in \cite{Hartman:2014oaa} and recently proved in \cite{Dey:2024nje}, and find that they exactly agree near the self-dual point.

For correlators that grow exponentially or faster at the self-dual point in the heavy limit, the dynamics exhibited by the rate function are far less constrained. The divergent nature of these correlators implies that non-identity contributions dominate the OPE and modify the universal transition predicted by the identity operator. The resulting rate functions may exhibit a smooth crossover or a non-universal transition encoded by interacting world-lines deformed from their geodesic in pure AdS. These possibilities are unique to the setting of four-point functions and do not have an analog to the behavior of torus partition functions, thus capturing a genuinely novel class of heavy dynamics.

We apply this formalism to a number of explicit examples, including correlators of $\phi^L$ operators in generalized free theories, chiral product correlators in 2d, and tree-level correlators of maximal giant gravitons in $\mathcal{N} = 4$ SYM. Using different maximally heavy observables to characterize different limits of these correlators, we organize maximally heavy operators into three physically distinct classes based on their degree of bulk localization:
\begin{itemize}[leftmargin=*]
\item {\bf Localized:} Correlators of localized operators are those whose rate functions exhibit sharp phase transitions around the self-dual point that are determined by the convex hull of their corresponding classical measure. These arise for correlators of $\phi^L$ operators in a generalized free theory when $\Df \to \infty$ at a fixed composite length $L$. In this example, the operators localize to free bulk geodesics, producing a universal phase structure analogous to the mutual information transition observed in \cite{Ceyhan:2025qrj}. Another notable example of localized operators are maximal giant gravitons in the planar limit \cite{Vescovi:2021fjf}.
\item {\bf Delocalized:} In contrast to localized states, correlators of delocalized operators have rate functions that exhibit a smooth crossover rather than a sharp transition around the self-dual point, and a classical measure that collapses to a single delta-mass. In the GFF example, these correspond to the limit of $L \to \infty$ with $\Df$ held fixed. These GFF states are interpreted as a gas of fixed mass particles that do not localize along any bulk world-line. While their classical measures collapse to a universal form, their rate functions remain sensitive to the microscopic details of the constituents, such as the fixed dimension of an individual field $\Df$.

\item {\bf Quasi-localized:} Correlators of quasi-localized operators interpolate between the two extremes, and are the most interesting from a diagnostic standpoint. While the classical measure for these correlators is trivial, the rate function still exhibits a sharp transition. Thus, these states cannot be fully characterized by either the classical measure or the global rate function alone; only the $\alpha$-local rate function at the critical value $\alpha = \alpha_c$ resolves their microscopic structure. The critical exponent $\alpha_c$ thus serves as a quantitative measure of the degree of bulk localization. In the GFF example, these states are produced in the ``fragmenton" limit by setting $\Df = \kappa \De^{\alpha_c}$ and $L = \lfloor \De^{1-\alpha_c } / \kappa\rfloor$ with $\alpha_c \in [0,1]$, $\kappa \in \mathbb{R}_+$, and taking $\De \to \infty$.
\end{itemize}

The paper is organized as follows. In section \ref{setup}, we set up the problem and discuss the details of the Lorentzian configurations we consider, introduce the radial monomial decomposition as a parallel to the $q$-expansion of the torus partition function, and develop the OPE measure formalism. In section \ref{maxheavy}, we formalize the notion of maximally heavy observables in terms of heavying sequences and good regularizations, and present the three regularization schemes described above. We prove the existence and structural properties of rate functions and classical measures, introduce the coherent state decomposition, prove the matching theorem connecting local and global pictures, and compare our bounds on the phase structure to those of Hartman, Keller, and Stoica~\cite{Hartman:2014oaa} and Dey, Pal, and Qiao~\cite{Dey:2024nje}. In section \ref{sec:examples}, we apply our formalism to explicit examples and develop the localized, delocalized, and quasi-localized classification. We conclude with a discussion of our results in section \ref{disc}, and relegate detailed proofs to the appendices.

\subsection*{Summary of notation}
\label{sec:notation}

{\small
\renewcommand{\arraystretch}{1.2}
\setlength{\tabcolsep}{4pt}
\setlength{\hfuzz}{5pt}

\noindent
\begin{tabular*}{\linewidth}{@{} l @{\;} l @{\quad} l @{\;} l @{}}
\hline
\noalign{\vskip10pt}
\multicolumn{2}{@{}l}{\textbf{Asymptotic notation}} &
\multicolumn{2}{@{}l}{\textbf{Bold 2-tuple ($\bs{x}=(x,\bar{x})$, $\bs{y}=(y,\bar{y})$)}} \\[1pt]
$f\sim g$   & $f/g\to C\in\mathbb{R}$      & $\bs{x}^{\bs{y}}$   & $x^y\bar{x}^{\bar{y}}$ \\
$f\simeq g$ & $f/g\to 1$                   & $\bs{x}\cdot\bs{y}$ & $xy+\bar{x}\bar{y}$ \\
$f\gg g$    & $f/g\to\infty$               & $\bar{\bs{x}}$      & $(\bar{x},x)$ \\
$f=O(g)$    & $|f|\le M|g|$   & $\|\bs{x}\|$        & $\sqrt{x^2+\bar{x}^2}$ \\
$f=o(g)$    & $f\ll g$                     & $|f(\bs{x})|^2$     & $f(x)f(\bar{x})$ \\
&& $(f(\bs{x}))$ & $(f(x),f(\bar{x})) $\\
\noalign{\vskip10pt}
\multicolumn{2}{@{}l}{\textbf{Kinematic variables}} &
\multicolumn{2}{@{}l}{\textbf{Domains \& sets}} \\[1pt]
$u,v$           & $z\bar{z}$,\;$(1{-}z)(1{-}\bar{z})$        & $\mathscr{C}$           & $\{0<z,\bar{z}<1\}\subset\mathbb{R}^2$ \\
$\rho,\bar\rho$ & $\rho=\frac{1-\sqrt{1-z}}{1+\sqrt{1-z}}$   & $\widehat{\mathscr{C}}$ & $(\rho,\bar\rho)\in D^2\times D^2$ \\
$\hat\rho$      & image of $\rho$ under $z\mapsto 1{-}z$     & $T^2_{\bs{\rho}_0}$     & torus fiber above $\bs{\rho}_0\in\mathscr{C}$ \\
$\chi,\bar\chi$ & $\chi=e^x=z/(1{-}z)$                       & $\mathcal{D}$           & $[0,\sqrt{2}]^2\subset\mathbb{R}^2_+$ \\
$x,\bar{x}$     & $\log(z/(1{-}z))$;\;$\bs{0} = $ s.d.\ point & $K$                   & $\mathrm{conv}(\mathrm{supp}(\nu))\subset\mathcal{D}$ \\
$\bs{\rho}_*$   & $(3{-}2\sqrt{2}); \; $ s.d. point            & $K_\delta,\,K_\delta^c$   & $\delta$-nbhd of $K$,\, its complement \\
&&$X\setminus Y$          & $\{x\in X\mid x\notin Y\}$ \\

\noalign{\vskip10pt}
\multicolumn{4}{@{}l}{\textbf{Correlator and OPE data}} \\[1pt]
$\widehat{\mathcal{G}}(\bs\rho)$ & unnorm.\ correlator \eqref{cft4ptfunction} & $d\mu(\bs\rho;\bs{h})$ & norm.\ OPE measure at $\bs\rho$ \\
$\mathcal{G}(\bs\rho)$ & $\widehat{\mathcal{G}}/\widehat{\mathcal{G}}(\bs\rho_*)$;\;$\mathcal{G}(\bs\rho_*)=1$ & $\mu(\bs{h})\equiv\mu(\bs\rho_*;\bs{h})$ & base OPE measure at s.d.\ point \\
$\phi_{\mu(\bs\rho_0)}(\bs{t})$ & $\mathcal{G}(\bs\rho_0 e^{i2\pi\bs{t}})/\mathcal{G}(\bs\rho_0)$ & $\omega_{j,k}(\bs{x})$ & $\int h^j\bar{h}^k\,d\mu_n(\bs{x};\bs{h})$ \\
$\bs{h}=(h,\bar{h})$ & eigvals of $H=\tfrac{D+J}{2}$,\;$\bar H=\tfrac{D-J}{2}$ &  $\omega_k\equiv\omega_{k,0}(\bs{0})$ & pure moments at s.d.~point;\;$\omega_1=\Delta/\sqrt{2}$ \\
\noalign{\vskip10pt}
\multicolumn{2}{@{}l}{\textbf{Heavying sequences \& regularizations}} &
\multicolumn{2}{@{}l}{\textbf{Gaussians}} \\[1pt]
$\{(\mathcal{G}_n,\Delta_n)\}$ & heavying seq., $\{\Delta_n\}\nearrow\infty$ & $\Lambda_\sigma(\xi)$ & $\frac{1}{\sqrt{2\pi\sigma^2}}e^{-\xi^2/(2\sigma^2)}$ \\
$\mbf{G}_\Delta$ & crossing-sym~correlators at dim.~$\Delta$ & $\bs\Lambda_\sigma(\bs\xi)$ & $\Lambda_\sigma(\xi)\Lambda_\sigma(\bar\xi)=\frac{1}{2\pi\sigma^2}e^{-\bs\xi\cdot\bs\xi/(2\sigma^2)}$ \\
$(\mathcal{T},\{\Phi_n\})$ & regularization (space, maps) & & \\
\noalign{\vskip10pt}
\multicolumn{2}{@{}l}{\textbf{Rate functions (Reg.~I)}} &
\multicolumn{2}{@{}l}{\textbf{Classical measures (Reg.~II)}} \\[1pt]
$\lambda_n(\bs{x})$ & $\tfrac{1}{\Delta_n}\log\mathcal{G}_n(\bs{x})$ & $\nu_n(\bs\eta)$ & $\mu_n(\Delta_n\bs\eta)$;\;prob.\ measure on $\mathbb{R}^2_+$ \\
$\lambda(\bs{x})$ & rate fn & $\nu$ & classical measure;\;$\mathrm{supp}(\nu)\subset\mathcal{D}$ \\
$\tilde{\lambda}_{\nu}(\bs{x})$&$ \max_{\bs{\eta} \in \mathrm{conv}(\mathrm{supp}(\nu))} \{\bs{x}\cdot \bs{\eta}/\sqrt{2} \}$ & $\tilde{\bs{t}}$ & $\Delta_n\bs{t}$\\
 $\Sigma$ & $\lim_{k\to\infty}\tfrac{1}{\Delta_{n_k}}\log\widehat{\mathcal{G}}_{n_k}(\bs\rho_*)$ & $\phi_\nu(\tilde{\bs{t}})$ & $\int e^{i2\pi\tilde{\bs{t}}\cdot\bs\eta}\,d\nu(\bs\eta)$ \\
 & & $\tilde\omega_{j,k}$ & $\int\eta^j\bar\eta^k\,d\nu(\bs\eta)$ \\

\noalign{\vskip10pt}
\multicolumn{4}{@{}l}{\textbf{Coherent states \& $\alpha$-local rate fns (Reg.~III)}} \\[1pt]
$\hbar\equiv 1/\Delta$ & effective Planck constant & $\tilde{\bs{x}}$ & $\Delta_n^\alpha\bs{x}$; $\alpha \in [0,1]$ \\
$\mu_\Delta(\bs{h})$ & $\nu(\bs{h}/\Delta)$;\;finite-dim.\ pullback & $\lambda_n(\alpha;\tilde{\bs{x}})$ & $\tfrac{1}{\Delta_n^{1-\alpha}}\log\mathcal{G}_n(\tilde{\bs{x}}/\Delta_n^\alpha)$ \\
$\tilde{\mathcal{G}}_\Delta(\bs{\chi})$ & $\int\chi^{h/\sqrt{2}}\bar\chi^{\bar{h}/\sqrt{2}}d\mu_\Delta(\bs{h})$ & $\lambda(\alpha;\tilde{\bs{x}})$ & $\alpha$-local rate fn;\;$\alpha\in[0,1]$ \\

\noalign{\vskip10pt}
\hline
\end{tabular*}
}

\section{Setup} 
\label{setup}

In this section, we review the basic setup of our problem and underline its parallel to the torus partition function in 2d CFT. Since this is mostly review, expert readers may skip this section. Consider a correlation function of four identical scalar operators $\cO$ with scaling dimension $\De$ in a unitary CFT on $\widehat{\mathbb{R}}^d =  \mathbb{R}^d\bigcup\{\infty\} \cong S^d $:
\ba
\langle \cO(x_1) \cO(x_2) \cO(x_3) \cO(x_4)\rangle = \langle \cO(x_1) \cO(x_2) \rangle \langle \cO(x_3) \cO(x_4) \rangle \mathcal{G}(u,v),
\label{cft4ptfunction}
\ea
where $\mathcal{G}(u,v)$ denotes the reduced correlator, which is a function of only cross ratios. Since $u,v$ are symmetric polynomials in $z,\bar{z}$, reduced correlators have a chiral symmetry of $\mathcal{G}(\bs{z}) = \mathcal{G}(\bar{\bs{z}})$. For all $x^\mu_i \in \widehat{\mathbb{R}}^d$, $u,v$ are positive and $\bar{z} = z^*$. 

We can analytically continue $x^0 \to i x^0$ to study the correlator in Lorentzian kinematics. In this signature, $\bs{z}$ is a real vector in $\mathbb{R}^2$, and operator configurations are equipped with a causal ordering. The convergence of OPE decompositions of the correlator in different channels depends on the value of $\bs{z}$. In this work, we focus on the regime $0<z,\bar{z}<1$ called the \textit{causal diamond} (in the language of \cite{Caron-Huot:2017vep}, or $\text{E}_{\text{st}}$ in the work of \cite{Qiao:2020bcs}). We denote this domain as $\mathscr{C}$. In this regime, the s and t OPE channels converge absolutely, and $\mathcal{G}$ is positive and real-analytic for finite $\De$. In the next section we establish this result by way of the radial monomial expansion.
 
\subsection{Radial monomial expansion}
\label{radialmonomialdecomp}
 
We now set up the correlator in the radial conformal frame to obtain a convergent expansion in radial monomials $\rho,\bar{\rho}$ with manifestly positive coefficients~\cite{Hogervorst:2013sma,Hartman:2015lfa,Kravchuk:2020scc}. Radial monomials are obtained as a direct transformation of the independent $z,\bar{z}$ variables, so chiral symmetry extends to $\mathcal{G}(\bs{\rho}) = \mathcal{G}(\bar{\bs{\rho}})$.
 
Working in the radial conformal frame of~\cite{Hogervorst:2013sma}, pairs of operators at $(x_1,x_2)$ and $(x_3,x_4)$ are placed at antipodal points of coplanar concentric circles around the origin of $\widehat{\mathbb{R}}^d$ with radius $r$ and $1$ respectively, and $\rho= \rhob^* = r e^{ i\theta}$. The simultaneous eigenstates $\{|h,\bh\rangle\}$ of the commuting operators $H \equiv \frac{D+ J}{2}$, $\bar{H} \equiv \frac{D - J}{2}$ form a complete orthonormal basis for the CFT Hilbert space denoted $\mathcal{H}_{\mathrm{CFT}}$, with $h,\bh \geq \frac{d-2}{2}$ for all non-identity states by unitarity~\cite{Simmons-Duffin2016-nn}. Here, $D$ denotes the generator of dilatations and $J = M_{12}$ the generator of rotations acting on the plane in which the four operators lie. Importantly, $J$ has a particular orientation and may have negative eigenvalues, unlike the standard quadratic Casimir operator which organizes the conformal block decomposition.
 
By the state-operator correspondence, the pair $\cO(-1)\cO(1)$ applied to the vacuum produces a state $|\psi\rangle$ whose overlaps with the eigenbasis define raw OPE coefficients $\lambda_{h,\bh} = \langle h, \bh | \psi\rangle$. While each $\lambda_{h,\bh}$ is individually finite, the formal state $|\psi\rangle = \sum_{h,\bh} \lambda_{h,\bh} |h,\bh\rangle$ is not normalizable in $\mathcal{H}_{\mathrm{CFT}}$, since $\langle\psi|\psi\rangle$ reduces to a four-point function with coincident operator insertions and therefore diverges. Nevertheless, inserting a complete set of states and performing the required scale transformation and rotation to introduce the correct kinematic dependence, the full correlator may be written as a convergent ``tilted" trace:
\ba
\langle \cO(1) \cO(-1) \cO(-\bs{\rho}) \cO(\bs{\rho})\rangle = \mathrm{Tr}_{\mathcal{H}_{\mathrm{CFT}}}\!\left(\hat{\Pi}_\psi  \rho^{H - \De} \rhob^{\bar{H} - \De}\right) = \sum_{h,\bh} |\lambda_{h,\bh}|^2\, \rho^{h-\De}\rhob^{\bh-\De},
\label{twistedtrace}
\ea
where $\hat{\Pi}_\psi = |\psi\rangle\langle\psi|$.
The non-normalizability of $|\psi\rangle$ corresponds to the divergence of $\mathrm{Tr}(\hat{\Pi}_\psi)$, but the tilted density operator $\hat{\Pi}_\psi \rho^H \rhob^{\bar{H}}$ is trace-class for all $|\rho|,|\rhob| < 1$, since the exponential damping of $\rho^h \rhob^{\bh}$ overwhelms the polynomial growth of $|\lambda_{h,\bh}|^2$ obtained through standard convergence estimates (see appendix \ref{convergence} or the classic paper \cite{Pappadopulo:2012jk}).

Factoring out the disconnected kinematic prefactor, the reduced correlator admits the following $\bs{h} \leftrightarrow\bar{\bs{h}}$ symmetric decomposition into radial monomials:
\ba
\mathcal{G}(\bs{\rho}) = \sum_{\bs{h}} a_{\bs{h}} \bs{\rho}^{\bs{h}} = \sum_{\bs{h}} a_{\bar{\bs{h}}} \bs{\rho}^{\bs{h}},
\label{monomialdecomposition}
\ea
where $|\lambda_{h,\bh}|^2 16^{-\De} = a_{\bs{h}}>0$ are real-positive with $a_{\bs{0}} = 1$, and the second equality follows from chiral symmetry. More geometrically, chiral symmetry arises from the equivalence of radially quantizing around the origin or the point at infinity in $\widehat{\mathbb{R}}^d$, at the cost of flipping $J \to -J$, resulting in the exchange of $H \leftrightarrow \bar{H}$.
 
The associativity of the OPE implies that $\mathcal{G}$ is invariant under permutations of the coordinates $\{x_i\}$, giving rise to non-trivial sum rules (crossing equations) relating different OPE decompositions. Of these, only the s-t crossing equation provides usable constraints within the causal diamond, since the u-channel decomposition does not converge in $\mathscr{C}$ due to the absence of co-dimension 1 spheres separately enclosing the operator pairs at $(x_1,x_3)$ and $(x_2,x_4)$. In this work, we therefore only impose s-t crossing symmetry, given by the following constraint:
\ba
\label{stcrossing}
\mathcal{G}(\bs{\rho}) = \left|\frac{4 \bs{\rho}}{(1-\bs{\rho})^2}\right|^{2\De}\mathcal{G}(\hat{\bs{\rho}}),
\ea
where $\hat{\bs{\rho}} = \left(\frac{1-\sqrt{\bs{\rho}}}{1+\sqrt{\bs{\rho}}} \right)^2$ is the image of $\bs{\rho}$ under $\bs{z} \to 1-\bs{z}$.

\subsubsection{A parallel to torus partition functions in 2d CFT}

The radial monomial decomposition for four-point CFT correlators in general dimensions bears a close resemblance to the $q$-expansion of the partition function of a 2d CFT on a torus. For a torus with modular parameter $\tau$ and $q = e^{2\pi i \tau}$, the partition function may be written as a trace over the CFT Hilbert space~\cite{Simmons-Duffin2016-nn}:
\ba
Z(\tau) = \text{Tr}_{\mathcal{H}_{\mathrm{CFT}}}\left( q^{L_0 - c/24} \bar{q}^{\bar{L}_0 - c/24} \right) = |\bs{q}|^{-c/12}\sum_{\bs{h}} n_{\bs{h}} \bs{q}^{\bs{h}},
\ea
where $L_0, \bar{L}_0$ are the standard generators in the left and right moving Virasoro algebras, and $n_{\bs{h}} = \mathrm{dim}(V_{h,\bh})$ counts the multiplicity of states with left and right moving conformal weights $(h,\bh)$. Multiplying through by a factor of $|\bs{q}|^{c/24}$ yields a decomposition entirely analogous to eq.~(\ref{monomialdecomposition}), with $a_{\bs{h}} \to n_{\bs{h}}$, radial monomials $\to$ $q$-monomials, and $\De \to c/24$.
 
The analogy can be sharpened by comparing the trace representations directly. The reduced torus partition function is the trace $\text{Tr}(q^{L_0} \bar{q}^{\bar{L}_0})$, while the reduced four-point correlator in eq.~(\ref{twistedtrace}) takes the form $\text{Tr}(\hat{\Pi}_\psi\, \rho^{H} \rhob^{\bar{H}})$: the same trace with the insertion of the projector $\hat{\Pi}_\psi = |\psi\rangle \langle \psi |$ onto the state created by the external operators. This projector replaces the state multiplicity $n_{\bs{h}}$ with the squared OPE coefficient $a_{\bs{h}}$, selecting the particular $\cO\times\cO$ fusion channel rather than summing over all states.
 
It is also common to parameterize the torus partition function by left and right moving temperatures $(\beta_L,\beta_R)$, related to the $q$-monomials as $(q,\bar{q}) = (e^{-\beta_L},e^{-\beta_R})$. One can do the same for the radial monomials and re-parameterize $(\rho,\bar{\rho}) = (e^{-\beta},e^{-\bar{\beta}})$. We will use these variables later to compare the mutual information transition observed in heavy correlators to the Hawking-Page transition observed in the torus partition function.

\subsection{The causal diamond and its complex extension}
\label{causaldiamond}
 
The analytic continuation from $\rho= \rhob^* = r e^{\pm i\theta}$ to independent real $0<\rho ,\bar{\rho} < 1$ places the operator configuration on $\widehat{\mathbb{R}}^{1,d-1}$ in a totally spacelike configuration. Applying a global complex dilatation $e^{i \frac{\pi}{2} D}$ to all operators in the correlator, we find that this configuration is conformally equivalent to the causal diamond configuration of~\cite{Ceyhan:2025qrj,Qiao:2020bcs}, where operators lie on the time-like separated corners of nested causal diamonds.\footnote{This is an example of an exact spacelike-timelike correspondence in CFT, analogous to the equivalence between measurements at null infinity and those on a generic null plane~\cite{Mueller:2018llt,Brunello:2025rhh}.}

Upon this analytic continuation, the full correlator picks up a monodromy factor of $e^{2 \pi i \De}$, which is stripped off when we isolate the reduced correlator. Thus, the remaining reduced correlator is independent of the path of analytic continuation. This was generally proven in~\cite{Qiao:2020bcs}, where the author showed that there exist pairings of $\rho,\bar{\rho}$ coordinates which are independent of the path of analytic continuation, and therefore pick up no monodromy factors. These path-independent quantities are as follows:
\ba\label{eq:path-ind}
(\rho \rhob)^\De \quad \forall \De \in \mathbb{R}^+ , \quad \text{and} \quad \left(\frac{\rho}{\rhob} \right)^{k/2} + \left(\frac{\rhob}{\rho} \right)^{k/2} \quad \forall k \in \mathbb{Z}.
\ea
These pairings are manifest in the radial monomial expansion. To see this, we use the $h \leftrightarrow \bh$ symmetry of the expansion to write
\ba
\mathcal{G}(\bs{\rho})= \frac{1}{2}\sum_{\bs{h}} a_{\bs{h}} (\rho \bar{\rho})^{h+\bh} \left( \left(\frac{\rho}{\rhob} \right)^{h-\bh}+\left(\frac{\rho}{\rhob} \right)^{\bh-h}  \right).
\ea
Since $h \mp \bh$ is an eigenvalue of $J$, it is an integer, so each term in the sum is a product of path-independent terms.
 
It is natural to extend $\mathscr{C}$ to a complex domain where the s and t-channel OPE decompositions converge absolutely. Writing $\rho = r e^{i2\pi t}$ and $\bar{\rho} = \bar{r} e^{i2\pi\bar{t}}$ with $(r, \bar{r}) \in [0,1)^2$ and $(t, \bar{t}) \in [0,1)^2$, the extended domain $\widehat{\mathscr{C}}$ is given by $(\rho, \bar{\rho}) \in D^2 \times D^2$, where $D^2$ denotes the open unit disk in the complex plane (see fig.~\ref{fig:extended causal diamond}). This space admits a natural $T^2$ fibration over the base $[0,1)^2$ via the projection $\pi: (\rho, \bar{\rho}) \mapsto (|\rho|, |\bar{\rho}|) = (r,\bar{r})$. The causal diamond corresponds to the real slice $t = \bar{t} = 0$, selecting a single point in each torus fiber. 

While the correlator is well-defined at the origin $(r, \bar{r}) = (0,0)$, where it reduces to the identity contribution $\mathcal{G}(\bs{0}) = 1$, the torus fiber degenerates there as the angular coordinates become ill-defined. The $T^2$ fibration is therefore only well-defined over the punctured base $(0,1)^2$, with fiber degeneration coinciding with the s-channel OPE limit. Fixing a base point $\bs{\rho}_0 = (\rho_0, \bar{\rho}_0) \in \mathscr{C}$, the torus fiber above it is
\ba
T^2_{\bs{\rho}_0} = \left\{ \left(\rho_0 e^{i2\pi t},\, \bar{\rho}_0 e^{i2\pi \bar{t}}\right) : (t, \bar{t}) \in [0,1)^2 \right\} \subset \widehat{\mathscr{C}}.
\ea
 
 \begin{figure}
     \centering
     \includegraphics[width=0.65\linewidth]{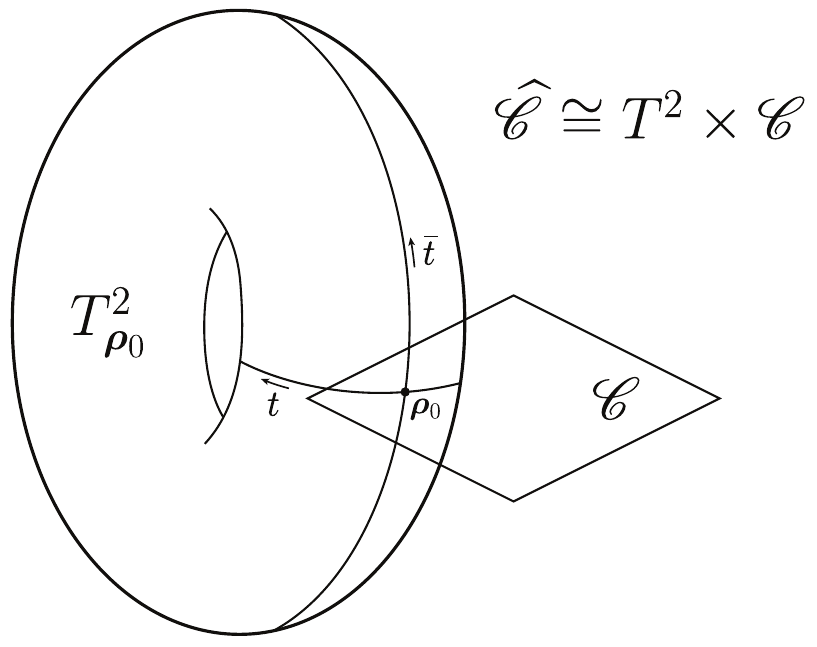}
     \caption{An illustration of the extended causal diamond $\widehat{\mathscr{C}} \cong T^2 \times\mathscr{C}$.}
    \label{fig:extended causal diamond}
 \end{figure}
 
\subsection{OPE measures}
\label{opemeasures}
 
We now review the basic formalism for recasting the OPE decomposition as an integral over a positive definite measure. Consider the radial monomial decomposition in (\ref{monomialdecomposition}) and fix some kinematics in the causal diamond $\boldsymbol{\rho} = (\rho,\bar{\rho})\in(0,1)^2$. We define the normalized OPE measure at point $\bs{\rho} \in \mathscr{C}$ as
\ba
d\mu(\bs{\rho}; \bs{h}) = \frac{1}{\mathcal{G}(\bs{\rho})}\sum_{\bs{h}'} a_{\bs{h}'} \bs{\rho}^{\bs{h}'} \delta^{(2)}(\bs{h} - \bs{h}') d\bs{h}.
\ea
We can now rewrite the radial monomial decomposition as 
\ba
\frac{\mathcal{G}\left(\bs{\rho}\right)}{\mathcal{G}\left(\bs{\rho}_0\right)} = \int_{\mathbb{R}^2_+} \left(\frac{\bs{\rho}}{\bs{\rho}_0}\right)^{\bs{h}} d\mu(\bs{\rho}_0 ; \bs{h}) .
\label{radmonomialdecomp}
\ea
This integral reproduces exactly the original sum and therefore absolutely converges in the unit disk of $\rho,\bar{\rho}$. 

Since the spectrum of $H,\bar{H}$ is discrete, for any compact subset $U \subset \mathbb{R}^2_+$, $\mu(U)$ receives contributions from a finite number of delta masses. In contrast, the pushforward measures
\ba
d\mu(\bs{\rho}; h) = \int_{\bh \in\mathbb{ R_+}} d\mu(\bs{\rho}; \bs{h}), \qquad
d\mu(\bs{\rho};\bh) = \int_{h \in\mathbb{ R_+}} d\mu(\bs{\rho}; \bs{h}),
\ea
have an infinite number of accumulation points. Aside from a single delta mass at $\bs{h} =0$, $\mu$ is supported in $\left[ \frac{d-2}{2},\infty \right)\times \left[\frac{d-2}{2},\infty \right) \subset \mathbb{R}^2_+$ as a result of the unitarity bound.
 
Fixing a base point $\bs{\rho}_0 \in \mathscr{C}$ and studying the correlator along the torus fiber $T^2_{\bs{\rho}_0}$ via $(\rho, \bar{\rho}) = (\rho_0 e^{i2\pi t}, \bar{\rho}_0 e^{i2\pi \bar{t}}) \subset \widehat{\mathscr{C}}$ leads directly to the interpretation of the normalized correlator as a {\it characteristic function} for the OPE measure. If we let $\bs{\rho}_0 \in (0,1)\times(0,1)$ and $\bs{t} \in \mathbb{R}^2$, the normalized correlator can be written as
\ba
\phi_{\mu(\bs{\rho}_0)}(\bs{t}) \equiv \frac{\mathcal{G}(\bs{\rho}_0 e^{i 2 \pi \bs{t}})}{\mathcal{G}(\bs{\rho}_0)} =  \mathcal{F}^{-1}[ \mu'(\bs{\rho}_0;\bs{h})](\bs{t}) \quad \text{with} \quad d\mu(\bs{\rho}_0;\bs{h}) = \mu'(\bs{\rho}_0;\bs{h}) d\bs{h},
\ea
which is the characteristic function corresponding to random variables $\bs{h}$ distributed according to the probability measure $\mu(\bs{\rho}_0;\bs{h})$. Since $\mu$ is a probability measure, its characteristic function is bounded and continuous with $|\phi_{\mu(\bs{\rho}_0)}(\bs{t})| \leq 1$ for all $\bs{t} \in \mathbb{R}^2$. Since the OPE measure is discrete, the OPE density $\mu'(\bs{\rho}_0;\bs{h})$ is a sum of weighted delta functions; in particular, $\mu'$ is a tempered distribution rather than an $L^1$ integrable function. The Fourier inversion formula
\ba
\label{inversionthm}
\mu'(\bs{\rho}_0;\bs{h}) = \mathcal{F}\left[\phi_{\mu(\bs{\rho}_0)} \right](\bs{h})
\ea
therefore holds in the sense of tempered distributions.
 
Since the Fourier transform acts as an automorphism on the space of tempered distributions, a uniqueness theorem for characteristic functions and their probability measures immediately follows~\cite{Chung2001,Durrett2019}. Thus, the OPE measure $\mu(\bs{\rho}_0;\bs{h})$ is uniquely determined by the restriction of the correlator to the torus fiber $T^2_{\bs{\rho}_0}$. For any finite $\De$, we can analytically continue this characteristic function away from real values of $\bs{t}$ to recover the entire normalized correlator, which establishes a bijection between the correlator (normalized to $1$ at $\bs{\rho}_0$) and the probability measure $\mu(\bs{\rho}_0;\bs{h})$.
 
Moreover, this probability measure admits a natural physical interpretation as a branching ratio for the fusion of local operators. Given a four-point configuration at cross-ratios $\bs{\rho}_0$ in the causal diamond, integrating $d\mu(\bs{\rho}_0;\bs{h})$ over a region $U \subset \mathbb{R}^2_+$ computes the probability that the external operators fuse via the s-channel into an intermediate state whose quantum numbers $\bs{h}$ lie within $U$. The dependence of this probability on the choice of base point $\bs{\rho}_0$ reflects the kinematic configuration: as one moves toward the s-channel OPE limit $\rho ,\bar{\rho}\to 0$, the measure localizes onto low-dimension operators dominated by the identity, while approaching the t-channel limit $\rho,\bar{\rho} \to 1$ requires summing over high-dimension states in the s-channel to reproduce the t-channel identity. The analogous limits for the thermal partition function are the low temperature $\beta \to \infty$ and high temperature $\beta \to 0$ limits, respectively. In the low-temperature limit, the partition function localizes around the ground state (the s-channel identity), and in the high-temperature limit, the Boltzmann factor tends to $1$ and the high-energy tails dominate the sum.

The family of OPE measures $\{ d\mu(\bs{\rho};\bs{h}) \}_{\bs{\rho} \in \mathscr{C}}$ forms an exponential family of probability measures fibrating the causal diamond. Fixing the self-dual point $\bs{\rho}_* = (3-2\sqrt{2},3-2\sqrt{2})$ as a reference, the measure at any other base point $\bs{\rho} \in \mathscr{C}$ is obtained from $d\mu(\bs{\rho}_*;\bs{h})$ by an exponential tilt:
\ba
d\mu(\bs{\rho};\bs{h}) = \frac{\mathcal{G}(\bs{\rho}_*)}{\mathcal{G}(\bs{\rho})} \left(\frac{\bs{\rho}}{\bs{\rho}_*}\right)^{\bs{h}} d\mu(\bs{\rho}_*;\bs{h}).
\ea
This identifies $d\mu(\bs{\rho}_*;\bs{h})$ as the base measure of the family and $\bs{h}$ as the sufficient statistic. As $\bs{\rho}$ varies over $\mathscr{C}$, the exponential tilt $(\bs{\rho}/\bs{\rho}_*)^{\bs{h}}$ continuously reweights the base measure, and the OPE measures trace out a smooth statistical manifold over the causal diamond. At each point the torus fiber $T^2_{\bs{\rho}}$ encodes the full measure via its characteristic function. 
 
\textbf{Note on notation:} Invoking $\bs{\rho}_*$ as the canonical reference point for the rest of this work, we simplify notation by writing $\mu(\bs{h}) \equiv \mu(\bs{\rho}_*; \bs{h})$ as the base measure. Additionally, we avoid cumbersome normalization factors by denoting $\mathcal{G}(\bs{\rho}) = \widehat{\mathcal{G}}(\bs{\rho})/\widehat{\mathcal{G}}(\bs{\rho}_*)$ as the correlator normalized to $1$ at the self-dual point, and $\widehat{\mathcal{G}}(\bs{\rho})$ as the ``raw" unnormalized correlator defined in eq.~(\ref{cft4ptfunction}).

\section{Maximally heavy observables}
\label{maxheavy}
 
Sequences of scalar operators with unbounded scaling dimension exist in any CFT, since the spectrum is countably infinite with dimensions accumulating at infinity \cite{Qiao:2017xif,Mukhametzhanov:2018zja}. They also arise naturally across families of CFTs, as in the case of giant gravitons at increasing rank $N$ or operators obtained from the flat-space limit of a rigid QFT on AdS. In both settings, the same question presents itself: \textit{How do we rigorously define and characterize four-point functions of such operators in the infinite scaling dimension limit?}
 
\begin{defn}[Heavying sequence]
\label{heavyingsequence}
Let $\{\cO_n\}_{n \in \mathbb{N}}$ denote a sequence of scalar operators in a CFT or family of CFTs, whose respective scaling dimensions $\{\De_n\}$ form a monotonically non-decreasing and unbounded sequence in $\mathbb{R}^+$. For each $n$, let $\mathcal{G}_n$ denote the reduced four-point function of $\cO_n$ normalized to $1$ at the self-dual point. We call the sequence $\{(\mathcal{G}_n, \De_n)\}_{n \in \mathbb{N}}$ a \textit{heavying sequence}. The set of heavying sequences is closed under taking subsequences: any infinite subsequence of a heavying sequence which maintains the relative ordering of elements is again a heavying sequence.
\end{defn}

\noindent The subtlety of the problem is that, as a correlator heavies, it becomes increasingly divergent and oscillatory on $\widehat{\mathscr{C}}$. This can happen for multiple reasons: additional operators may enter the OPE, or existing OPE coefficients, such as those of stress tensor families, may diverge. However, the most universal reason is that the t-channel identity contribution $(u/v)^\De$ diverges exponentially on the open subset $\{u>v\} \subset \mathscr{C}$ and oscillates with increasing frequency along any torus fiber whose basepoint lies in this subset. Thus, the pointwise limit of a heavying sequence,
\ba
\mathcal{G}_\infty(u,v) = \lim_{n \to \infty} \mathcal{G}_n(u,v),
\ea
gives an uninformative limit on $\mathscr{C}$ and an ill-defined limit on $\widehat{\mathscr{C}}$. In particular, the correlator sequence is not relatively compact in any standard function space on $\mathscr{C}$ or $\widehat{\mathscr{C}}$, so there is no natural topology in which it converges. A heavying sequence should therefore be treated as a sequence of formal objects, and the correct question is not whether the correlator sequence itself converges, but whether one can pair it with a suitable sequence of \textit{good regularizations} which extract well-defined information about the maximally heavy dynamics it encodes. 
 
We now make this precise. Let $\mbf{G}_{\De}$ denote the space of crossing-symmetric reduced four-point functions of identical scalar operators with external dimension $\De$, normalized to $1$ at the self-dual point.
 
\begin{defn}[Regularization]
\label{def:regularization}
A \textit{regularization} associated to an unbounded sequence of dimensions $\{\De_n\} \subset \mathbb{R}^+$ is a pair $(\mathcal{T}, \{\Phi_n\})$, where $\mathcal{T}$ is a topological space and $\{\Phi_n\}$ is a sequence of maps
\ba
\Phi_n :  \mbf{G}_{\De_n} \to \mathcal{T}.
\ea
The regulated sequence associated to a correlator $\mathcal{G}_n \in \mbf{G}_{\De_n}$ is $\{\Phi_n(\mathcal{G}_n)\}_{n \in \mathbb{N}} \subset \mathcal{T}$.
\end{defn}
 
\begin{defn}[Good regularizations and maximally heavy observables]
\label{def:goodreg}
A regularization $(\mathcal{T}, \{\Phi_n\})$ is \textit{good} if $\mathcal{T}$ is metrizable, and for every heavying sequence $\{(\mathcal{G}_n, \De_n)\}$, the regulated sequence $S = \{\Phi_n(\mathcal{G}_n)\}$ is relatively compact in $\mathcal{T}$. Since $\mathcal{T}$ is metrizable, relative compactness of $S$ is equivalent to sequential compactness of its closure $\bar{S}$, which guarantees that every regulated sequence possesses at least one accumulation point in $\mathcal{T}$. We define any such accumulation point as a \textit{maximally heavy observable}.
\end{defn}

\begin{rem}
By a standard topological fact, the set of accumulation points of any sequence in a metrizable space is closed, so the set of maximally heavy observables associated to a given heavying sequence and good regularization is a closed subset of $\mathcal{T}$. Metrizability simultaneously ensures that every accumulation point is the limit of some subsequence and that the limit of any specific convergent subsequence is unique.
\end{rem}

\begin{rem}
The monotonicity in definition \ref{heavyingsequence} is doing the ``heavy lifting" here; since $\{\De_n\}$ is monotonically non-decreasing and unbounded, no infinite subsequence of a heavying sequence can have bounded external dimensions. This rules out the possibility of trivial accumulation points arising from finite-dimensional correlators and justifies the label \textit{maximally heavy}: every accumulation point of a good regularization is associated to the heavy limit.
\end{rem}
We encounter a closely analogous situation in the study of statistical systems in the thermodynamic limit. For a system on a lattice with $N$ sites, extrinsic quantities such as the total internal energy or partition function $Z_N$ diverge as $N\to\infty$, but the free energy density $f = \lim_{N\to\infty} \frac{1}{N}\log Z_N$ remains well-defined. Our goal is to identify and characterize the analogous intrinsic quantities of a heavying correlator sequence, where $\mathcal{G}_n$ is viewed as a partition function, $\De_n$ the system size, and the cross ratios serve as control parameters analogous to temperature. 
 
In this work, we will study three good regularization schemes, each extracting qualitatively different heavy observables from a heavying correlator sequence. These are defined as follows:
 
\begin{defn}[Dynamical free energy density]
\label{def:ratefn}
The \textit{dynamical free energy density} lies in $\mathcal{T} = C(\mathscr{C})$ equipped with the compact-open topology, with the regulated sequence
\ba
 \lambda_n \equiv  \frac{1}{\De_n} \log \mathcal{G}_n : \mathscr{C} \to \mathbb{R}.
\ea
\end{defn}
 
\begin{defn}[Rescaled OPE measure and characteristic function]
\label{def:rescaledmeasure}
The \textit{rescaled OPE measure} lies in $\mathcal{T} = \mathcal{P}(\mathbb{R}^2_+)$ equipped with the weak topology, with the regulated sequence
\ba
\nu_n(\bs{\eta}) \equiv \mu_n(\De_n \bs{\eta}),
\ea
where $\mu_n' = \mathcal{F}^{-1}[\mathcal{G}_n]$ is the OPE density at the self-dual point, obtained by taking the inverse Fourier transform of $\mathcal{G}_n$ along the universal cover of the torus fiber $T^2_{\bs{\rho}_*} \subset \widehat{\mathscr{C}}$. In turn, the sequence of \textit{characteristic functions} $\{\phi_{\nu_n}\}$ of $\{\nu_n\}$ provides an additional regularized sequence in $\mathcal{T} = C(\mathbb{R}^2)$ equipped with compact-open topology, with the regulated sequence 
\ba
\phi_{\nu_n}(\tilde{\bs{t}}) \equiv \mathcal{G}_n(\rho_* e^{i2\pi \tilde{\bs{t}}/\De_n })  = \mathcal{F}[\nu_n' ].
\ea 
Since these regulated sequences are simply related by Fourier transform, one should view them as providing different representations of equivalent data, rather than being entirely distinct regularizations.
\end{defn}
 
\begin{defn}[$\alpha$-local dynamical free energy density]
\label{def:alphalocal}
Let $\alpha \in [0,1]$ and $(\bs{x}) = (\log(\frac{\bs{z}}{1-\bs{z}}))$ be a change of coordinates that maps $\mathscr{C} \to \mathbb{R}^2$ with the self-dual point at the origin. We parameterize the correlator on $\mathscr{C}$ as $\mathcal{G}_n(\bs{x})$. The \textit{$\alpha$-local dynamical free energy density} lies in $\mathcal{T} = C(\mathbb{R}^2)$ equipped with the compact-open topology, with the regulated sequence
\ba
\lambda_n(\alpha; \tilde{\bs{x}}) \equiv \frac{1}{\De_n^{1-\alpha}} \log \mathcal{G}_n(\tilde{\bs{x}}/\De_n^{\alpha}): \mathbb{R}^2 \to \mathbb{R}.
\ea
\end{defn}
 
\begin{rem}
    Each regularization can be classified as postcompositional, precompositional, or hybrid, according to whether it acts on the codomain or domain of $\mathcal{G}_n$. Restricting to the causal diamond and viewing $\mathbf{G}_{\De} \subset \mathrm{Hom}(\mathscr{C}, \mathbb{R}_+)$, the dynamical free energy density is purely postcompositional: it maps the codomain $\mathbb{R}_+ \to \mathbb{R}$ via $\mathcal{G}_n \mapsto \frac{1}{\De_n}\log \mathcal{G}_n$, leaving the domain $\mathscr{C}$ untouched and producing a global picture with no preferred basepoint. The characteristic function regularization is purely precompositional: it reparametrizes the torus fiber $T^2_{\bs{\rho}_*} \subset \widehat{\mathscr{C}}$ by $\mathbb{R}^2$ via $\tilde{\bs{t}} \mapsto \bs{\rho}_* e^{i2\pi\tilde{\bs{t}}/\De_n}$, so that $\phi_{\nu_n}(\tilde{\bs{t}}) = \mathcal{G}_n(\bs{\rho}_* e^{i2\pi\tilde{\bs{t}}/\De_n})$ is the normalized correlator restricted to this fiber. Since any fixed $\tilde{\bs{t}}$ is mapped arbitrarily close to $\bs{\rho}_*$ as $n \to \infty$, this regularization probes asymptotic deviations of size $\sim 1/\De_n$ around the self-dual point, giving a local picture. The rescaled OPE measure $\nu_n$ is related to $\phi_{\nu_n}$ by Fourier transform, which acts as an invertible change of representation rather than as a regularization in its own right.
 
The $\alpha$-local dynamical free energy density is a hybrid of both: it reparametrizes $\mathscr{C} \cong \mathbb{R}^2$ by $\tilde{\bs{x}} \mapsto \bs{x} = \tilde{\bs{x}}/ \De_n^\alpha$, where $\bs{x} = \log(\bs{z}/(1-\bs{z}))$ maps the self-dual point to the origin, and postcomposes by $\frac{1}{\De_n^{1-\alpha}}\log$. This allows it to interpolate between global and local pictures: at $\alpha = 0$ the precomposition is trivial and one recovers the dynamical free energy density, while at $\alpha = 1$ the postcompositional factor $\De_n^{1-\alpha} \to 1$ and the regularization reduces to a purely precompositional one, with the remaining $\log$ serving as an invertible change of representation analogous to the Fourier transform in the characteristic function case.
\end{rem}
 
Most importantly, all three regularizations are good. We state this as a proposition and sketch the key ingredients of each proof:
\begin{prop}[Goodness of the three regularization schemes]
\label{prop:goodnesspreview}
For any heavying sequence $\{(\mathcal{G}_n, \De_n)\}$:
\begin{enumerate}[label=(\roman*)]
\item The dynamical free energy density $\{\lambda_n\}$ is locally uniformly bounded and uniformly equicontinuous on $\mathscr{C} \cong \mathbb{R}^2$. By Arzel\`a--Ascoli (theorem \ref{aathm}) and the gradient bound of lemma \ref{uniformgradbound}, there exists a convergent subsequence whose limit is a 1-Lipschitz continuous function $\lambda(\bs{x}) \in C(\mathbb{R}^2)$ called a \textit{rate function}. (See section~\ref{ratefunctions}.)
 
\item The rescaled OPE measure sequence $\{\nu_n\}$ is tight on $\mathbb{R}^2_+$ (definition \ref{deftightness}). By Prokhorov (theorem \ref{thm:prokhorov}), there exists a weakly convergent subsequence whose limit is a compactly supported probability measure $\nu \in \mathcal{P}(\mathcal{D})$, where $\mathcal{D} = [0,\sqrt{2}]^2$, called a \textit{classical measure}. Crossing symmetry implies that classical measures are invariant under a discrete group isomorphic to the Klein four-group $\mbf{V}_4$. (See section~\ref{classicalmeasures}.)
 
\item The $\alpha$-local dynamical free energy density $\{\lambda_n(\alpha;\tilde{\bs{x}})\}$ is locally uniformly bounded and uniformly equicontinuous on $ \mathbb{R}^2$ for each fixed $\alpha$. By Arzel\`a--Ascoli and the same gradient bound, there exists a convergent subsequence whose limit is a $1$-Lipschitz continuous function $\lambda(\alpha; \tilde{\bs{x}}) \in C(\mathbb{R}^2)$ called an \textit{$\alpha$-local rate function}. (See section~\ref{localratefunctions}.)
\end{enumerate}
Moreover, the characteristic functions of the rescaled OPE measure sequence define another good, local regularization scheme whose accumulation points $\{\phi_\nu\}$ are in bijection with the set of classical measures $\{\nu\}$.
\end{prop}
 
Applying these regularizations to the same heavying correlator sequence may yield a different number of accumulation points for each. Figure~\ref{fig:heavyobservables} illustrates this: while the set of classical measures and their characteristic functions are in bijection, the map between rate functions and classical measures is generically not injective. A heavying correlator sequence may have multiple rate functions but only one classical measure, and vice-versa, so different correlator subsequences are best characterized by different regularization schemes. We demonstrate precisely what this entails through a number of examples in section~\ref{sec:examples}.
 \begin{figure}[t]
    \centering
    \includegraphics[width=0.85\linewidth]{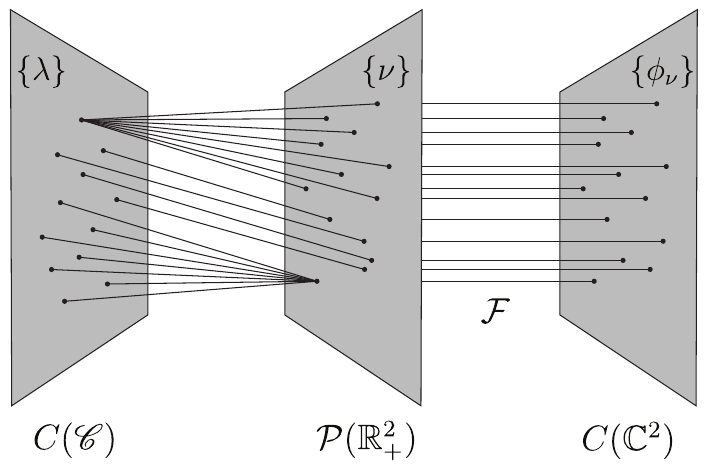}
    \caption{Different regularization schemes access different heavy observables (accumulation points). While the set of classical measures $\{\nu\}$ are in bijection with their characteristic functions $\{\phi_\nu\}$ via Fourier transform $\mathcal{F}$, rate functions $\{\lambda\}$ and classical measures are generally not. That is, a heavying correlator sequence may have multiple rate functions, but only one classical measure, and vice-versa.}
    \label{fig:heavyobservables}
\end{figure}

The remainder of this section is organized as follows. We begin in section~\ref{ratefunctions} by establishing the rate function as a well-posed object, proving existence of convergent subsequences via the Arzel\`a--Ascoli theorem and deriving its universal properties from uniform moment bounds on the sequence of OPE measures. In section~\ref{classicalmeasures}, we establish existence and structural properties of the classical measure, including the $G$-invariance theorem as a result of crossing symmetry and the isotropy subgroup classification of its support. In section~\ref{coherentstatedecomp}, we introduce the coherent state decomposition as a canonical procedure for constructing exact crossing-symmetric solutions from a classical measure, called coherent state correlators, and provide conditions under which these accurately describe the true correlator within regions of size $\sim 1/\sqrt{\De}$ around the self-dual point along the torus fiber $T^2_{\bs{\rho}_*} \subset\widehat{\mathscr{C}}$. We show that the coherent state rate function admits a simple, scale-invariant form uniquely determined by the convex hull of the corresponding classical measure. In section~\ref{localratefunctions}, we introduce the $\alpha$-local dynamical free energy density as a means of connecting the global and local pictures. An important result of this section is a matching theorem that establishes an equivalence between the $\alpha$-local rate function and the coherent state rate function for a particular range of $\alpha$, assuming the satisfaction of certain super-exponential tightness conditions on the rescaled OPE measure sequence. Additionally, we connect select coherent state rate function bounds, assuming a ``classical'' twist gap, to the work of Dey, Pal, and Qiao (DPQ) on torus partition functions at large central charge \cite{Dey:2024nje}.

\subsection{Dynamical free energy densities and rate functions}
\label{ratefunctions}

The first regularization scheme we consider is the dynamical free energy density as given in definition~\ref{def:ratefn}, defined by a sequence of regulations that act non-linearly on the correlator over $\mathscr{C}$ to produce a continuous function $\lambda_n: \mathscr{C} \to \mathbb{R}$. In this case, the target topological space $\mathcal{T}$ is $C(\mathscr{C})$ with a compact-open topology. Via the diffeomorphism $\bs{x} = \log(\bs{z}/(1-\bs{z})): \mathscr{C} \to \mathbb{R}^2$ introduced in definition~\ref{def:alphalocal}, every accumulation point in $\mathcal{T}$ is a $1$-Lipschitz continuous function on $\mathbb{R}^2$, and we use these coordinates throughout when stating gradient bounds. 

Recall that for a heavying sequence $\{(\mathcal{G}_n,\De_n)\}$, the dynamical free energy densities are
\ba
\lambda_n =  \frac{1}{ \Delta_n} \log\left(  \mathcal{G}_n \right).
\ea
This quantity is readily utilized in the study of dynamical phase transitions \cite{Heyl:2017blm}, where $\mathcal{G}_n$ is called the Loschmidt amplitude, $|\mathcal{G}_n|^2$ called the Loschmidt echo, and $\lambda = \lim_{n \to \infty} \lambda_n$ is also known as the rate function. In the dynamical phase transition literature, the Loschmidt amplitude is defined as 
\ba
\mathcal{G}(t)  = \langle \psi | e^{i H t} | \psi \rangle,
\ea
where $H$ is the Hamiltonian of the system and $|\psi\rangle$ is a generic quantum state. More generally, we can view the Loschmidt amplitude as being the expectation value of a family of unitary operators acting on a fixed state in the Hilbert space, parameterized by kinematic variables that control the evolution of the state. 

In an analogy to how the phase transitions of a system can be diagnosed by non-analyticities in their free energy in the thermodynamic limit, \textit{dynamical phase transitions} associated to a maximally heavy correlator can be diagnosed by discontinuities emerging in their dynamical free energy density. In other settings, such as in ergodic theory or large deviation theory, this quantity is known as the \textit{topological pressure} or \textit{rescaled cumulant generating function}, respectively. For simplicity of language and alignment with dynamical phase transition conventions, we will refer to an accumulation point in a sequence of dynamical free energy densities as a \textit{rate function}. More colloquially, this function describes the exponential rate at which the correlator diverges on $\mathscr{C}$.

\subsubsection{Existence and universal bounds}
It is prudent to first address the mathematical subtlety regarding the existence of a rate function for a given sequence of correlators. That is, whether there exists a subsequence $\{\lambda_{n_{k}}\}_{k\in\mathbb{N}} \subset\{\lambda_n\}$ that converges to a rate function $\lim_{k\to \infty} \lambda_{n_k} = \lambda$. Proving this requires the following theorem:
\begin{thm}[Arzel\`a-Ascoli,~\cite{Munkres2000}] 
\label{aathm}
Let $C(X)$ denote the space of continuous functions on a complete metric space $X$. We say that a sequence of functions $\{f_n(x)\}_{n\in\mathbb{N}} \subset C(X)$ is locally uniformly bounded if for any compact subset $A \subset X$, $|f_n(x)| \leq M \; \forall x \in A$ and $n \in\mathbb{N}$, where $M \in\mathbb{R}^+$ is a constant independent of $n,x$. We say a sequence is uniformly equicontinuous if, for every $\epsilon > 0$, there exists a $\delta > 0$ such that $|f_n(x) - f_n(y)| < \epsilon$ whenever $|x-y| < \delta$ for all $n\in\mathbb{N}$, where $\delta$ does not depend on $x,y,$ or $n$. If the sequence $\{f_n\}_{n\in\mathbb{N}}$ is locally uniformly bounded and uniformly equicontinuous, then it has a subsequence $\{f_{n_k}\}_{k\in \mathbb{N}}$ which converges locally uniformly to $f \in C(X)$, i.e.~$\lim_{k\to\infty}f_{n_k} = f $ uniformly on all compact subsets of $X$. 
\end{thm}
Let us canonically parameterize the correlator in terms of the logarithmic coherent state variables $\bs{x} = (x,\bar{x}) = ( \log(\chi),\log(\bar{\chi}))$. Since $\bs{\chi} \in \mathbb{R}^2_+$ in the causal diamond, $\bs{x}$ maps the causal diamond to $\mathbb{R}^2$, which is a complete metric space with respect to the standard Euclidean norm. We first state the following proposition:
\begin{prop}
\label{preaa}
Let $\{\lambda_n(\bs{x})\} \subset C(\mathbb{R}^2)$ denote a sequence of dynamical free energy densities associated to the correlator sequence $\{\mathcal{G}_n(\bs{x})\}$. The sequence $\{\lambda_n(\bs{x})\}$ is locally uniformly bounded and uniformly equicontinuous. 
\end{prop}
\begin{proof}
The proof makes use of the sharp bound on the first moment of the OPE measure for all points in the causal diamond. See appendix \ref{proofofpreaa}  for the complete proof. 
\end{proof}
The Arzel\`a-Ascoli theorem then implies that there exists a subsequence $\{\lambda_{n_k}(\bs{x})\}_{k\in\mathbb{N}}$ and a rate function $\lambda(\bs{x}) \in C(\mathbb{R}^2)$ such that
\ba
\lim_{k\to \infty} \lambda_{n_k}(\bs{x}) = \lambda(\bs{x}) \quad \forall \bs{x} \in A
\ea
uniformly on all compact subsets $A\subset \mathbb{R}^2$. Additionally, we claim that, since the derivative bounds that we use to prove the proposition are sharp, the rate function we define is the unique choice that uniformly regulates the divergences of the maximally heavy correlator on all compact subsets of the causal diamond.

Crossing symmetry also constrains rate functions in a straightforward way. Taking a logarithm of both sides of eq.~(\ref{stcrossing}), dividing by $\De$, and taking the $\De \to \infty$ limit gives the following statement of crossing symmetry for rate functions:
\ba
\lambda(\bs{x}) = (x+ \bar{x}) + \lambda(-\bs{x}) \quad \forall\bs{x} \in \mathbb{R}^2.
\ea
The key intermediate result for proving proposition \ref{preaa} is the sharp gradient bound\footnote{This gradient upper bound can also be written as $ \| \nabla_{\bs{x}}\lambda(\bs{x})\|_{\infty} \leq1$, where $\|\bs{y}\|_{\infty} = \max(y,\bar{y})$ is the $L^\infty$ norm on $\mathbb{R}^2$.}
\ba
\bs{0}\leq \nabla_{\bs{x}}\lambda(\bs{x}) \leq \bs{1}.
\ea
It immediately follows that the rate function is a $1$-Lipschitz continuous function on $\mathbb{R}^2$. Additionally, this gradient bound, along with the universal first moment of all OPE measures, gives rise to the following universal bounds on rate functions:

\begin{prop}[Bounds on rate functions]
\label{universalratefunctionbound}
    Let $\{(\mathcal{G}_n,\De_n)\}$ be a heavying sequence with a subsequence of dynamical free energy densities $\{\lambda_{n_k}\}$ that converges locally uniformly to a rate function $\lambda(\bs{x})=\lim_{k\to\infty}\lambda_{n_k}(\bs{x}) $. Define
    \ba
    f(x) = \frac{\log(\rho(x)/\rho_*)}{\sqrt{2}}=\frac{x-2 \log \left(\sqrt{e^x+1}+1\right)+2 \sinh ^{-1}(1)}{\sqrt{2}}.
    \ea
    For all $\bs{x} \in \mathbb{R}^2$, we have the bound
    \ba
  \lambda_-(\bs{x}) \leq \lambda(\bs{x}) \leq \lambda_+(\bs{x}) ,
    \ea
    where
    \ba
    \lambda_-(\bs{x}) &= \max\left( f(x)+f(\bar{x}), (x+\bar{x}) + f(-x) + f(-\bar{x}), -\Sigma , (x+\bar{x}) - \Sigma \right),\\
    \lambda_+(\bs{x}) &= \max(0,x) + \max(0,\bar{x}),
    \ea
    and $\Sigma$ denotes the free energy density at the self-dual point, given by
    \ba
    \Sigma = \lim_{k \to \infty} \frac{1}{\De_{n_k}} \log\left(\widehat{\mathcal{G}}_{n_k}(\bs{\rho}_*) \right),
    \ea
    where $\widehat{\mathcal{G}}_{n_k}(\bs{\rho}_*)$ is the value of the unnormalized correlator evaluated at the self-dual point. See figure \ref{fig:ratefuncbound} for plots of this bound. 
\end{prop}
\begin{proof}
     The upper bound is derived by integrating the upper gradient bound from $\bs{0}$ to $\bs{x}$ and imposing that the rate function is 0 at $\bs{x} = 0$. The lower bound is derived in two main regions: near the self-dual point and near the OPE limits. Near the self-dual point, the lower bound is given by a crossing symmetrized application of Jensen's inequality. Near the OPE limits, we use that $\widehat{\mathcal{G}}_n(\bs{\rho}) \geq 1$ for all $\bs{\rho} \in \mathscr{C}$ due to the presence of the s-channel identity operator, which is then normalized by the value of the correlator at the self-dual point to obtain the constant shift by $-\Sigma$. For the detailed proof, see appendix \ref{proofofuniversalratefuncbound}.
\end{proof}
\begin{rem}
    If the self-dual free energy density $\Sigma$ is unbounded, that is, the value of $\widehat{\mathcal{G}}_{n_k}(\bs{\rho}_*)$ grows faster than an exponential in $\De_{n_k}$, then the $\max(-\Sigma, (x+\bar{x}) - \Sigma)$ terms will give trivial lower bounds. If $\Sigma = 0$, then the lower bound strengthens to $\max(0,x+\bar{x})$, which entirely fixes the behavior of the rate function along the diagonal limit.
\end{rem}

\begin{figure}[t]
    \centering
    \includegraphics[width=\linewidth]{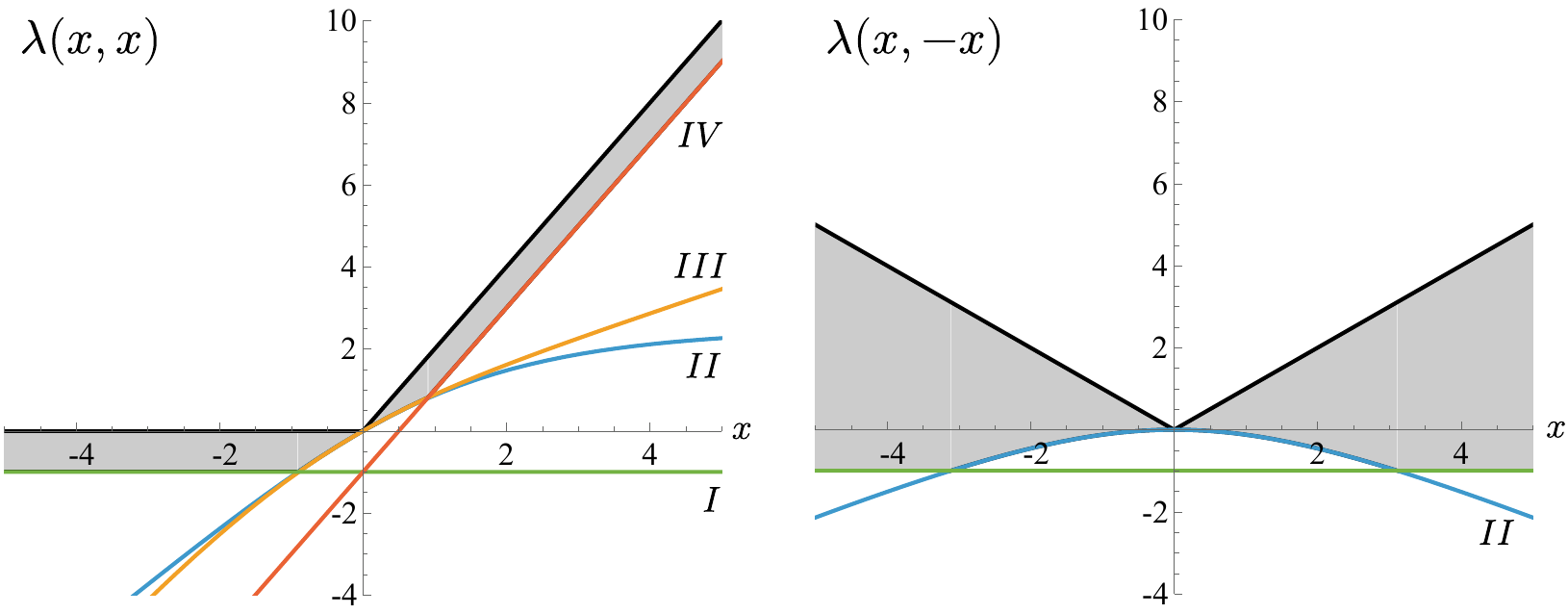}
    \caption{Bounds on rate functions along the diagonal limit (left) and self-dual line (right), as given in proposition \ref{universalratefunctionbound}. The black line denotes the upper bound $\lambda_+$ and the colored lines represent the individual curves whose maximum over $\bs{x} \in\mathbb{R}^2$ gives the total lower bound $\lambda_-$: $(I) = -\Sigma$, $(II) = f(x) + f(\bar{x})$, $(III) = (x+\bar{x}) + f(-x)+f(-\bar{x})$, and $(IV) = (x+\bar{x}) - \Sigma$. The gray area between these curves is the allowed region for a rate function $\lambda$ with self-dual free energy density $\Sigma = 1$.  }
    \label{fig:ratefuncbound}
\end{figure}

\subsubsection{Fluctuations: phase transition or cross-over?}
\label{phasetransitions}

Now that we have addressed the existence and universal properties of the rate function, we can move on to exploring the interesting maximally heavy dynamics that it encodes. The central question is: when does a rate function exhibit a sharp phase transition, and when is it smooth? As we will see, the answer is deeply connected to the degree to which the bulk state sourced by the external operator localizes around a geodesic in AdS, and also motivates the subsequent development of $\alpha$-local rate functions and classical measures.
 
It is instructive to study a simple example to gain some intuition for the global dynamics that the rate function probes. Consider a four-point function of a generalized free field (GFF) with dimension $\De$:
\ba
\mathcal{G}_{\De}(\bs{\chi}) = 1 + \left| \frac{\bs{\chi}}{1+\bs{\chi}}\right|^\De + |\bs{\chi}|^\De.
\ea
Here we are using the coordinates $\bs{\chi} = \left(e^{\bs{x}}\right) = \left(\frac{\bs{z}}{1-\bs{z}}\right)$. We can then compute the rate function in the $\De \to \infty$ limit to be
\ba
\lambda^{\text{(GFF)}}(\bs{x}) &=  \begin{cases}
    x+\bar{x} & \text{if } x+\bar{x} > 0 \\
    0 & \text{if } x+\bar{x} \leq 0
\end{cases}\\&= \max(0,x+\bar{x}).
\ea
This rate function is piecewise linear with a jump discontinuity in its first derivative along the self-dual line $x+ \bar{x} = 0$ ($z = 1-\bar{z}$). The mechanism behind this phase transition is the dominance of identity contributions in different regions of the causal diamond: for small values of $|\bs{\chi}|$, the s-channel identity contribution dominates as the other terms in the correlator exponentially decay, while for $|\bs{\chi}|>1$ the t-channel identity contribution exponentially diverges and dominates the rate function growth.
 
To put this on more physical footing, we point out that these identity terms come from disconnected contributions to the correlator, i.e.~if we write the total correlator as $\langle \cO_1 \cO_2 \cO_3\cO_4\rangle$, then the s and t-channel identity terms come from the disconnected pieces $\langle \cO_1 \cO_2 \rangle \langle \cO_3 \cO_4\rangle$ and $\langle \cO_1 \cO_4 \rangle \langle \cO_2 \cO_3\rangle$ respectively. These factorized two-point structures can be computed via a holographic world-line path integral
\ba
\label{wleffectiveaction}
\langle \cO_i \cO_j \rangle \langle \cO_k \cO_l\rangle = \int \mathcal{D} \mathcal{P}_{ij} \mathcal{D} \mathcal{P}_{kl}
e^{-\frac{\De}{R}\left(  \ell(\mathcal{P}_{ij}) +\ell(\mathcal{P}_{kl})  \right)},
\ea
where $R$ is the AdS radius, $\mathcal{P}_{ij}$ indexes the parametrized curves in AdS$_{d+1}$ which connect points $i,j$ on the boundary, and $\ell(\mathcal{P}_{ij})$ computes their proper length. In the limit $\De \to \infty$ with $R$ fixed, this path integral can be approximated via saddle-point to obtain
\ba
\label{saddlepointapprox}
\langle \cO_i \cO_j \rangle \langle \cO_k \cO_l\rangle \approx e^{-\frac{\De}{R}\left(\ell_{\text{min}}(\mathcal{P}_{ij}) +\ell_{\text{min}}(\mathcal{P}_{kl})\right)},
\ea
where $\ell_{\text{min}}(\mathcal{P}_{ij})$ denotes the minimal geodesic distance between points $i,j$ on the boundary. 

This tells us that, within the different phases of the rate function $\lambda^{\text{GFF}}$, the correlator is effectively described by the pair of minimal-length bulk geodesics connecting pairs of operators on the boundary. In the s-channel phase, operators that lie in the same causal diamond are connected by a geodesic, while in the t-channel phase, operators in different causal diamonds are connected. The phase transition occurs when $\ell_{\text{min}}(\mathcal{P}_{12}) + \ell_{\text{min}}(\mathcal{P}_{34}) \approx \ell_{\text{min}}(\mathcal{P}_{14})+\ell_{\text{min}}(\mathcal{P}_{23})$, which coincides with the self-dual line configuration $\bar{z} = 1-z$. Since each path is weighted by $e^{-\frac{\De}{R}\ell(\mathcal{P})}$, paths whose proper length exceeds the minimal geodesic length by deviations $\delta\ell \gg R/\Delta$ are exponentially suppressed relative to the contribution of $e^{-\frac{\De}{R}\ell_{\min}}$. In other words, the world-line path integral localizes onto paths deviating from the minimal geodesic by proper lengths of order $R/\Delta$. On the CFT side, this is reflected by the fact that, if we rescale $\bs{x} \to \bs{x}/\De$ and compute the resulting rate function, the sharp transition turns into a smooth crossover.
 
Underlying this is a general principle: the distance scale over which a transition becomes smooth is a direct probe of how tightly the path integral localizes around the action-minimizing world-line. More generally, for states whose fluctuations around the minimal geodesic are of order $\sim 1/\De^{\alpha}$ for some $\alpha \in [0,1]$, the rate function will exhibit a sharp transition when viewed at scales larger than $1/\De^{\alpha}$ and a smooth crossover when viewed at finer scales. This is precisely the physics that the $\alpha$-local rate function, introduced later in section \ref{localratefunctions}, is designed to capture: it probes the correlator at asymptotic deviations of size $\sim 1/\De^{\alpha}$ around the self-dual point, and therefore directly resolves the degree of path integral localization parameterized by $\alpha$.
 
We expect the bulk world-line effective action of eq.~(\ref{wleffectiveaction}) to be accurate when applied to correlation functions of heavy operators with dimensions satisfying $\sqrt{c_T} \lesssim \De \lesssim c_T$, since they are heavy enough to form a bound-state which localizes to a bulk world-line, while not being so heavy as to completely fill the non-compact AdS space. In the case of huge operators we would not expect the geodesic approximation to be accurate, as it has been demonstrated previously in \cite{Abajian2023-xw} that these states are dual to extended ``banana''-like geometries rather than a point particle propagating along a world-line. Speculatively, in the case of light operators with dimensions satisfying $\De \lesssim \sqrt{c_T}$, the story is less clear as these states can take the form of both weakly interacting gases, which do not localize in the bulk, or single particle states that do, depending on the microscopic description of the CFT state. In figure \ref{geodesic picture}, we provide a schematic picture for the holographic interpretation of the dynamical phase transition we observe in $\lambda^{(\text{GFF})}$, which generalizes to any correlator whose heavy external states localize onto bulk geodesics.
 
\begin{figure}[t]
    \centering
    \includegraphics[width=\linewidth]{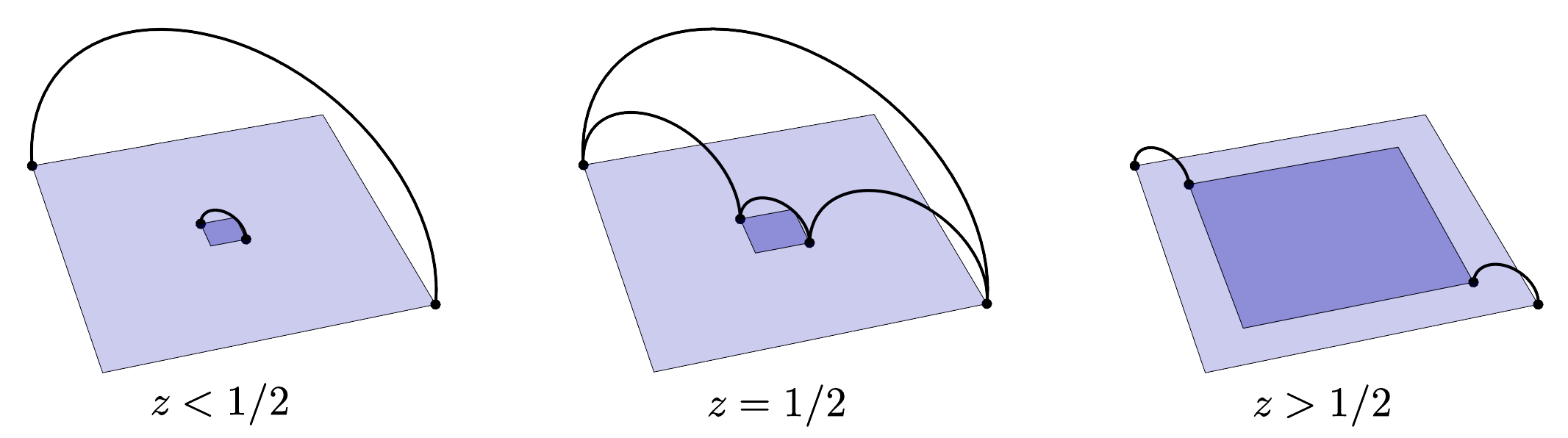}
    \caption{The holographic interpretation of dynamical phase transitions in the Poincar\'e patch (restricted to the diagonal limit of the correlator). Left: for $z<1/2$, geodesic transport between heavy operators in the same causal diamond dominates. Center: for $z = 1/2$, both s and t-channel pairs of geodesics contribute equally, along with other non-universal contributions (not pictured). Right: for $z >1/2$, geodesic transport between pairs of heavy operators across the nested causal diamonds dominates.}
    \label{geodesic picture}
\end{figure}
 
This analysis raises a natural question: \textit{is there a simple way to diagnose the presence of a phase transition at a given point in the causal diamond?} To answer this, we once again take inspiration from statistical systems. In this context, the second derivative of the free energy density is proportional to the specific heat of a system, which diverges at first-order phase transitions. Alternatively, the second derivative of the free energy density computes the fluctuations in energy density, or in our case the fluctuations of the quantum numbers $\bs{h}$ which are distributed according to the OPE measure defined at a given point. When these fluctuations diverge, we know that a first-order phase transition is present. 

We can compute this quantity as\footnote{More generally, one can consider the quantity $g_{ij}(\bs{x}) = \partial_{x_i}\partial_{x_j}\lambda(\bs{x})$ called the ``Fisher information metric,'' which is defined on the statistical manifold of OPE measures assigned to each point in $\mathscr{C}$.}
\ba
\nabla^2 \lambda_n(\bs{x}) = \partial_x^2\lambda_n(\bs{x}) + \partial_{\bar{x}}^2\lambda_n(\bs{x}) ,
\ea
and evaluate the result at whichever point one suspects to be critical in the $n \to \infty$ limit. This quantity is precisely related to the first and second cumulants of the OPE measure, which are related to the OPE moments as
\ba
\label{totalfluctations}
\nabla^2 \lambda_n(\bs{x}) = \frac{1}{\De_n}\left(  \underbrace{\frac{1}{1+e^x}\left(\omega_{2,0}(\bs{x}) - \omega_{1,0}^2(\bs{x}) \right)}_{\text{quadratic}}- \underbrace{\frac{e^{x}}{2(1+e^x)^{3/2}} \omega_{1,0}(\bs{x})}_{\text{linear}} + (x \leftrightarrow \bar{x})\right),
\ea
where we define
\ba
\omega_{j,k}(\bs{x}) = \int  h^j \bar{h}^k d\mu_n(\bs{x};\bs{h})
\ea
as the moments of $\mu_n$ at a given point $\bs{x}$ in the causal diamond. These moments satisfy $\omega_{j,k}(\bs{x}) = \omega_{k,j}(\bar{\bs{x}})$ from the definition of the OPE measure. 

In lemma \ref{sharpmomentbound} in the appendix we further prove that the first moment satisfies $\omega_{1,0}(\bs{x}) = O(\De_n)$ for all points in the causal diamond, so the linear term is always bounded as $n\to\infty$. Thus, the quadratic term, or the variance of the OPE measure at $\bs{x}$, controls whether there is a first-order phase transition at $\bs{x}$. This term is manifestly positive via Jensen's inequality:
\ba
\omega_{2,0}(\bs{x}) \geq \omega_{1,0}^2(\bs{x}),
\ea
and is only bounded in the heavy limit when $\left(\omega_{2,0}(\bs{x}) -\omega_{1,0}^2(\bs{x}) \right) = O(\De_n)$. In terms of the OPE measure, this is the statement that the mass of $\mu(\bs{x};\bs{h})$ is concentrated around its mean within standard deviations of size $O(\sqrt{\De_n})$. If the mass of $\mu(\bs{x};\bs{h})$ is concentrated around its mean with a standard deviation that grows faster than $\sim \sqrt{\De_n}$, then the first derivative of the rate function will exhibit a jump discontinuity, or first-order transition, around that point.
 
We can refine our analysis of these fluctuations at the self-dual point $\bs{x} = 0$. In appendix \ref{convergence}, we prove a rigorous two-sided bound on the second moment at the self-dual point (lemma \ref{sdmomentbounds}), which takes the form
\ba
\frac{\De^2}{2} = \omega_1^2 \leq \omega_2 \leq \frac{1}{8}\left(6\De^2 + 3\De + \sqrt{\De^2(4\De(\De+3)+13)}\right).
\ea
Here $\De$ is the external dimension, $\omega_k = \omega_{k,0}(\bs{0}) =\omega_{0,k}(\bs{0})$ denote the pure moments at the self-dual point, and $\omega_1 = \De/\sqrt{2}$ is fixed by crossing. 

Now consider a heavying sequence of OPE measures $\{\mu_n\}$ at the self-dual point. Assume that the second moments of this measure sequence satisfy
\ba
\left(\omega_2 -\frac{\De_n^2}{2}\right) \sim \De_n^\gamma
\ea
as $n\to\infty$. Then, the quadratic term in eq.~(\ref{totalfluctations}) is of order $\sim\De_n^{\gamma -1}$. If $\gamma > 1$ a phase transition at the self-dual point must occur since the quadratic term diverges as $\De_n \to \infty$, while if $\gamma \leq1$ the quadratic term remains bounded in the heavy limit and we have a crossover around self-duality. A notable example of the latter case occurs when studying the $\langle \phi^L \phi^L \phi^L \phi^L\rangle$ correlator in the ``long'' limit where we send $L\to\infty$ with $\Df$ held fixed. We explore this example in section \ref{GFF example}. The long limit for $\Df =1$ coincides with a correlator of an infinite normal ordered product of $1/2$-BPS displacement operators on the Wilson line defect of $\mathcal{N}=4$ SYM, which was studied at strong coupling in~\cite{ferrero2}, and may be interpreted as a weakly interacting gas that does not localize to a bulk geodesic.
 
Another key fact to point out is that a smooth crossover may only occur when the bound of proposition \ref{universalratefunctionbound} allows it. The self-dual free energy density $\Sigma = \lim_{n\to\infty} \frac{1}{\De_n}\log(\widehat{\mathcal{G}}_n(\bs{\rho}_*))$ controls the rate at which the unnormalized correlator $\widehat{\mathcal{G}}_n$ diverges at the self-dual point. If $\Sigma = 0$, then the lower bound of proposition~\ref{universalratefunctionbound} strengthens to
\ba
\lambda(\bs{x}) \geq \max(0, x + \bar{x})
\ea
on all of $\mathbb{R}^2$. Combined with the upper bound, this entirely fixes the behavior along the diagonal limit to $\lambda(x,x) = \max(0,2x)$. In other words, correlator sequences which are sub-exponentially bounded at the self-dual point are forced to exhibit the universal geodesic phase transition along the diagonal, regardless of the microscopic details of the theory. 

Away from the diagonal limit we may still have non-universal behavior, and we will bound the location of this non-universal region under the assumption of a ``classical'' twist gap in section \ref{selectbounds}. For the case of $\Sigma > 0$, this is equivalent to the statement that the unnormalized correlator at the self-dual point diverges exponentially as $\sim e^{\De \Sigma}$. In other words, the unnormalized OPE measure at the self-dual point must exhibit a ``pile-up'' of non-identity contributions around the mean $h,\bar{h} = \De/\sqrt{2}$, whose total weight diverges exponentially in $\De$ in order for there to be a crossover. However, this does not imply that every heavying correlator sequence that diverges exponentially at the self-dual point exhibits a crossover, as we can still have a transition arising from large fluctuations around the mean of the OPE measure.
 
The subject of dynamical phase transitions around self-duality is a story that will continue to unfold throughout the remaining sections. In the next section, we introduce rescaled OPE measure sequences and their accumulation points, called {\it classical measures}, which serve as a key tool for obtaining a general understanding of the emergence of phase transitions. As we will show, the support of the classical measure controls the phase structure of the rate function in a neighborhood of the self-dual point: when the support has a non-trivial convex hull, a sharp transition generically emerges, while when the support collapses to a single point, the rate function is locally (but not necessarily globally) smooth. 

We will make this mechanism precise by constructing a canonical family of crossing-symmetric correlators from a given classical measure in section~\ref{coherentstatedecomp}, and proving a matching theorem (theorem \ref{matchinglocalratefunctions}) that identifies the precise range of $\alpha$ for which the $\alpha$-local rate function, introduced in section~\ref{localratefunctions}, is uniquely determined by the classical measure. This establishes a sharp equivalence between the smoothness of the rate function at a given scale and how quickly the mass of the rescaled OPE measure concentrates around the support of the classical measure. In the examples of section \ref{sec:examples} we will then realize all three possible outcomes: operators whose rate functions exhibit sharp transitions predicted by classical measures with non-trivial support (localized), operators with smooth crossovers and classical measures that collapse to a single point (delocalized), and operators in an intermediate regime where only the $\alpha$-local rate function at a critical value $\alpha = \alpha_c$ resolves the microscopic structure of the state (quasi-localized).

\subsection{Classical measures}
\label{classicalmeasures}

The dynamical free energy density is a regularization which acts pointwise on correlator sequences, i.e.~it maps a poorly behaved function on $\mathscr{C}$ to a locally uniformly bounded and uniformly continuous function on the same domain. One can also consider regularizations that act on correlators via Fourier transform, allowing one to directly access the underlying OPE data which describes heavy dynamics. If we fix a base-point $\bs{\rho} \in\mathscr{C}$ and compute the Fourier transform of a heavying correlator sequence along the torus fiber $T^2_{\bs{\rho}}$, we compute the sequence of OPE measures $\{\mu_n(\bs{\rho};\bs{h})\} \in \mathcal{P}(\mathbb{R}^2_+)$. The na\"ive hope is that this sequence of measures is relatively compact in $\mathcal{P}(\mathbb{R}^2_+)$, so that there exist accumulation points, and thus well-defined maximally heavy observables. The difficulty lies in showing that this is indeed the case, since the region of OPE measure support (as estimated using OPE convergence) grows as $\sim \De_n$. Intuitively, this means that mass can ``escape to infinity" in the heavy limit, which spoils most notions of convergence in the space of probability measures. 

The remedy is to perform a uniform rescaling of the spectral parameter and study measures as a function of $\bs{\eta} = \bs{h}/\De_n$. This rescales the approximate support of the OPE measure so that it is uniformly $O(1)$ rather than $O(\De_n)$. In precise terms, this is the property of \textit{tightness} of a measure sequence which, via Prokhorov's theorem, implies the sequence is relatively compact in $\mathcal{P}(\mathbb{R}^2_+)$ and contains accumulation points. The remaining free parameter is the base-point in $\mathscr{C}$, which we choose to be the self-dual point. This choice is not required; however, as we show, accumulation points in the space of self-dual rescaled OPE measures possess an exact discrete ``reflection" symmetry and are supported within the compact subspace of $\mathcal{D} = [0,\sqrt{2}]^2 \subset\mathbb{R}^2_+$. Additionally, such a choice of base-point will ultimately allow us to connect OPE measures to rate functions, since the rate function is fixed to be to $0$ at the self-dual point. 

Throughout this section, we will discuss convergence of measures. To begin, let us define the notions of convergence we will be focusing on. 
\begin{defn}[Weak convergence]
Let $(\mu_n)_{n\in\mathbb{N}}$ denote a sequence of probability measures on $X$, and $ C_B(X): X \to \mathbb{R}$ the space of smooth and bounded functions on $X$. We say a measure sequence converges \textit{weakly} to $\mu$ if for all $f \in C_B(X)$
\ba
\lim_{n \to \infty} \int_X f d\mu_n = \int_X f d\mu.
\ea
\end{defn}

\begin{defn}[Tight measure sequence]
\label{deftightness}
Let $\{\mu_n\}_{n \in \mathbb{N}}$ be a sequence of probability measures on $X$. We say this sequence is \textit{tight} if for all $\epsilon>0$, there exists a compact set $K_\epsilon$ and $N\in\mathbb{N}$ such that
\ba
\mu_n(X\setminus K_\epsilon) <\epsilon \quad \forall n \geq N.
\ea
\end{defn}
\noindent Intuitively, this condition tells us that the mass of the sequence $\{\mu_n\}$ does not ``escape to infinity," and has uniform tail bounds that confines the mass to a ``central" region of $X$. Given a tight sequence of measures, Prokhorov's theorem states the following:

\begin{thm}[Prokhorov's theorem,~\cite{Billingsley1999,Prokhorov1956}]
\label{thm:prokhorov}
Let $(X,d)$ be a complete separable metric space, and let $\mathcal{P}(X)$ denote the space of Borel probability measures on $X$ equipped with the weak topology. A family $\Pi \subset \mathcal{P}(X)$ is relatively compact in $\mathcal{P}(X)$ if and only if it is tight.
\end{thm}
\begin{cor}[Existence of a weakly convergent subsequence]
\label{cor:prokhorov}
If $\{\mu_n\}_{n \in \mathbb{N}} \subset \mathcal{P}(X)$ is a tight sequence of probability measures, then there exists a subsequence $\{\mu_{n_k}\}_{k \in \mathbb{N}}$ and a probability measure $\mu \in \mathcal{P}(X)$ such that $\mu_{n_k} \to \mu$ weakly as $k \to \infty$.
\end{cor}

\noindent Here, $X = \mathbb{R}^2_+$ is a closed subset of $\mathbb{R}^2$ and hence a complete separable metric space, so Prokhorov's theorem applies directly. Corollary~\ref{cor:prokhorov} is the key tool we will use in proving the existence of accumulation points.

It is often convenient to prove weak convergence of probability measures from the side of their respective characteristic functions, for which we have the following theorem:
\begin{thm}[L\'evy's continuity theorem,~\cite{Williams1991}]
Suppose we have a sequence of $k$-dimensional random variables $\{X_n\}_{n\in\mathbb{N}}$ distributed according to the probability measures $\{\mu_n\}$, and a corresponding sequence of characteristic functions $\phi_{\mu_n}(t) = \mathbb{E}[e^{i tX_n}]$. If the sequence of characteristic functions converges pointwise to some function $\phi_{\mu}(t)$ for all $t \in \mathbb{R}^k$, then the following statements are equivalent:
\begin{center}
\begin{tabular}{rl}
    i. & $X_n$ converges weakly to some random variable $X$ ($\mu_n \to\mu$ weakly).\\[0.5em]
    ii. & $\{X_n\}_{n\in\mathbb{N}}$ is tight. \\[0.5em]
    iii. & $\phi_{\mu}(t)$ is a characteristic function of some random variable $X$.\\[0.5em]
    iv. & $\phi_{\mu}(t)$ is a continuous function of $t$.\\[0.5em]
    v. & $\phi_{\mu}(t)$ is continuous at $t = 0$.
\end{tabular}
\end{center}
\label{levycont}
\end{thm}

\subsubsection{Existence, compactness, and Klein four-group symmetry}

Let $\{(\mathcal{G}_n,\De_n)\}$ denote a heavying sequence. For each correlator, we can take the inverse Fourier transform along the non-compact universal cover of $T^2_{\bs{\rho}_*}$ with $\bs{\rho}_* = (3-2 \sqrt{2},3-2\sqrt{2})$ to produce a sequence of OPE measures at the self-dual point denoted $\{\mu_n(\bs{h})\}$. We introduce a rescaled measure sequence $\{\nu_n(\bs{\eta})\}$ where $\nu_n(\bs{\eta}) = \mu_n(\De_n\bs{\eta})$. Since $\De_n>0$, these measures share the same support with $\bs{\eta} \in \mathbb{R}^2_+$. To prove that the regularization of $\{(\mathcal{G}_n,\De_n)\} \to \{\nu_n\}$ is good, one needs to show that there exist accumulation points in any heavying sequence of rescaled OPE measures. To do this, we will use Prokhorov's theorem, with the key intermediate result being the following proposition:

\begin{prop}
\label{tightnessofmeasuresequence}
$\{\nu_n\}_{n \in \mathbb{N}}$ is a tight sequence of probability measures on $\bs{\eta} \in \mathbb{R}^2_+$. 
\end{prop}
\begin{proof}
    See appendix \ref{prop:tight} for the complete proof using uniform moment bounds on $\{\nu_n\}$ and Markov's inequality.
\end{proof}
\noindent As a consequence, corollary \ref{cor:prokhorov} of Prokhorov's theorem applies and there exists a subsequence $\{\nu_{n_k}\}_{k\in\mathbb{N}}$ which weakly converges to a \textit{classical} measure $\nu$. Since such a convergent subsequence always exists, we will condense notation and denote $\{\nu_{n_k}\}_{k\in\mathbb{N}} \to \{\nu_k\}_{k\in\mathbb{N}}$ as one of these convergent sequences and assume there is only one accumulation point. 

We can now consider the more non-trivial properties of a classical measure. Let $\mathcal{D} = [0,\sqrt{2}] \times [0,\sqrt{2}]$ denote the square in $\mathbb{R}_+^2$ parametrized by $\bs{\eta} = (\eta,\bar{\eta})$. We define transformations $\mathcal{R},\mathcal{C}:\mathcal{D} \to \mathcal{D}$ by
\ba
\mathcal{R}\circ\bs{\eta} = (\sqrt{2} -\bar{\eta}, \sqrt{2} - \eta) \quad\text{and}\quad \mathcal{C}\circ\bs{\eta} = \bar{\bs{\eta}}.
\ea
Geometrically, $\mathcal{R}$ is a reflection through the anti-diagonal $\{\eta + \bar{\eta} = \sqrt{2}\}$, while $\mathcal{C}$ is a reflection through the diagonal $\{\eta = \bar{\eta} \}$. These generate a group $G =\langle \mathcal{R},\mathcal{C}\rangle$ isomorphic to the Klein four-group $\mathbf{V}_4$, with presentation
\ba
G = \langle \mathcal{R},\mathcal{C}| \mathcal{R}^2 = \mathcal{C}^2 = (\mathcal{R}\mathcal{C})^2 = 1 \rangle.
\ea
The group $G$ acts on the space $\mathcal{P}(\mathcal{D})$ of probability measures on $\mathcal{D}$ by pushforward: for $g \in G$ and $\nu \in \mathcal{P}(\mathcal{D})$, define 
\ba
g_\sharp\nu(A) = \nu( g^{-1} \circ A) \quad \forall A \subset \mathcal{D}.
\ea
A measure $\nu$ is $G$-invariant if $g_\sharp\nu = \nu$ for all $g \in G$. We denote the space of $G$-invariant probability measures by $\mathcal{P}(\mathcal{D})^G$. With this notation established, we can state the main structural result as follows:
\begin{thm} (Properties of classical measures)
\label{thm:maxheavy}
Let $\{\nu_n\}_{n \in \mathbb{N}}$ denote a heavying sequence of rescaled OPE measures which weakly converges to the classical measure $\nu$. Then $\nu \in \mathcal{P}(\mathcal{D})^G$. That is:
\begin{center}
\begin{tabular}{rl}
    i. & $\mathrm{supp}(\nu) \subset \mathcal{D}$, \\[0.5em]
    ii. & $\mathcal{R}_\sharp \nu = \nu$ \quad (Reflection symmetry), \\[0.5em]
    iii. & $\mathcal{C}_\sharp \nu = \nu$ \quad (Chiral symmetry).
\end{tabular}
\end{center}
\end{thm}
\begin{proof}
We prove $(i)$ in proposition \ref{prop:compactsupport}, and $(iii)$ is implied by the chiral symmetry of the correlator normalized at the self-dual point. For $(ii)$, we first use the crossing equation to establish $(\mathcal{RC})_\sharp \nu = \nu$, and then compose with $(iii)$ to obtain $\mathcal{R}_\sharp \nu = \mathcal{C}_\sharp (\mathcal{RC})_\sharp \nu = \nu$. Using lemma \ref{qtermasym} and \ref{phitermasym} to approximate the monomial in the crossed channel at $\rho_0 = \rho_* = 3-2\sqrt{2}$, setting $\bs{\rho} = (\rho_* e^{i 2\pi  \tilde{\bs{t}}/\De_n } ) $, and writing the crossing equation in terms of rescaled OPE measures, we have
\ba
\phi_{\nu_n}(\tilde{\bs{t}}) = \int e^{i2\pi \bs{\eta}\cdot \tilde{\bs{t}}} d\nu_n(\bs{\eta}) &= \int e^{i2\pi (\mathcal{C}\circ \mathcal{R} \circ \bs{\eta})\cdot \tilde{\bs{t}}}\left(1 + O\left(\frac{\tilde{t}^2 \eta}{\De_n} \right) +O\left(\frac{\tilde{\bar{t}}^2 \bar{\eta}}{\De_n} \right)  \right)d\nu_n(\bs{\eta}) \\
&= \phi_{(\mathcal{RC})_{\sharp} \nu_n }(\tilde{\bs{t}})  +O(1/\De_n),
\ea
where we take $\{\nu_n\}$ to be a subsequence that converges weakly to $\nu$ as $n\to\infty$, and we have used the uniform moment bounds of lemma \ref{sdmomentbounds} to bound the error term. Since the error term decays for each fixed $\tilde{\bs{t}} \in \mathbb{R}^2$ as $n \to \infty$, $\phi_{\nu_n}$ and $\phi_{(\mathcal{RC})_{\sharp}\nu_n}$ both converge pointwise to $\phi_{\nu}$. Theorem \ref{levycont} then implies weak convergence of $\{ (\mathcal{RC})_{\sharp}\nu_{n} \} \to \nu $. Since $\nu$ and $(\mathcal{RC})_\sharp\nu$ are both supported on the compact set $\mathcal{D}$, equality against $C_b(\mathcal{D}) = C(\mathcal{D})$ implies equality as Borel measures by the Riesz representation theorem. Composing with $(iii)$ gives $\mathcal{R}_\sharp\nu = \mathcal{C}_\sharp(\mathcal{RC})_\sharp\nu = \nu$.
\end{proof}
We can systematically construct an element of $\mathcal{P}(\mathcal{D})^G$ from a probability measure on a fundamental domain. A fundamental domain for the $G$-action is
\ba
\mathbf{F} = \{ (\eta,\bar{\eta}) \in \mathcal{D} \mid \eta \leq \bar{\eta}, \; \eta + \bar{\eta} \leq \sqrt{2}\},
\ea
which is a closed triangle with vertices at $(0,0)$, $(0,\sqrt{2})$, and $(\sqrt{2}/2,\sqrt{2}/2)$. Any element of $\mathcal{P}(\mathcal{D})^G$ is uniquely determined by its restriction to $\mathbf{F} $. Conversely, given a probability measure $\sigma\in \mathcal{P}(\mathbf{F})$, the symmetrization
\ba
\nu =\frac{1}{|G|} \sum_{g \in G} g_\sharp \sigma
\ea
uniquely defines an element of $\mathcal{P}(\mathcal{D})^G$, where $|G| = 4$.

\begin{cor}
\label{cor:momentbounds}
The moments $\tilde{\omega}_{j,k} = \int \eta^j \bar{\eta}^k d\nu$ of any classical measure $\nu \in \mathcal{P}(\mathcal{D})^G$ satisfy the sharp upper bound
\ba
\tilde{\omega}_{j,k} \leq 2^{(j+k-2)/2} \quad \forall \; j+k \geq 1,
\ea
saturated by the measure produced from the $G$-orbit of $(0,0) \in \mbf{F}$, $\nu^{(+)} = \frac{1}{2}\left(\delta_{(0,0)} + \delta_{(\sqrt{2},\sqrt{2})}\right)$. For pure moments, the lower bound $\tilde{\omega}_{k,0} =\tilde{\omega}_{0,k}\geq 2^{-k/2}$ is saturated by the measure produced from the $G$-orbit of $(1/\sqrt{2},1/\sqrt{2}) \in \mbf{F}$, $\nu^{(0)} = \delta_{(1/\sqrt{2},1/\sqrt{2})}$. Combining these gives
\ba
2^{-k/2} \leq \tilde{\omega}_{k,0} \leq 2^{(k-2)/2} \quad \forall \; k \geq 1,
\ea
which recovers the analogous two-sided bounds for $\eta,\bar{\eta}$ to those derived for scaling moments in~\cite{Poland:2024hvb}. For mixed moments with $j,k \geq 1$, the lower bound is $\tilde{\omega}_{j,k} \geq 0$, saturated by the measure produced from the $G$-orbit of $(0,\sqrt{2})\in\mbf{F}$, $\nu^{(-)} = \frac{1}{2}\left(\delta_{(\sqrt{2},0)} + \delta_{(0,\sqrt{2})}\right)$.
\end{cor}
\begin{proof} Since $\nu$ is a $G$-orbit average of its restriction $\sigma \in \mathcal{P}(\mathbf{F})$, the moment $\tilde{\omega}_{j,k}$ equals the $G$-orbit average of $\eta^j\bar{\eta}^k$ integrated against $\sigma$. For any measurable function $f: \mbf{F}\to\mathbb{R}$, $\min_{\bs{\eta} }f(\bs{\eta}) \leq \int f \,d\sigma \leq\max_{\bs{\eta} }f(\bs{\eta})$, so $\tilde{\omega}_{j,k}$ is extremized when $\sigma = \delta_{\bs{\eta}}$ for some $\bs{\eta} \in \mathbf{F}$. The rearrangement inequality implies $a^jb^k + b^ja^k \leq a^{j+k} + b^{j+k}$, and applying this to both chiral pairs in the orbit average reduces the problem to optimizing $g_n(a) + g_n(b)$, where $g_n(x) = \frac{1}{2}[x^n + (\sqrt{2}-x)^n]$ and $n = j+k$. The function $g_n$ is convex on $[0,\sqrt{2}]$ with its maximum $2^{(n-2)/2}$ at the endpoints. The upper bound follows from $a = b = 0$. Setting $k=0$, the pure moments have their minimum at $a = b = \sqrt{2}/2$. For mixed moments, $\tilde{\omega}_{j,k} \geq 0$ is immediate from non-negativity of the integrand.
\end{proof}

\subsubsection{Isotropy subgroup classification}

The action of $G$ partitions $\mathcal{D}$ into orbits. A generic point has a four-element orbit, but $G$ has nontrivial isotropies along different subsets of $\mathcal{D}$. We say a subset $S \subset \mathcal{D}$ is $H$-isotropic for a subgroup $H \subset G$ if $S \subset \mathcal{D}^H$, where $\mathcal{D}^H = \{\bs{\eta} \in \mathcal{D} \mid h \circ \bs{\eta} = \bs{\eta} \text{ for all } h \in H\}$ is the fixed-point set of $H$. Elements of $\mathcal{P}(\mathcal{D})^G$ admit a classification by isotropy type. A measure $\nu \in \mathcal{P}(\mathcal{D})^G$ is of Type $k$ according to the isotropy of its support, as enumerated in table~\ref{tab:orbits}:
\begin{table}[h]
\centering
\begin{tabular}{c|c|c|c}
Type & Support & Isotropy & Orbit size \\
\hline
0 & $\{(\sqrt{2}/2, \sqrt{2}/2)\} = \mathcal{D}^G$ & $G$ & 1 \\
1 & $\{\eta = \bar{\eta}\} = \mathcal{D}^{\langle\mathcal{C} \rangle}$ & $\langle \mathcal{C} \rangle$ & 2 \\
2 & $\{\eta+ \bar{\eta} = \sqrt{2}\} =\mathcal{D}^{\langle\mathcal{R} \rangle} $ & $\langle \mathcal{R} \rangle$ & 2 \\
3 & $\mathcal{D}$ & $\langle 1\rangle$ & 4 \\
\end{tabular}
\caption{Classification of $G$-orbits in $\mathcal{D} = [0,\sqrt{2}]^2$ by isotropy subgroup.}
\label{tab:orbits}
\end{table}
These types form a hierarchy under inclusion. The fixed-point sets satisfy $\mathcal{D}^G \subset \mathcal{D}^{\langle \mathcal{C} \rangle}\subset \mathcal{D}$ and $\mathcal{D}^G \subset \mathcal{D}^{\langle \mathcal{R} \rangle} \subset \mathcal{D}$, so that a measure of Type 0 is also of Types 1, 2, and 3, while a measure of Type 1 or 2 is also of Type 3. Additionally, a measure that is both Type 1 and 2 is of Type 0, since $\mathcal{D}^G = \mathcal{D}^{\langle\mathcal{C}\rangle} \cap\mathcal{D}^{\langle\mathcal{\mathcal{R}}\rangle}$. We say a measure is \emph{purely} of Type $k$ if its support is $H$-isotropic for the corresponding subgroup but not for any larger subgroup. In section~\ref{phasetype}, we discuss how this isotropy classification extends to a classification of the different global phases exhibited by rate functions, along with their physical interpretations.

\subsubsection{Coarse-graining ambiguity and its resolution}

The classical measure is a well-defined maximally heavy observable. Moreso, it is defined by a canonical choice of rescaling. If we instead defined a different scaled measure sequence, for example $\mu_n(\De_n^\gamma\bs{\eta}) = \nu_n^{(\gamma)}(\bs{\eta})$, then for $\gamma<1$, the regularization $\{(\mathcal{G}_n,\De_n)\} \to \{\nu_n^{(\gamma)}\}$ would not be good, since the resulting measure sequence would not be tight and therefore would not necessarily have a convergent subsequence. On the other hand, if $\gamma > 1$, then the sequence would have a unique, but trivial, limiting measure given by a single delta mass at the origin. Our choice of $\gamma = 1$ ensures that the regularization is good, while not making it trivial in that it is forced to have a single trivial accumulation point. 

The act of rescaling the measure sequence in this way introduces a kind of ``coarse-graining" ambiguity which obscures the fine structure of the true OPE measure in the heavy limit. In this subsection, we will address this coarse-graining ambiguity by studying the characteristic function of the classical measure $\nu$, and provide a condition on measure sequences for when the domain of agreement between the characteristic function of $\nu$ and the correlator can be extended.

The inverse Fourier transform of a classical measure computes its unique characteristic function
\ba
\phi_\nu(\tilde{\bs{t}}) = \int  e^{i 2 \pi \tilde{\bs{t}}\cdot \bs{\eta}}d\nu(\bs{\eta}).
\ea
By the Paley-Wiener theorem, the compactness of the support of the measure ensures that the characteristic function is an entire function on $\mathbb{C}^2$. This characteristic function can also be obtained by taking the limit of the corresponding correlator subsequence under the change of variables:
\ba
\bs{\rho} = (\bs{\rho}_* e^{i2\pi \tilde{\bs{t}}/\De_n})=( \rho_* e^{i 2 \pi \tilde{t}/\Delta_n}, \rho_* e^{i 2 \pi \tilde{\bar{t}}/\Delta_n}),
\ea
where $\rho_* = 3-2\sqrt{2}$ is the self-dual point. 

Plugging this into the radial monomial expansion for the correlator and noting $\bs{h} = \De_n \bs{\eta}$, it is straightforward to see that
\ba
\label{pointwiseconvergence1}
\lim_{n \to \infty} \mathcal{G}_n( \rho_* e^{i 2\pi \tilde{\bs{t}} /\Delta_n}) = \lim_{n \to \infty} \int  e^{i 2 \pi \tilde{\bs{t}} \cdot \bs{\eta} } d\mu(\Delta_n \bs{\eta})=\lim_{n \to \infty} \int  e^{i 2 \pi \tilde{\bs{t}} \cdot \bs{\eta} } d\nu_n(\bs{\eta}) = \phi_\nu(\tilde{\bs{t}}).
\ea
This is pointwise convergent for each fixed $\tilde{\bs{t}}$: the correlator evaluated at a position-space deviation $\sim \tilde{\bs{t}}/\De_n$ from the self-dual point converges to $\phi_\nu(\tilde{\bs{t}})$. Moreover, since the rescaled measure sequence $\{\nu_n\}$ is tight, the sequence of characteristic functions $\{\phi_{\nu_n}\}$ is uniformly bounded and equicontinuous for $\tilde{\bs{t}}\in\mathbb{R}^2$.

A natural question is whether this pointwise convergence can be improved to allow $\tilde{\bs{t}}$ to grow with some power of $\De_n$, which would extend the domain of agreement to a larger neighborhood around the self-dual point. In equations, this is the question of whether
\ba
\lim_{n \to \infty} &\left(  \mathcal{G}_n(\rho_0 e^{i 2\pi \tilde{\bs{t}} / \De_n^{1-\beta} } )   - \phi_{\nu}(\De_n^{\beta}\tilde{\bs{t}})\right) \\ 
&=\lim_{n \to \infty} \left(  \phi_{\nu_n}(\De_n^{\beta}\tilde{ \bs{t}})  - \phi_{\nu}(\De_n^{\beta}\tilde{\bs{t}})\right)\stackrel{?}{=} 0 \quad \forall \tilde{\bs{t}} \in \mathbb{R}^2 \quad \text{with} \quad \beta \in \mathbb{R}.
\ea
For $\beta \leq0$, this holds trivially since we are probing position-space deviations that vanish faster than $\sim 1/\De_n$. For $\beta >0$, we are asking whether convergence extends to arguments $\tilde{\bs{t}}$ growing like $\De_n^{\beta}$ (or $\bs{t} \sim 1/\De_n^{1-\beta}$), which would establish agreement over position-space deviations of size $\sim 1/\De_n^{1-\beta} \gg 1/\De_n$ from the self-dual point. This requires additional assumptions on the measure sequence.

We can make this precise by introducing the finite-dimension pull-back of the classical measure:
\ba
\mu_{\De}(\bs{h}) \equiv \nu(\bs{h}/\De).
\ea
This pull-back measure satisfies many of the same properties as the classical measure, but now rescaled by $\Delta$. Namely, we have compact support $\mathrm{supp}(\mu_\Delta) \subset [0,\sqrt{2} \Delta]^2 $, reflection symmetry $\mu_\Delta(\bs{h}) = \mu_\Delta(\mathcal{R}\circ\bs{h}) = \mu_\Delta(\sqrt{2}\Delta - \bar{h},\sqrt{2}\De - h)$, and chiral symmetry $\mu_\Delta(\bs{h}) = \mu_\Delta(\bar{\bs{h}})$. Its characteristic function satisfies $\phi_{\mu_\De}(\tilde{\bs{t}}) = \phi_\nu(\De \tilde{\bs{t}})$, which for $\De = \De_n$ is precisely the quantity we compare to $\phi_{\mu_n}(\tilde{\bs{t}}) = \phi_{\nu_n}(\De_n \tilde{\bs{t}})$. 

Before establishing and proving our desired result, we state the convolution theorem for reference:
\begin{thm}[Convolution Theorem] 
\label{convthm}
Let $f(\bs{t}),g(\bs{t})$ be functions on $\mathbb{R}^n$, then
\ba
\mathcal{F}[f * g](\bs{\omega}) = \mathcal{F}[f](\bs{\omega}) \cdot \mathcal{F}[g](\bs{\omega}) \quad \text{and} \quad f(\bs{t})\cdot g(\bs{t}) = \mathcal{F}^{-1} [ \mathcal{F}[f] * \mathcal{F}[g]](\bs{t}).
\ea
\end{thm}
\begin{proof}
    This is a standard result in Fourier analysis, see e.g.~\cite{McGillemCooper1984}.
\end{proof}

\noindent As a point of notation, we denote $d\mu(\bs{h}) = \mu'(\bs{h})d\bs{h}$ for both the exact and finite-dimension pull-back measures. It is now a simple exercise in applying the above theorem to prove the following proposition:
\begin{prop}
\label{convolutioncond}
Let $\{\mu_n\}$ denote a heavying sequence of OPE measures, with their respective rescaled measure sequence $\{\nu_n\}$ weakly converging to a classical measure $\nu$. Let $\{\mu_{\De_n}\}$ denote the corresponding sequence of finite-dimension pull-back measures of $\nu$. If 
\ba
\lim_{n\to \infty}\| \bs{\Lambda}_{\De_n^{1-\alpha}} *(\mu_n' - \mu_{\De_n}')\|_{L^1(\mathbb{R}^2)} = 0 \quad \text{(Convolution condition)}
\ea
for some $\alpha>0$, then $\forall \tilde{\bs{t}} \in \mathbb{R}^2$ and $\beta < \alpha$
\ba
\label{pointwiseasymptoticconvergence}
\lim_{n \to \infty} \left( \phi_{\nu_n}(\De_n^{\beta} \tilde{\bs{t}}) -  \phi_{\nu}(\De_n^{\beta} \tilde{\bs{t}}) \right)  = 0.
\ea
\end{prop}

\begin{proof}
Since $f(\bs{h})= \bs{\Lambda}_{\De_n^{1-\alpha}} *(\mu_n' - \mu_{\De_n}')$ is the difference between two probability densities convolved with an $L^1$ integrable function (a Gaussian), and $L^1$ integrable functions are closed under addition, $f(\bs{h})$ is $L^1$ integrable. We then use the general fact that $\| f\|_{L^1(\mathbb{R}^2)} \geq | \mathcal{F}^{\pm1}[f]|$ for any $L^1$ integrable function $f$ on $\mathbb{R}^2$, and the rest is the convolution theorem. See appendix \ref{convolutioncondproof} for the full proof.
\end{proof}

When the convolution condition holds for some $\alpha > 0$, the parameter $\beta$ can be taken into the non-trivial range $0 < \beta <\alpha$. This means asymptotic convergence extends to $|\bs{t}| \sim \De_n^{\beta}$, establishing agreement over a position-space neighborhood of size $\sim 1/\De_n^{1-\beta} \gg 1/\De_n$ around self-duality; a strictly larger domain than what is generically guaranteed by pointwise convergence alone. Proposition \ref{convolutioncond} plays a key role in the main theorem of the next section.

\subsection{Coherent state decomposition}
\label{coherentstatedecomp}

In the previous subsection, we introduced the notion of a finite-dimensional pull-back measure $\mu_{\De}$, which is constructed such that its characteristic function is comparable to a correlator with the corresponding external scaling dimension $\De$. However, when we integrate this finite-dimensional pull-back measure against a radial monomial factor as in eq.~(\ref{radmonomialdecomp}), it is easy to see that the result is not a solution to the s-t crossing equation for the associated external dimension. The reason for this is that, unlike the finite-dimension pull-back, the true OPE measure does not have compact support, and instead encodes an infinite sum over local operators in order to correctly reconstruct the t-channel OPE singularity. This fact leads to a key question: \textit{given a classical measure and its finite-dimension pullback $\mu_{\De_n}$, can one canonically construct an exact solution to the s-t crossing equation for external dimension $\De_n$?} Moreover, \textit{where in position space does this canonically constructed solution match the solution which uniquely corresponds to the exact OPE measure $\mu_n$? }

The answer to the first question is in the affirmative, where the reconstruction procedure involves integrating the finite-dimension pull-back measure against a term which is given by an infinite sum over radial monomials. We call this term a ``coherent state," due to the fact that it tends towards a minimal wavepacket in $(\bs{\xi},\bs{\eta})$ phase-space (where $\bs{\xi} = 2\pi\tilde{\bs{t}}$) that maximizes the Fourier uncertainty relation $\sigma_{\bs{\xi}}\sigma_{\bs{\eta}} \geq \hbar$, after making the identification $\hbar \equiv 1/ \De$ and taking the classical limit $\hbar\to0$. 

The second question is answered in an extremely similar way to how we remedied the aforementioned coarse-graining ambiguity. By requiring a convolution condition between the finite-dimension pull-back measures and their corresponding exact OPE measures to be satisfied, we can guarantee that the correlator computed with the coherent state decomposition agrees with the true correlator within position-space deviations of size $\ll 1/\sqrt{\De}$ around the self-dual point. As we will see in later sections, the coherent state decomposition of a correlator is a convenient object to work with, and motivates a natural parameterization of the correlator on $\mathscr{C}$. 

To begin, let us assume that there exists a function $\psi(h;\chi)$, such that
\ba
\tilde{\mathcal{G}}_\De(\bs{\chi}) = \int  \psi(h;\chi) \psi(\bh;\bar{\chi})d\mu_\De(\bs{h})
\label{psidecomp}
\ea
is normalized to $1$ at $\bs{\chi} = \mathbf{1}$ and satisfies the s-t channel crossing equation for external dimension $\Delta$:
\ba
\tilde{\mathcal{G}}_\De(\bs{\chi})  = |\bs{\chi}|^\De \tilde{\mathcal{G}}_\De(1/\bs{\chi}) .
\ea
We can equate the image of the decomposition in eq.~(\ref{psidecomp}) under reflection symmetry to its image under crossing to obtain the consistency condition
\ba
\int  \psi(\sqrt{2}\De - h;\chi) \psi(\sqrt{2}\De - \bh;\bar{\chi})d\mu_\De(\bs{h}) = |\bs{\chi}|^\De\int \psi(h;1/\chi) \psi(\bh;1/\bar{\chi}) d\mu_\De(\bs{h}) .
\label{decompconst}
\ea
\noindent There may be many choices of $\psi$ that satisfy this condition. However, we propose a \textit{canonical} choice with the following two properties:
\begin{enumerate}
    \item The consistency condition in eq.~(\ref{decompconst}) holds true term-by-term, so that
    \ba
    \psi(\sqrt{2}\De - h;\chi) = \chi^\De \psi(h;1/\chi).
    \ea
    \item The $h = 0$ term only includes the contribution from the identity operator in the direct channel:
    \ba
    \psi(0;\chi ) = 1.
    \ea
\end{enumerate}
The first condition is motivated by the observation that the reflection symmetry of the classical measure is derived from crossing symmetry. Therefore, eq.~(\ref{decompconst}) should not further constrain any properties of $\mu_\De$, which implies that the integrands themselves are equivalent. The second condition is motivated by the fact that, in standard decompositions of the correlator into radial monomials or global conformal blocks, the term associated with vanishing quantum numbers (in this case $h = 0$) is independent of position space variables and therefore only captures contributions from identity exchange. 

Imposing these two conditions fixes
\ba
\psi(h;\chi) = \chi^{h/\sqrt{2}}.
\ea
It is easy to see that this choice satisfies both $\psi(0;\chi) = 1$ and $\chi^\De \psi(h;1/\chi) = \chi^{\De - h/\sqrt{2}} = \psi(\sqrt{2}\De - h;\chi)$. Moreover, if we plug this back into eq.~(\ref{psidecomp}), then reflection symmetry becomes identical to s-t crossing symmetry, as we have engineered it to be. 
\subsubsection{Exchanging minimal wavepackets}
The coherent state decomposition admits a natural interpretation as a semiclassical expansion: each $\bs{\chi}^{\bs{h}/\sqrt{2}}$ contribution to the correlator arises from the exchange of a minimal wavepacket that saturates the Fourier uncertainty relation in the limit $\De \to \infty$. For simplicity, let us work in the diagonal limit and introduce conjugate position and spectral variables $\xi$ and $\eta$. We define the effective Planck constant 
\ba
\hbar \equiv \frac{1}{\De},
\ea
so that the heavy limit $\De \to \infty$ becomes the classical limit $\hbar \to 0$. Let $\xi = 2\pi t$ denote the physical angular coordinate on the torus fiber, so that $\rho = \rho_* e^{i\xi}$ and $\xi/\hbar = 2\pi \tilde{t}$ recovers the rescaled variable used elsewhere. In standard quantum mechanics, a position-space wavefunction takes the form
\ba
\langle \xi | \Psi \rangle = \Psi(\xi) = \frac{1}{\sqrt{2\pi \hbar}} \int \Phi(\eta) e^{\frac{i}{\hbar} \eta \xi} d\eta,
\label{wavefunction}
\ea
where $\langle \xi | \eta\rangle = \frac{1}{\sqrt{2\pi\hbar}} e^{\frac{i}{\hbar} \eta \xi}$ is a normalized plane-wave eigenstate with momentum $\eta$. Up to the normalization factor, this is directly analogous to the correlator evaluated along the torus fiber:
\ba
\mathcal{G}_n(\rho_* e^{i\xi}) = \int e^{\frac{i}{\hbar} \eta \xi} d\nu_n(\eta),
\ea
under the identification $\hbar = 1/\De_n$ and $d\nu_n(\eta) \leftrightarrow \Phi(\eta) d\eta$.
 
Now consider a single coherent state term $\chi^{h_0}$ contributing to a correlator with external dimension $\De = 1/\hbar$. Using the standard expansion of $\chi$ into radial monomials, we have for general $\rho$:
\ba
\chi(\rho)^{h_0} = \sum_{k= 0}^\infty \frac{4^{h_0} \Gamma(2h_0 +k)}{\Gamma(2h_0) k! } \rho^{h_0 +k}.
\ea
Evaluating at $\rho = \rho_* e^{i\xi}$ and writing $\eta_0 = h_0 \hbar$ for the rescaled exchanged dimension, this becomes
\ba
\chi(\rho_* e^{i\xi})^{\eta_0/\hbar} = \int e^{\frac{i}{\hbar} \eta \xi} d\nu_\hbar(\eta),
\ea
where the exact spectral measure is a discrete sum supported on a lattice with spacing $\hbar$:
\ba
d\nu_\hbar(\eta) = \sum_{k = 0}^\infty \frac{4^{\eta_0/\hbar} \Gamma(2\eta_0/\hbar +k)}{\Gamma(2\eta_0/\hbar) k!} \rho_*^{\eta_0/\hbar + k} \delta(\eta  -\eta_0 - k \hbar) d\eta.
\label{nuhbarexact}
\ea
 
This spectral measure is well-approximated by a Gaussian wavepacket. Expanding $\log\chi(\rho_* e^{i\xi})$ to second order in $\xi$ gives
\ba
\chi(\rho_* e^{i\xi})^{\eta_0/\hbar} = e^{i\sqrt{2}\, \eta_0\, \xi/\hbar \;-\; \eta_0\xi^2/(4\hbar)} \left( 1 + O(\eta_0 \xi^3/\hbar) \right),
\label{gaussianapprox}
\ea
which is a Gaussian envelope in position space of width $\sqrt{2\hbar/\eta_0}$. The approximation is valid within this envelope: evaluating the error at $\xi \sim \sqrt{\hbar/\eta_0}$ gives $O(\sqrt{\eta_0 \hbar})$, which is small in the classical limit for fixed $\eta_0$.
 
Taking the Fourier transform of eq.~(\ref{gaussianapprox}) with respect to the plane wave $e^{i\eta\xi/\hbar}$ identifies the approximate spectral measure as a Gaussian centered at $\eta = \sqrt{2}\eta_0$:
\ba
d\nu_\hbar(\eta) \approx \Lambda_{\sqrt{\eta_0 \hbar/2}}(\eta - \sqrt{2}\eta_0) \, d\eta,
\label{nuhbarapprox}
\ea
with spectral width $\sigma_\eta = \sqrt{\eta_0 \hbar/2}$. Since the position-space characteristic function has an amplitude width $\sigma_\xi = \sqrt{2\hbar/\eta_0}$, the uncertainty product is
\ba
\sigma_\eta \, \sigma_\xi = \hbar.
\ea
Among all probability measures with a given mean and this characteristic function width, the Gaussian uniquely minimizes $\sigma_\eta$, so the coherent state is a minimal wavepacket in this precise sense~\cite{FollandSitaram1997}.

Comparing the Gaussian approximation in eq.~(\ref{nuhbarapprox}) to the exact discrete measure in eq.~(\ref{nuhbarexact}), we see that the true spectrum is supported on the lattice
\ba
\eta = \eta_0 + \hbar \mathbb{Z}_+,
\ea
where the lattice spacing $\hbar$ is set by contributions of individual radial monomials. As $\hbar \to 0$ with $\eta_0$ fixed, the spectral width $\sigma_\eta = \sqrt{\eta_0\hbar/2} \to 0$, and the Gaussian collapses to a single delta mass at $\eta = \sqrt{2}\eta_0$. In other words, the ``quantum'' modes given by the contributions of $\langle \xi | \eta \rangle$ to the correlator resum into a single ``classical'' mode with rescaled dimension $\eta = \sqrt{2}\eta_0$. This is precisely the reason why we refer to $\nu$ as a \textit{classical} measure: it records the locations and relative weights of these delta masses, which are the endpoints of the coherent state concentrating to a point as $\hbar \to 0$.

\subsubsection{Where do semiclassics apply?}
We now want to address the question of where along the torus fiber around the self-dual point $T^2_{\bs{\rho}_*} \subset \widehat{\mathscr{C}}$ is the correlator described semiclassically, that is, in terms of exchanges of coherent states which are uniquely determined by the classical measure. This is answered by the following theorem:

\begin{thm}[Coherent state decomposition]
\label{cohstatedecompthm}
Let $\{\mu_n\}$ denote a heavying sequence of OPE measures, with their respective rescaled measure sequence $\{\nu_n\}$ weakly converging to a classical measure $\nu$.
Let $\{\mu_{\De_n}\}$ denote the corresponding sequence of finite-dimension pull-back measures of $\nu$. Define the following sequences of functions for $n \in \mathbb{N}$:
\ba
\mathcal{G}_n(\tilde{\bs{t}}) = \int  e^{i2\pi \tilde{\bs{t}}\cdot \bs{h} /\De_n} d\mu_n(\bs{h})\quad \text{(Exact correlator)}
\ea
and
\ba
\tilde{\mathcal{G}}_{\De_n}(\tilde{\bs{t}}) = \int \chi(\tilde{t})^{h/\sqrt{2}} \chi(\tilde{\bar{t}})^{\bar{h}/\sqrt{2}}d\mu_{\De_n} (\bs{h}) , \quad \text{(Coherent state correlator)}
\ea
with $\chi(\tilde{t}) = \frac{4 \rho_* e^{i2\pi \tilde{t}/\De_n}}{(1-\rho_* e^{i2\pi \tilde{t}/\De_n})^2}$ and $\rho_* = 3-2\sqrt{2}$. Then $\mathcal{G}_n$ and $\tilde{\mathcal{G}}_{\De_n}$ are solutions to the s-t crossing equation for external dimension $\De_n$, admit a positive decomposition into radial monomials, and are normalized to $1$ at $\tilde{\bs{t}} = \bs{0}$ (the self-dual point). Additionally, if the convolution condition from proposition~(\ref{convolutioncond}) is satisfied for $\alpha = 1/2$, then
\ba
\lim_{n \to \infty} \left( \mathcal{G}_n(\De_n^{\beta}\tilde{\bs{t}}) - \tilde{\mathcal{G}}_{\De_n}(\De_n^{\beta}\tilde{\bs{t}}) \right) = 0  \quad \forall\tilde{ \bs{t}} \in \mathbb{R}^2 \;\text{and} \; \beta < 1/2.
\label{cohstatevsexact}
\ea
\end{thm}
\begin{proof}
The first part of the theorem follows by construction, and the second part is proven by approximating the coherent state with lemma \ref{cohstateapprox} and combining with proposition \ref{convolutioncond} for $\alpha = 1/2$. See appendix \ref{cohstatedecompthmproof} for the full proof.
\end{proof}
\begin{rem}
    Both the above theorem and proposition \ref{convolutioncond} conclude with a statement of the form $\lim_{n\to\infty} ( f_n -g_n) = 0$ for some sequences of functions $f_n,g_n$ which do not vanish on their domain. In our setting, both of these function sequences are rapidly oscillating, and are not guaranteed to have a well-defined limit, that is $\lim_{n \to \infty } f_n,g_n $ may not exist. This means our conclusions regard only the asymptotic behavior of these functions. More precisely, we say that two sequences of functions are \textit{additively} asymptotically equivalent if $\lim_{n\to\infty}(f_n -g_n) =0$.  This is contrasted to the standard notion of \textit{multiplicative} asymptotic equivalence, when $\lim_{n\to\infty} f_n/g_n = 1$. In standard Landau little-o notation, we can say $ f_n = g_n +o(1)$, where $f_n,g_n$ are the different pairs of $\phi_{\nu_n}(\De_n^\beta \tilde{\bs{t}}),\;\phi_{\nu}(\De_n^{\beta}\tilde{\bs{t}}),\;\text{and}\;\tilde{\mathcal{G}}_{\De_n}(\De_n^\beta \tilde{\bs{t}})$.
\end{rem}

While the results given by proposition \ref{convolutioncond} and theorem \ref{cohstatedecompthm} seem quite technical, the underlying physical intuition is simple. Studying the correlator along the torus fiber $T^2_{\bs{\rho}_*} \subset \widehat{\mathscr{C}}$ can be viewed as probing the real-time\footnote{``Time'' should be plural since we are dealing with two real (unrescaled) time parameters $\bs{t} = (t,\bar{t})$, but if we restrict to the diagonal $\bar{t} =t$ the standard real-time picture is recovered with the Hamiltonian given by $D - 2\De$.} evolution of the state $\hat{\Pi}_{\psi}\rho_*^{D - 2\De}$ under $e^{i 2\pi \left( t(H-\De) + \bar{t}(\bar{H}-\De) \right)}$. One should view $\phi_{\nu}(\tilde{\bs{t}})$ and $\tilde{\mathcal{G}}_{\De_n}(\tilde{\bs{t}})$ as effective descriptions of the exact time evolution given by $\mathcal{G}_n(\tilde{\bs{t}})$, valid for parametrically small times $\bs{t} = \tilde{\bs{t}}/\De_n$. The pointwise convergence in eq.~(\ref{pointwiseconvergence1}) guarantees that this effective description is accurate and determined entirely by the classical measure $\nu$ for any fixed $\tilde{\bs{t}}$. If we instead evolve the state for parametrically longer (rescaled) times $\tilde{\bs{t}} \sim \De_n^{\beta}$, then the effective descriptions may break down whenever the exact evolution depends on fine structure of $\nu_n$ that is lost in the weak limit to $\nu$. The convolution condition in proposition \ref{convolutioncond} gives a sufficient condition for the effective description given by $\phi_{\nu}(\tilde{\bs{t}})$ to remain accurate at such scales: the discrepancy between $\nu_n$ and $\nu$ must vanish after smoothing the spectrum parametrized by $\bs{\eta}$ over intervals of size $\sim\De_n^{1-\alpha}$ with $\alpha > \beta$. However, $\phi_\nu$ does not correspond to a physical correlator, as it is not a solution to the crossing equations. 

Theorem \ref{cohstatedecompthm} remedies this by constructing an explicit crossing-symmetric solution, the coherent state correlator, which serves as an alternative effective description that is accurate for finite $\tilde{\bs{t}}$. The theorem states that if the convolution condition holds for sufficiently large $\alpha$, then the coherent state correlator remains accurate for $\tilde{\bs{t}} \sim \De_n^{\beta}$ with $\beta < \alpha$, but only up to $\alpha = 1/2$ since the coherent states are wavepackets of width $\sim \sqrt{\De_n}$. This is analogous to the Ehrenfest time scale in semiclassical quantum mechanics, beyond which coherent state approximations break down due to wavepacket spreading \cite{Bambusi1999}: under the identification $\hbar = 1/\De$, our coherent states provide an accurate effective description for times $\bs{t} \lesssim \sqrt{\hbar}$, and resolving dynamics at later times requires information beyond the classical measure. In particular, a physically realizable (crossing-symmetric) effective description of the exact dynamics for times parametrically larger than $\bs{t} \sim \sqrt{\hbar}$ requires modifying the fine structure of the coherent state, which necessarily involves the entire radial monomial decomposition.

\subsubsection{Coherent state rate function}

One can also directly compute the rate function associated to a family of coherent state correlators defined by a classical measure $\nu$, which admits an extremely simple form depending only on the convex hull of the support of $\nu$. We will make use of this result in the next section when we connect the local and global pictures of maximally heavy correlators. 

\begin{prop}[Rate function of a coherent state correlator]
\label{ratefunctioncompute}
Consider a family of coherent state correlators $\{\tilde{\mathcal{G}}_\De(\bs{x})\}_{\De \in \mathbb{R}^+}$ constructed from the classical measure $\nu$. The corresponding coherent state rate function is given by
\ba
    \tilde{\lambda}_{\nu}(\bs{x}) = \lim_{\De\to\infty} \frac{1}{\De} \log\left(\tilde{\mathcal{G}}_\De(\bs{x})\right) = \max_{\bs{\eta} \in \mathrm{conv}(\mathrm{supp}(\nu))} \{ \bs{x} \cdot \bs{\eta}/\sqrt{2}\}.
\ea
\end{prop}
\begin{proof}
As a consequence of Bauer's maximum principle, the maximum of a linear map over a set is equal to the maximum over the convex hull of the set, thus it suffices to prove
\ba
\tilde{\lambda}_{\nu}(\bs{x}) = \max_{\bs{\eta} \in \mathrm{supp}(\nu)} \{ \bs{x} \cdot \bs{\eta}/\sqrt{2}\}.
\ea
Fix any $\bs{x}$ and let $\max_{\bs{\eta} \in \mathrm{supp}(\nu)}\{ \bs{x}\cdot \bs{\eta}/\sqrt{2} \} = \bs{x} \cdot \bs{\eta}_* /\sqrt{2}$. We can write
\ba
\tilde{\mathcal{G}}_\De(\bs{x}) = e^{\De \bs{x} \cdot \bs{\eta}_* /\sqrt{2} } \left( A + \int  e^{\De ( \bs{x}\cdot \bs{\eta} - \bs{x} \cdot \bs{\eta}_*)/\sqrt{2}} (d\nu(\bs{\eta})- A \delta(\bs{\eta} - \bs{\eta}_*)d\bs{\eta})\right),
\ea
where $A\leq1$ is a positive constant denoting the weight of $\nu$ at $\bs{\eta}_*$. By definition, $\bs{x} \cdot \bs{\eta} \leq \bs{x}\cdot\bs{\eta}_*$ for all $\bs{\eta} \in \mathrm{supp}(\nu)$, thus $0<e^{\De ( \bs{x}\cdot \bs{\eta} - \bs{x} \cdot \bs{\eta}_*)/\sqrt{2}} \leq 1$ for all $\De\geq0$. This gives us the bound
\ba
A e^{\De \bs{x} \cdot \bs{\eta}_* /\sqrt{2} }  \leq \tilde{\mathcal{G}}_\De(\bs{x}) \leq e^{\De \bs{x} \cdot \bs{\eta}_* /\sqrt{2} }.
\ea
The rate function is then bounded as
\ba
\lim_{\De \to \infty} \left( \bs{x}\cdot\bs{\eta}_*/\sqrt{2} + \log(A)/\De\right) \leq \tilde{\lambda}_{\nu}(\bs{x}) \leq \bs{x}\cdot\bs{\eta}_*/\sqrt{2}.
\ea
Taking the limit gives $\tilde{\lambda}_{\nu}(\bs{x}) = \bs{x}\cdot \bs{\eta}_*/\sqrt{2} = \max_{\bs{\eta} \in \mathrm{supp}(\nu)}\{ \bs{x} \cdot \bs{\eta}/\sqrt{2}\}$ as we desired.
\end{proof}

\begin{cor} 
\label{corollaryofcohstateratefunction}
Let $\lambda_{\nu_1},\lambda_{\nu_2}$ be rate functions associated to the maximally heavy measures $\nu_1,\nu_2$ respectively. If $\mathrm{conv}(\mathrm{supp}(\nu_1) )\subset \mathrm{conv}(\mathrm{supp}(\nu_2) )$, then $\lambda_{\nu_1}(\bs{x}) \leq \lambda_{\nu_2}(\bs{x})$ for all $\bs{x} \in \mathbb{R}^2$. 
\end{cor}
\begin{proof}
This is another consequence of Bauer's maximum principle, noting that the extremal points of $\mathrm{supp}(\nu_1) $ are always bounded within the extremal points of $\mathrm{supp}(\nu_2) $.
\end{proof}

\subsection{Local rate functions}
\label{localratefunctions}

Up to this point, we have provided two well-defined pictures of the maximally heavy correlator: a \textit{local} one, given by the characteristic function of the classical measure at the self-dual point, and a \textit{global} one, given by a logarithmically regulated limit of a correlator subsequence. Our local picture probes asymptotic deviations of size $\sim1/\De$ around the self-dual point, while our global picture probes arbitrary $O(1)$ deviations. The remaining task is to connect these pictures by studying asymptotic deviations of size $\sim 1/\De^{\alpha}$ around the self-dual point. To this end, we introduce a family of ``$\alpha$-local" dynamical free energy densities, defined as follows:

\begin{defn}[$\alpha$-local dynamical free energy density]
\label{localratefuncdef}
Let $\{(\mathcal{G}_n,\De_n)\}$ denote a heavying sequence of correlators normalized to 1 at the self-dual point. We define their $\alpha$-local dynamical free energy densities with $\tilde{\bs{x}}\in\mathbb{R}^2$ and $\alpha \in [0,1]$ as
\ba
\lambda_n(\alpha;\tilde{\bs{x}}) = \frac{1}{\De_n^{1-\alpha}} \log\left(  \mathcal{G}_n(\tilde{\bs{x}}/\De_n^{\alpha}) \right).
\ea
\end{defn}
\noindent This defines a family of regularizations that maps $\mathbf{G}_{\De_n}\times\mathbb{R} \to C(\mathbb{R}^2)\times\mathbb{R}$. One can again ask if this regularization is good, in that there exists a convergent subsequence that converges locally uniformly to an $\alpha$-local rate function on the same domain. This question is answered in an analogous way to the global rate function via Arzel\`a-Ascoli. We can easily verify that the hypothesis of the theorem holds:
\begin{prop}
\label{preaa2}
Fix $\alpha$ and let $\{\lambda_n(\alpha;\tilde{\bs{x}})\} \subset C(\mathbb{R}^2)$ denote a sequence of $\alpha$-local dynamical free energy densities. The sequence $\{\lambda_n(\alpha;\tilde{\bs{x}})\}$ is locally uniformly bounded and uniformly equicontinuous. 
\end{prop}
\begin{proof}
The proof of proposition \ref{preaa} for the standard dynamical free energy density relies only on the uniform bounds
\ba
\bs{0} \leq \nabla_{\bs{x}} \lambda_n(\bs{x}) \leq \bs{1} \quad \forall \bs{x} \in \mathbb{R}^2.
\ea
Thus, it suffices to prove that
\ba
\label{localgradbound}
\bs{0} \leq \nabla_{\tilde{\bs{x}}} \lambda_n(\alpha;\tilde{\bs{x}}) \leq \bs{1} \quad \forall \tilde{\bs{x}} \in \mathbb{R}^2.
\ea
By definition, we have $\lambda_n(\alpha; \De_n^\alpha\bs{x}) = \De_n^{\alpha} \lambda_n(\bs{x})$ and $\nabla_{\De_n^\alpha\bs{x}} = \De_n^{-\alpha} \nabla_{\bs{x}}$, so the gradient is invariant under rescaling $\tilde{\bs{x}} \to \De_n^\alpha \bs{x}$. Since an arbitrary rescaling of the coordinate maps $\mathbb{R}^2 \to \mathbb{R}^2$, the uniform gradient bounds hold for $\tilde{\bs{x}}$ on the same domain.
\end{proof}
Therefore, Arzel\`a-Ascoli also applies to the sequence of $\alpha$-local dynamical free energy densities; thus, there exists a subsequence which converges locally uniformly to an $\alpha$-local rate function $\lambda(\alpha;\tilde{\bs{x}})$. That is, for all compact subsets $A \subset \mathbb{R}^2$ we have uniform convergence
\ba
\lim_{k\to\infty}\lambda_{n_k}(\alpha;\tilde{\bs{x}}) = \lambda(\alpha;\tilde{\bs{x}}) \quad \forall \tilde{\bs{x}} \in A.
\ea

It is also straightforward to derive the following uniform two-sided bound for $\alpha$-local rate functions with $\alpha>0$:

\begin{prop}[Bounds on $\alpha$-local rate functions]
\label{universalboundonlocalratefunc}
Let $\{(\mathcal{G}_n,\De_n)\}$ be a heavying sequence with a subsequence of $\alpha$-local dynamical free energies $\{\lambda_{n_k}(\alpha;\tilde{\bs{x}})\}$ that converges locally uniformly to an $\alpha$-local rate function $\lambda(\alpha;\tilde{\bs{x}}) = \lim_{k\to\infty} \lambda_{n_k}(\alpha;\tilde{\bs{x}})$. For all $\tilde{\bs{x}} \in \mathbb{R}^2$ and $\alpha>0$, we have the bound
\ba
\frac{\tilde{x} + \tilde{\bar{x}}}{2} \leq \lambda(\alpha;\tilde{\bs{x}}) \leq \max(0,\tilde{x})+\max(0,\tilde{\bar{x}}).
\ea
\end{prop}

\begin{proof}
    The upper bound is derived in the same way as proposition \ref{universalratefunctionbound} using the gradient bound in eq.~(\ref{localgradbound}). The lower bound is derived using the coherent state approximation (lemma \ref{cohstateapprox}), where the error term decays in the heavy limit for $\alpha>0$, and applying Jensen's inequality to the result.
\end{proof}

\subsubsection{A matching theorem}

The appeal of the $\alpha$-local rate function lies in its ability to interpolate between the global and local pictures. For $\alpha \to 0$, we exactly produce the original global rate function with $\tilde{\bs{x}} \to \bs{x}$, while for $\alpha \to 1$, we recover an object directly related to the cumulant generating function of the classical measure. To see this, apply lemma~\ref{cohstateapprox} at $\alpha = 1$. Since $s_n = 1$, we have
\ba
\lambda_n(1;\tilde{\bs{x}}) = \log\mathcal{G}_n(\tilde{\bs{x}}/\De_n) = \log M_{\nu_n}(\tilde{\bs{x}}/\sqrt{2}) + O(\tilde{\bs{x}}^2/\De_n).
\ea
The corollary of proposition~\ref{mgfbound} ensures that $\log M_{\nu_n}(\tilde{\bs{x}}/\sqrt{2}) \leq |\tilde{x}| + |\tilde{\bar{x}}| + O(1/\De_n)$ for any fixed $\tilde{\bs{x}}$, so the first term is bounded and converges to $\log M_\nu(\tilde{\bs{x}}/\sqrt{2})$ as $n\to\infty$, where $M_\nu$ denotes the moment generating function of the classical measure. This is precisely the cumulant generating function of $\nu$ after identifying $\tilde{\bs{\theta}} = \tilde{\bs{x}}/\sqrt{2}$ as the relevant control parameter. 

Now that we have laid out the groundwork and addressed the limiting cases, a key question arises: we know that if we take $\alpha \to 1$ we are accessing a region of the correlator which is uniquely characterized by its classical measure, but for which correlator subsequences are we also able to ``zoom-out," taking $\alpha<1$, and still have the $\alpha$-local rate function be fully determined by the classical measure? Which conditions need to be satisfied to allow for this, and how small can we take $\alpha$ to be? A representative example of a case where this is true for any $\alpha<1$ is given by the family of coherent state correlators constructed from a classical measure $\nu$, whose $\alpha$-local rate function is defined and computed as follows:

\begin{prop}[$\alpha$-local rate function of coherent state correlator]
Consider a family of coherent state correlators $\{\tilde{\mathcal{G}}_\De(\bs{x})\}_{\De\in\mathbb{R}_+}$ constructed from the classical measure $\nu$. The corresponding $\alpha$-local coherent state rate function is given by
\ba
\tilde{\lambda}_{\nu}(\alpha;\tilde{\bs{x}}) &= \lim_{\De \to \infty} \frac{1}{\De^{1-\alpha}} \log \left( \tilde{\mathcal{G}}_{\De}(\tilde{\bs{x}}/\De^\alpha) \right)\\ &= \max_{\bs{\eta} \in \mathrm{conv}(\mathrm{supp}(\nu))} \{\tilde{\bs{x}}\cdot \bs{\eta}/\sqrt{2} \} \quad \forall \alpha  <1.
\ea
\end{prop}
\begin{proof}
    The proof is identical to \ref{ratefunctioncompute} after identifying $\bs{x} \to \tilde{\bs{x}}$ and $\De \to \De^{1-\alpha}$. For all $\alpha < 1$, $\De^{1-\alpha} \to \infty$ as $\De\to\infty$, so the resulting limit is the same. 
\end{proof}
Thus, we have $\tilde{\lambda}_{\nu}(\alpha;\tilde{\bs{x}}) = \tilde{\lambda}_{\nu}(\tilde{\bs{x}})$ for all $\tilde{\bs{x}} \in \mathbb{R}^2$ and $\alpha <1$. One can also find a number of examples, such as correlators of $\phi^L$ generalized free fields with fixed $L$ and $\Df \to \infty$, where the global rate function is also fully determined by the classical measure (these are special cases of coherent state correlators). One example of a correlator subsequence where one can \textit{not} take $\alpha \to 0$ and have the $\alpha$-local rate function determined by the classical measure are correlators of $\phi^L$ generalized fields with $L\to\infty$ and fixed $\De_\phi$. In this case, one can directly show that the $0$-local rate function disagrees with the coherent state rate function predicted by the classical measure. We can also consider the limit of these correlators where $L,\Df \to\infty$ with $L \sim \Df^{1-\alpha_c}$ and $\alpha_c \in (0,1)$. In this case, the $\alpha$-local rate function only matches the coherent state rate function for $1>\alpha > \alpha_c$, and the $\alpha_c$-local rate function serves as a key tool for characterizing the non-universal dynamics of these ``quasi-local" external states. We will work these examples out in detail and give a discussion of their physical interpretations in section \ref{sec:examples}. 

More formally, we answer the key question stated above for the case of $\alpha>0$ with the following theorem:
\begin{thm}[Matching local rate functions]
\label{matchinglocalratefunctions}
Let $\{(\mathcal{G}_n,\De_n)\}$ be a heavying subsequence of correlators such that $\lim_{n\to\infty}\lambda_n(\alpha;\tilde{\bs{x}}) = \lambda(\alpha,\tilde{\bs{x}})$ uniformly on all compact subsets of $\mathbb{R}^2$ and $\lim_{n\to\infty} \nu_n = \nu$ weakly, with $K = \mathrm{conv}(\mathrm{supp}(\nu))$. Then 
\ba
\lambda(\alpha;\tilde{\bs{x}}) =\tilde{\lambda}_{\nu}(\tilde{\bs{x}}) \qquad\text{for all}\qquad \alpha \geq \beta > 0,\, \tilde{\bs{x}} \in \mathbb{R}^2,
\ea
if and only if for every $\delta > 0$ we have
\ba
\lim_{n\to\infty} \frac{\log(\nu_n(K_\delta^c))}{\De_n^{1-\beta}} = -\infty,\quad (\star)
\ea
where $K_\delta = \{ \bs{\eta}\in \mathbb{R}^2\,|\, d(\bs{\eta},K)\leq \delta\}$ and $d(\bullet,\bullet)$ is the Euclidean metric on $\mathbb{R}^2$, with $K_\delta^c$ its complement in $\mathbb{R}^2$.
\end{thm}

\begin{proof}
To sketch, we prove that super-exponential tightness $(\star)$ implies matching using the squeeze theorem, with the Portmanteau theorem determining the lower bound, and that matching implies super-exponential tightness with the upper bound of the G\"artner-Ellis theorem \cite{DemboZeitouni1998}. See appendix \ref{proofofmatchinglocalratefunctions} for the detailed proof.
\end{proof}

The intuitive picture of this result is that, in a sufficiently small window of the self-dual point, the rate function is fully determined by the support of the limiting classical measure $\nu$, as long as the total mass outside the support of $\nu$ decays sufficiently fast as $n \to\infty$. If the mass outside the support of $\nu$ does not decay quickly enough, then it is possible for there to be other classical modes or fat tails of $\{\nu_n\}$ that cause the true rate function to deviate from the coherent state rate function away from asymptotic deviations $\sim 1/\De_n$ around self-duality. The strength of this result lies in the fact that it is a true equivalence between a matching statement for $\alpha$-local rate functions and the super-exponential tightness property $(\star)$. This allows us to directly relate a statement which is position-space in flavor to a rather strong statement that is spectral in flavor. In practice, it is the reverse direction that is most useful. It is rather straightforward to compute rate functions, but difficult to work with the exact rescaled OPE measure sequences. The fact that we can make such a strong statement from matching rate functions alone is a genuinely useful result; any failure of super-exponential tightness at speed $\De_n^{1-\beta}$ would manifest itself as a deviation of the $\alpha$-local rate function from the coherent state prediction at the scale $\alpha = \beta$, giving a direct observational signature.

\subsubsection{Select bounds on coherent state rate functions}
\label{selectbounds}
In proposition \ref{universalratefunctionbound}, we presented a universal bound on rate functions defined globally on $\mathscr{C}$. When the self-dual free energy density $\Sigma = \lim_{n\to\infty} \log(\widehat{\mathcal{G}}_n(\bs{\rho}_*))/\De_n$ satisfies $\Sigma = 0$, these bounds collapse to a universal function along the diagonal limit given by $\lambda(x,x) = \max(0,2x)$. Away from this diagonal, the upper and lower bounds remain non-trivial. We now want to sharpen the upper bound in the case of coherent state rate functions constructed from a classical measure $\nu$ that satisfies certain properties. More specifically, we want to consider a classical measure $\nu$ which receives finite contributions from the identity modes, and has a gap of $\eta_0$ in the spectrum of non-identity classical modes. This can be thought of as a classical twist gap assumption: one is \textit{not} assuming that there are no quantum states in the true OPE measure lying below $ \eta_0\De$. Instead, one is merely assuming that, after passing to the rescaled OPE measure and taking the heavy limit, whichever quantum modes are present in the OPE measure coalesce into classical saddles which lie above this gap. Through the discrete symmetry of the classical measure, this classical twist gap implies that all non-identity modes must lie within a closed square with corners at $\{(\eta_0,\eta_0),(\sqrt{2} - \eta_0,\eta_0),(\eta_0,\sqrt{2}-\eta_0),(\sqrt{2}-\eta_0,\sqrt{2}-\eta_0)\}$. 

Combining this non-universal region with the identity operator contributions at $\bs{\eta} = (0,0),(\sqrt{2},\sqrt{2})$, the convex hull of $\nu$ lies within a diamond with vertices at $\{(0,0),(\sqrt{2} -\eta_0,\eta_0),(\eta_0,\sqrt{2} - \eta_0),(\sqrt{2},\sqrt{2})\}$. This is the maximal convex hull, assuming our classical twist gap assumption. On the other hand, the minimal convex hull is simply a straight line connecting the identity modes, where non-identity modes may still be present along this diagonal. From corollary \ref{corollaryofcohstateratefunction}, we know that a rate function computed with these convex hull assumptions will give an upper and lower bound, respectively, on all coherent state rate functions computed from a classical measure satisfying these assumptions. 

We can then compute the upper bound directly by placing coherent states at the points $\{(0,0),(\sqrt{2} -\eta_0,\eta_0),(\eta_0,\sqrt{2} - \eta_0),(\sqrt{2},\sqrt{2})\}$ in the classical measure as follows:
\ba
\label{ratetwistgap}
\lambda_{\eta_0}(\bs{x}) &= \lim_{\De \to \infty} \frac{1}{\De} \log\left(1 + e^{\De((1 -\eta_0/\sqrt{2})x +\eta_0 \bar{x}/\sqrt{2}) } + e^{\De((1 -\eta_0/\sqrt{2})\bar{x} +\eta_0 x/\sqrt{2}) } + e^{\De(x+\bar{x})}\right)\\
&=\begin{cases}
    0 & \text{if } \bs{x} \in H^{\leq}(\frac{\eta_0 -\sqrt{2}}{\eta_0}) \bigcap H^{\leq}(\frac{\eta_0}{\eta_0 - \sqrt{2}}) \\
    (1-\eta_0/\sqrt{2})x + (\eta_0/\sqrt{2})\bar{x} & \text{if } \bs{x} \in H^{\geq}(\frac{\eta_0 -\sqrt{2}}{\eta_0}) \bigcap H^{\leq}(\frac{\eta_0}{\eta_0 - \sqrt{2}}) \\
    (1-\eta_0/\sqrt{2})\bar{x} + (\eta_0/\sqrt{2})x & \text{if } \bs{x}\in H^{\leq}(\frac{\eta_0 -\sqrt{2}}{\eta_0}) \bigcap H^{\geq}(\frac{\eta_0}{\eta_0 - \sqrt{2}})
    \\
    x +\bar{x} & \text{if } \bs{x} \in H^{\geq}(\frac{\eta_0 -\sqrt{2}}{\eta_0}) \bigcap H^{\geq}(\frac{\eta_0}{\eta_0 - \sqrt{2}})
\end{cases},
\ea
where $H^{ \bullet}(\kappa) = \{ (x,\bar{x}) \in \mathbb{R}^2 \,|\, \bar{x} \bullet \kappa x \; \;\text{or} \;\; x \bullet \bar{x}/\kappa \}$ denotes a half plane. 

Doing the same for the lower bound with support at $\{(0,0),(\sqrt{2},\sqrt{2})\}$ gives us the following inequality for all $\bs{x} \in\mathbb{R}^2$:
\ba
\max(0,x+\bar{x}) \leq \tilde{\lambda}_{\nu}(\bs{x}) \leq \lambda_{\eta_0}(\bs{x}),
\ea
where $\tilde{\lambda}_{\nu}(\bs{x}) = \lambda_{\eta_0}(\bs{x})$ for $\bs{x} \in \left( H^{\leq}(\frac{\eta_0 -\sqrt{2}}{\eta_0}) \bigcap H^{\leq}(\frac{\eta_0}{\eta_0 - \sqrt{2}})\right) \bigcup \left( H^{\geq}(\frac{\eta_0 -\sqrt{2}}{\eta_0}) \bigcap H^{\geq}(\frac{\eta_0}{\eta_0 - \sqrt{2}}) \right)$.\newline

\noindent This result tells us that we can expect three main dynamical phases:
\begin{itemize}
\item a {\it low-temperature phase} associated with the s-channel identity saddle, 
\item a {\it high-temperature phase} associated with the t-channel identity saddle, and
\item a {\it non-universal phase} induced by off-diagonal modes allowed within the region specified by our classical twist gap. 
\end{itemize}
We denote the region that the non-universal phase occurs in as 
\ba
D&{}_\text{non-univ.~heavy}^{(\eta_0)} \\
&= \left(H^{\leq}\left( \frac{\eta_0 - \sqrt{2}}{\eta_0}\right) \bigcap H^{\geq}\left( \frac{\eta_0}{\eta_0 - \sqrt{2}}\right) \right)\bigcup \left(H^{\geq}\left( \frac{\eta_0 - \sqrt{2}}{\eta_0}\right) \bigcap H^{\leq}\left( \frac{\eta_0}{\eta_0 - \sqrt{2}}\right) \right).
\ea
In figure \ref{fig:convexhullandratefunction}, we plot these rate functions for different values of $\eta_0$, along with the associated convex hulls of $\mathrm{supp}(\nu)$.

\begin{figure}[t] 
    \centering 
    \includegraphics[width=1\textwidth]{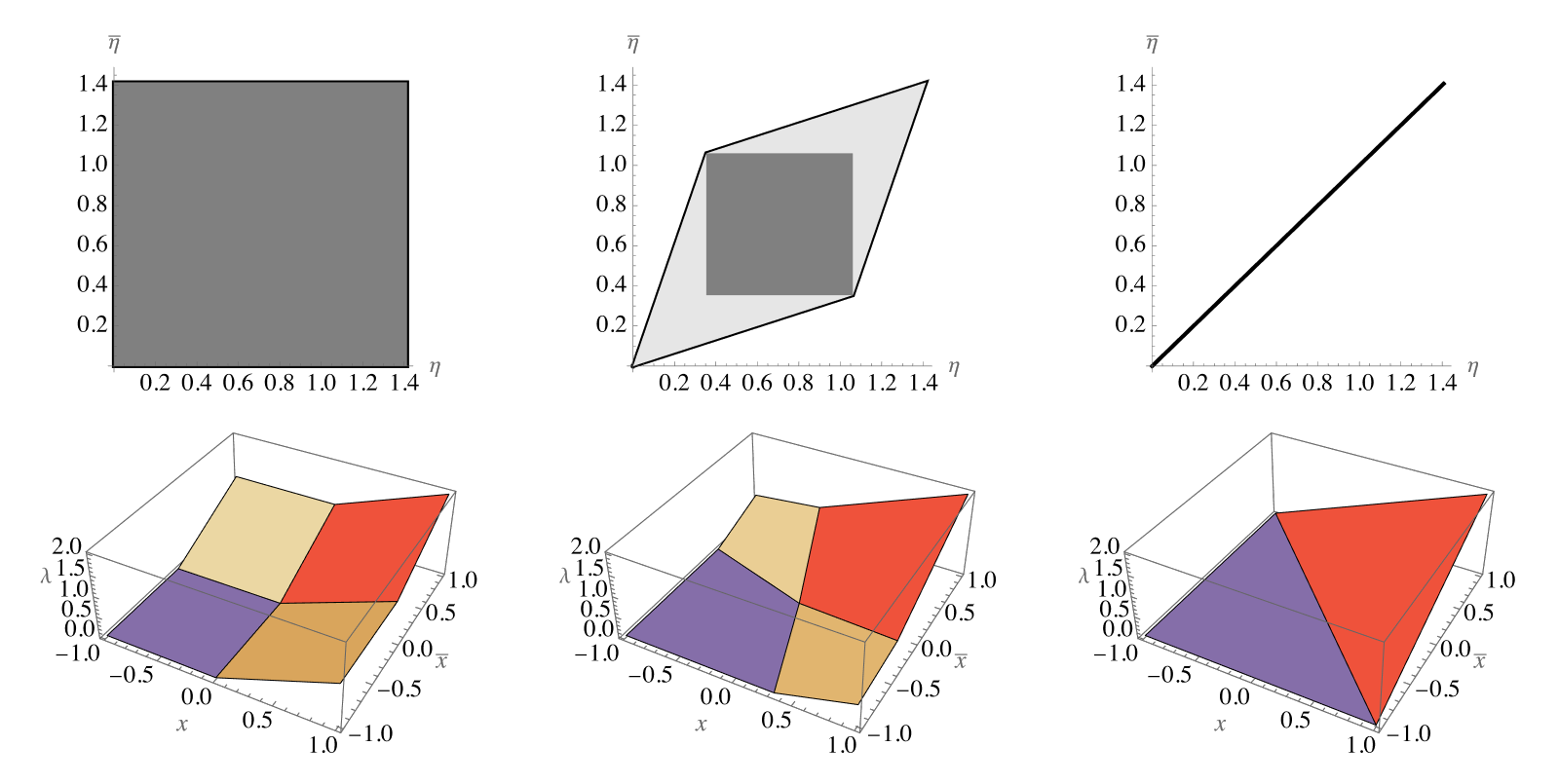}
    \caption{Top row: the light gray region is the convex hull of $\mathrm{supp}(\nu)$ and the dark gray region is the allowed region for non-identity contributions with classical twist gaps of $\eta_0 = \{0,\sqrt{2}/4,\sqrt{2}/2\}$. Bottom row: the coherent state rate functions computed from $\nu$ plotted over $\bs{x} \in [-1,1]\times[-1,1]$. The red region is the high-temperature t-channel phase, the blue region is the low-temperature s-channel phase, and the tan region is a non-universal chiral phase induced by the presence of off-diagonal modes in $\nu$.} 
    \label{fig:convexhullandratefunction} 
\end{figure}

\subsubsection{A parallel to the analysis of HKS and DPQ}
\label{parallelHKS}

The phase structure of maximally heavy correlators described above bears close resemblance to the universal behavior of torus partition functions in 2d CFT at large central charge. In a seminal work by Hartman, Keller, and Stoica (HKS)~\cite{Hartman:2014oaa}, the authors conjectured that the free energy of a torus partition function at large central charge becomes universal for all values of temperature, provided the low-lying spectrum is sufficiently sparse. In a subsequent work by Dey, Pal, and Qiao (DPQ)~\cite{Dey:2024nje}, the authors proved this conjecture and provided a more general formulation of phase diagrams dependent on the twist gap of the theory. The main result of this section is that the non-universal domains in both problems coincide near the self-dual point, when expressed in terms of coherent state variables defined analogously for each setup, with the twist gap parameters related by a simple rescaling.

The torus partition function is studied using the familiar $q$-expansion, given by
\ba
\widehat{Z}(\bs{\beta}) = \sum_{\bs{h}} n_{\bs{h}} \bs{q}^{\bs{h}},
\ea
where $\bs{q} = (q,\bar{q}) = (e^{-\beta_L},e^{-\beta_R})$. Modular invariance is the condition that
\ba
\widehat{Z}(\beta_L,\beta_R) = e^{\frac{c}{24}\left( \frac{4\pi^2}{\beta_L} - \beta_L + \frac{4\pi^2}{\beta_R} - \beta_R \right)}\widehat{Z}(4\pi^2/\beta_L, 4\pi^2/\beta_R).
\ea
In parallel to our setup, we identify $\De = c/24$ and introduce analogous coherent state variables for the torus partition function $\bs{\chi} = (\chi,\bar{\chi}) = (e^x, e^{\bar{x}})$, with
\ba
x = \frac{4\pi^2}{\beta_L} - \beta_L \quad \text{and} \quad \bar{x} = \frac{4\pi^2}{\beta_R} - \beta_R.
\label{cohstatetorus}
\ea
Under this change of variables, the modular crossing equation takes the identical form to the s-t crossing equation for four-point correlators:
\ba
\widehat{Z}(\bs{\chi}) = |\bs{\chi}|^{2\De} \widehat{Z}(1/\bs{\chi}),
\ea
where the self-dual point $(\beta_L,\beta_R) = (2\pi,2\pi)$ maps to $\bs{x} = 0$. 

Now, let us present the results from DPQ. The main object of study is the free energy in the limit $c \to \infty$:
\ba
\Omega(\beta_L,\beta_R) = \log\left(Z(\beta_L,\beta_R)\right) = \log\left(e^{\frac{c}{24}(\beta_L + \beta_R)} \widehat{Z}(\beta_L,\beta_R)\right).
\ea
For $\beta_L \beta_R \geq 4\pi^2$, the free energy is dominated by the vacuum contribution
\ba
\Omega(\beta_L,\beta_R) = \frac{c}{24}(\beta_L + \beta_R) + \mathcal{E}(\beta_L,\beta_R),
\ea
while for $0 < \beta_L \beta_R \leq 4\pi^2$, the free energy is dominated by the vacuum under the image of the modular transform
\ba
\Omega(\beta_L,\beta_R) = \frac{c}{24}\left(\frac{4\pi^2}{\beta_L} + \frac{4\pi^2}{\beta_R}\right) + \mathcal{E}\left(\frac{4\pi^2}{\beta_L}, \frac{4\pi^2}{\beta_R}\right).
\ea
The error term $\mathcal{E}(\beta_L,\beta_R)$ is precisely what DPQ bounds. Along the diagonal $\beta_L = \beta_R$, the free energy transitions directly from the low-temperature phase (thermal AdS) to the high-temperature phase (AdS black holes), with no intermediate non-universal region. Away from the diagonal, near the self-dual line $\beta_L \beta_R = 4\pi^2$, a non-universal phase can emerge where non-vacuum states dominate.

DPQ characterize the non-universal domain in terms of the function
\ba
f(\zeta; y) = \frac{1}{2}\left(\zeta(2-y) + \sqrt{\zeta^2(2-y)^2 + 4(1-\zeta)}\right),
\ea
where the twist gap is $\min(h,\bh) = \zeta (c/24)$. The non-universal domain is
\ba
D^{(\zeta)}_{\text{non-univ.~torus}} = \left\{(\beta_L,\beta_R) \,\Bigg|\, \frac{\beta_L}{2\pi} < f\left(\zeta, \frac{\beta_R}{2\pi}\right), \frac{2\pi}{\beta_R} < f\left(\zeta, \frac{2\pi}{\beta_L}\right)\right\} \bigcup (\beta_L \leftrightarrow \beta_R).
\ea
For all values of $\zeta$, this domain pinches off at the self-dual point $(\beta_L,\beta_R) = (2\pi,2\pi)$. For $\zeta > 1/2$, the domain terminates at finite points on the self-dual line, while for $\zeta \leq 1/2$ it extends along the entire self-dual line.

We now compare the DPQ result to our bound on the rate function given in eq.~(\ref{ratetwistgap}). First, consider the location of the phase transition in $(\beta,\bar{\beta})$ variables, related to radial monomials as $\bs{\rho} = (e^{-\beta}, e^{-\bar{\beta}})$. In these variables, the self-dual line is given by
\ba
\sech^2\left(\frac{\beta}{2}\right) + \sech^2\left(\frac{\bar{\beta}}{2}\right) = 1,
\ea
which is the location of the mutual information transition observed in the causal diamond, in agreement with the results of \cite{Ceyhan:2025qrj}. 

In terms of temperature, this transition occurs at a different location than the Hawking-Page transition. However, both transitions can be compared on equal footing by expressing each in terms of their respective coherent state variables. For four-point functions, we define
\ba
\bs{\chi} = (e^x,e^{\bar{x}}) = \left( \csch^2\left( \frac{\beta}{2}\right), \csch^2\left( \frac{\bar{\beta}}{2}\right)\right),
\ea
while for the torus partition function we define coherent states as in eq.~(\ref{cohstatetorus}). These definitions are chosen so that the crossing equation and modular invariance condition take the same functional form in $\bs{\chi}$.

To express the DPQ domain in terms of coherent state variables, we solve eq.~(\ref{cohstatetorus}) for $(\beta_L,\beta_R)$ to obtain
\ba
\beta_L = \frac{1}{2}\left(-x + \sqrt{16\pi^2 + x^2}\right) \quad \text{and} \quad \beta_R = \frac{1}{2}\left(-\bar{x} + \sqrt{16\pi^2 + \bar{x}^2}\right),
\ea
selecting the branch where $\beta_L, \beta_R > 0$ for all $\bs{x} \in \mathbb{R}^2$. Substituting into the definition of $D^{(\zeta)}_{\text{non-univ.~torus}}$ gives a complicated set of functional inequalities in $\bs{x}$. However, expanding both sides of the inequalities around $\bs{x} = 0$, the boundary conditions linearize. Explicitly, the inequality $\frac{\beta_L}{2\pi} < f(\zeta, \frac{\beta_R}{2\pi})$ becomes
\ba
\bar{x} > x\frac{\zeta - 2}{\zeta} + O(x^2 \zeta),
\ea
and similarly $\frac{2\pi}{\beta_R} < f(\zeta, \frac{2\pi}{\beta_L})$ becomes
\ba
\bar{x} < x\frac{\zeta}{\zeta - 2} + O(x^2 \zeta ).
\ea
The linearized domain is therefore
\ba
D^{(\zeta)}_{\text{non-univ.~torus}} = \left\{(x,\bar{x}) \,\Bigg|\, \bar{x} > x\frac{\zeta - 2}{\zeta} + O(x^{2}\zeta), \; \bar{x} < x\frac{\zeta}{\zeta - 2} + O(x^{2}\zeta)\right\} \bigcup (x \leftrightarrow \bar{x}),
\ea
which, at leading order in small $\bs{x}$, is exactly $D^{(\eta_0)}_{\text{non-univ.~heavy}}$ with $\eta_0 = \zeta/\sqrt{2}$. The non-universal domains therefore coincide near the self-dual point. 

More precisely, we have
\ba
A\left( B_{a}\bigcap\left(D^{(\zeta/\sqrt{2})}_{\text{non-univ.~heavy}} \setminus D^{(\zeta)}_{\text{non-univ.~torus}}\right) \right)/A(B_a) =  O(a \zeta),
\ea
where $B_{a}$ denotes an open $2$-ball of radius $a$ centered around $\bs{x} = 0$, $A(U)$ denotes the area of an open set $U \subset \mathbb{R}^2$, and $D^{(\zeta)}_{\text{non-univ.~torus}} \subset D^{(\zeta/\sqrt{2})}_{\text{non-univ.~heavy}} $ so we can define $D^{(\zeta/\sqrt{2})}_{\text{non-univ.~heavy}}\setminus D^{(\zeta)}_{\text{non-univ.~torus}}$ as standard set subtraction. Thus, for any fixed $\zeta$ and $a\to 0$ or fixed $a$ and $\zeta \to 0$, the area of the difference between the non-universal regions within the ball vanish relative to the total area of the ball. The factor of $\sqrt{2}$ discrepancy between the twist-gap assumption of HKS and DPQ and our classical twist-gap assumption arises precisely from the fact that we are directly bounding the classical spectrum. Indeed, a coherent state with classical dimension $\eta_0$ has a twist-gap in its exact OPE measure at $\zeta = \sqrt{2}\eta_0$, which coincides with the twist gap assumption of HKS and DPQ. 

In figure~\ref{fig:nonuniversalregions}, we plot these non-universal regions for different values of $\zeta$, comparing the exact DPQ domain in $\bs{x}$-coordinates to the linearized domain predicted for maximally heavy correlators with $\eta_0 = \zeta/\sqrt{2}$. The matching of phase diagrams reflects a common underlying mechanism: in both cases, the phase structure is determined by a competition between identity-type contributions localized at the boundary of the spectral support, with the twist gap controlling the extent of the non-universal region where other states can dominate.

\begin{figure}[t] 
    \centering 
    \includegraphics[width=1\textwidth]{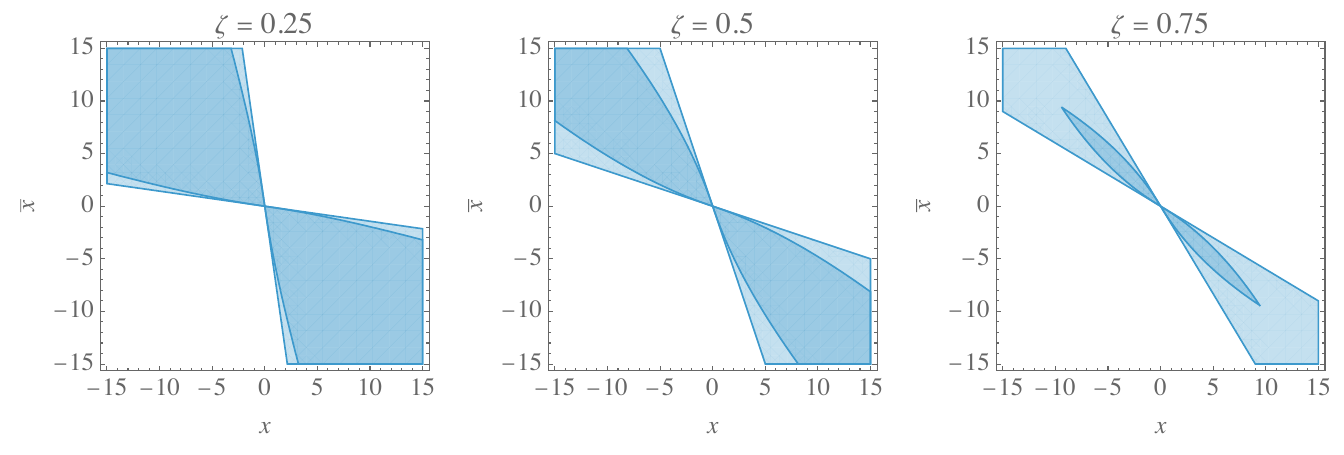} 
    \caption{The non-universal regions of the free energy of a torus partition function (dark blue) and the coherent state rate function of a maximally heavy correlator (light blue) with classical twist gap $\eta_0 = \zeta/\sqrt{2}$ for $\zeta = \{1/4,1/2,3/4\}$. The domains coincide for all $\bs{x}$ as $\zeta \to 0$, and when $|\bs{x}| \ll 1/\sqrt{\zeta}.$} 
    \label{fig:nonuniversalregions}
\end{figure}

\subsubsection{Phases near self-duality}
\label{phasetype}
Since the coherent state rate function constructed from a classical measure $\nu$ depends only on $\mathrm{supp}(\nu)$, the isotropy subgroup classification of the measure support shown in table \ref{tab:orbits} naturally extends to a classification of phases in the vicinity of the self-dual point. In this section, we describe in more detail the different phase structures each type of measure gives rise to. In the next section we will also give explicit examples of Type 0, 1, and 3 measures and the corresponding coherent state correlators. 

\paragraph{Type 0:} This type of classical measure is trivial and contains only a single normalized delta mass at $\bs{\eta} = (\sqrt{2}/2, \sqrt{2}/2)$. As a result, the coherent state rate function is also trivial and is given by
\ba
\tilde{\lambda}_0(\bs{x}) = \frac{x+\bar{x}}{2}.
\ea
When an $\alpha$-local rate function matches this coherent state rate function for $\alpha > \alpha_c$ it means that the bulk states produced by the external operators are delocalized when viewed at scales smaller than $1/\De^{\alpha_c}$. If $\alpha_c = 0$, then the state is completely delocalized and no transition is observed by the global rate function. This rate function is also the lower bound in proposition \ref{universalboundonlocalratefunc} after replacing $\bs{x} \to \tilde{\bs{x}}$.

\paragraph{Type 1:} Measures of purely this type have non-trivial support contained on the diagonal $\{\eta = \bar{\eta} \}$, and therefore have coherent state rate functions which are discontinuous along the self-dual line. If the classical measure has finite contributions coming from the s and t-channel identity saddles, then the convex hull is the entire line segment beginning at $\bs{\eta} = (0,0)$ and ending at $\bs{\eta} = (\sqrt{2},\sqrt{2})$. These measures are quite common in the simplest examples, and give rise to a universal s-channel phase where there is no mutual information shared by operator pairs on the corners of nested causal diamonds, and a universal t-channel phase where mutual information scales linearly as one brings the operator pairs close to one another~\cite{Ceyhan:2025qrj}. The sharpness of this transition comes from the fact that the bulk states sourced by the correlated operators localize along different bulk geodesics that are causally disconnected in the s-channel phase and causally connected in the t-channel phase. 

The coherent state rate functions associated with a general Type 1 measure are of the form
\ba
\tilde{\lambda}^{(\kappa)}_1(\bs{x}) = \begin{cases}
    (1-\kappa)(x+\bar{x}) & \text{if } x + \bar{x} > 0 \\
    \kappa (x+\bar{x}) & \text{if } x+\bar{x} \leq 0 
\end{cases},
\ea
where $\mathrm{conv}(\mathrm{supp}(\nu)) = \{(\sqrt{2} \lambda, \sqrt{2} \lambda)\,|\,
 \forall \lambda\in[\kappa,1-\kappa] \}$ is a diagonal line segment of length $2 (1-2\kappa)$ centered around $(\sqrt{2}/2,\sqrt{2}/2)$, and $\kappa \in [0,1/2]$. The $\kappa = 0$ case is when there is a finite identity contribution, and the segment extends to the full $\mathcal{D}^{\langle C \rangle }$ subspace.  The $\kappa = 1/2$ case simply reduces to $\tilde{\lambda}_0$.

\paragraph{Type 2:} These measures are a $90^{\circ}$ rotated version of Type 1 measures, where instead of the support contained in the diagonal, it is contained in the anti-diagonal of $\{\eta +\bar{\eta} = \sqrt{2}\}$. As a result, the significance of the self-dual line and the diagonal limit flips, and these rate functions encounter a discontinuity only along the self-dual line. 

We can compute a general form of the coherent state rate function in an analogous way to the Type 1 case, which gives
\ba
\tilde{\lambda}_2^{(\kappa)}(\bs{x}) = \begin{cases}
   (1-\kappa) x & \text{if } x  - \bar{x} > 0 \\
    (1-\kappa) \bar{x} & \text{if } x -\bar{x} \leq 0 
\end{cases},
\ea
where $\mathrm{conv}(\mathrm{supp}(\nu)) = \{(\sqrt{2} \lambda, \sqrt{2}(1- \lambda))\,|\, \forall \lambda\in[\kappa,1-\kappa] \}$ is a line segment on the anti-diagonal. The phases in this case arise from the presence of saddles with non-zero classical spin, and should therefore be heuristically interpreted as classical ``vortex" phases, where the direction in which the vortex is spinning is determined by the sign of $x - \bar{x}$. 

While we have not yet found a physically realizable correlator which produces such a phase diagram, it is possible to write down a solution to the s-t crossing equation which does, given by
\ba
\tilde{\mathcal{G}}_\De^{(\kappa)}(\bs{\chi}) = \chi^{\De \kappa}\bar{\chi}^{\De(1- \kappa)} + \chi^{\De(1-\kappa)}\bar{\chi}^{\De \kappa}.
\ea
Computing the coherent state rate function for $\tilde{\mathcal{G}}_\De^{(\kappa)}(\bs{\chi})$  with $\bs{\chi} = ( e^{\bs{x}})$ produces $\tilde{\lambda}_2^{(\kappa)}$.

\paragraph{Type 3:} These are the most general kind of measures, of which Type 0, 1, and 2 are subsets. Coherent state rate functions associated to these measures do not admit a simple general form, and may have phase transitions along both the self-dual line and diagonal limit (along with other lines which cross through the self-dual point). This type of phase diagram arose in our previous discussion of correlators with a specified twist gap \ref{selectbounds}. We further consider a physically realizable correlator of this type in example \ref{chiralproduct}.

\section{Examples}
\label{sec:examples}

In this section, we apply the formalism developed above to a collection of four-point correlators arising in generalized free theories in general dimension, chiral product theories in 2d, and planar $\mathcal{N} = 4$ SYM at tree level. We compute the classical measures and rate functions for each example, and use the results to organize maximally heavy operators into three classes based on the degree to which they localize in the bulk:
\begin{itemize}[leftmargin=*]
\item \textbf{Localized:} operators whose rate functions exhibit sharp phase transitions around the self-dual point that are determined by their classical measure. In these cases, the bulk states sourced by the external operators localize along worldlines, producing distinct phases described by connecting different pairs of external operators.
\item \textbf{Delocalized:} operators whose rate functions are smooth, exhibiting a crossover rather than a sharp transition. These may correspond to states which fragment into a gas of finite-mass particles as $\De \to \infty$, and do not localize along any bulk geodesic. The classical measure collapses to the universal Type 0 form, and loses all sensitivity to the microscopic composition of the state.
\item \textbf{Quasi-localized:} operators which are intermediate between localized and delocalized, and cannot be fully characterized by either the classical measure or the rate function alone. These states are instead distinguished by the $\alpha$-local rate function, for which there exists a critical value $\alpha_c \in (0,1)$ that quantifies exactly how localized the operator is: for $\alpha > \alpha_c$, the local rate function trivializes to the coherent state rate function of the Type 0 classical measure, while for $\alpha < \alpha_c$, a sharp phase transition re-emerges encoding that the state looks localized over large enough distance scales. Only at $\alpha = \alpha_c$ does the local rate function resolve the microscopic details of the state.
\end{itemize}
The first two examples, given by generalized free fields and chiral product correlators respectively, allow one to realize all three classes by tuning the relative scaling of the elementary field dimension $\Df$ and the composite length $L$ in the $\De \to \infty$ limit. The third example, tree-level four-point functions of maximal giant gravitons in $\mathcal{N} = 4$ SYM, provides a physically distinct realization of the localized class: here the external operators have $\De = N$ and are built from elementary fields with fixed classical dimension, so the localization mechanism does not rely on tuning $\Df$ to infinity.

\subsection{Generalized free fields}
\label{GFF example}
The simplest example we will discuss is a correlator of elementary scalar fields, and normal-ordered products thereof, in a generalized free field theory (GFF) on $\mathbb{R}^d$. The boundary action is given by
\ba
S = \int_{\mathbb{R}^d} \phi(\bs{y}) (-\De)^{\gamma} \phi(\bs{y}) d\bs{y} \quad \text{with} \quad \gamma \equiv \frac{d-2\Df}{2},
\ea
where $(-\De)^{\gamma}$ is the fractional Laplacian with dimension $[(-\De)^\gamma] = 2 \gamma$, and $\gamma$ is chosen such that $[\phi] = \Df$. When $\Df = \frac{d-2}{2}$, this just reduces to the standard free field Lagrangian. The equation of motion is simply $(-\De)^\gamma\phi = 0$, so that all correlation functions are computed with Wick contractions. We use an intermediate unit normalization of $\phi$ such that
\ba
\langle \phi^L(0) \phi^L(y)\rangle = \frac{L!}{y^{2\Df}} \quad \text{with} \quad \phi^L \equiv :\phi^L:.
\ea

The correlators we will study are given by
\ba
\widehat{\mathcal{G}}_L(u,v) = \frac{\langle \phi^L(x_1) \phi^L(x_2) \phi^L(x_3)\phi^L(x_4) \rangle}{\langle \phi^L(x_1) \phi^L(x_2) \rangle \langle \phi^L(x_3) \phi^L(x_4) \rangle},
\ea
with external dimension $\De_{\phi^L} = L\Df $. These were computed in \cite{Poland:2024hvb}, and admit a simple closed form as a sum over hypergeometrics
\ba
\widehat{\mathcal{G}}_L(u,v) = \sum_{k}^L \binom{L}{k}^2 u^{k \Df} {}_2 F_1 ( -k,-k;1;v^{-\Df}).
\ea
For integer $k$, the hypergeometric function simplifies into a polynomial, giving the explicit summands of
\ba
\label{hypergeometricsummand}
u^{k \Df} {}_2 F_1 ( -k,-k;1;v^{-\Df}) = \left( \frac{u}{v}\right)^{\Df k} \sum_{j = 0}^k \binom{k}{j}^2 v^{\Df j}.
\ea
We now want to argue that the normalized correlator in both the heavy limits of $\Df \to \infty$ with $L$ fixed and $L \to \infty$ with $\Df$ fixed can be approximated by the coherent state decomposition 
\ba
\label{cohstateGFF}
\widehat{\mathcal{G}}_L(\bs{\chi}) \simeq  \binom{2L}{L}^{-1}\sum_{k}^L \binom{L}{k}^2 |\bs{\chi}|^{2 \Df k}.
\ea
For the limit of $\Df \to \infty$ and $L$ fixed, this approximation is immediate since any term in~(\ref{hypergeometricsummand}) with $j\geq1$ becomes relatively exponentially suppressed at any point in $\widehat{\mathscr{C}}$. Thus, we are left only with the leading $j=0$ term which is the coherent state of $|\bs{\chi}|^{2\Df k}$. 

For the limit of $L\to \infty$ with $\Df$ fixed, the approximation comes from the fact that the squared binomial coefficients $\binom{L}{k}^2$ are sharply peaked around $k \sim L/2$, so that, after normalization, the terms in the sum with $k = O(1) $ are exponentially suppressed in the $L\to\infty$ limit since $\widehat{\mathcal{G}}_L(\bs{\rho}_*) \gtrsim \binom{2L}{L} \approx 4^L/\sqrt{\pi L}$. Thus, we assume that non-negligible terms in the total sum have $ k \propto L$. We apply a similar logic to the sum in~(\ref{hypergeometricsummand}). Here, the sum is not explicitly normalized, so the $j = O(1)$ terms are not exponentially suppressed, but the terms with $j \propto k$ are since $k \propto L \to \infty$ and $|v| < 1$ in $\widehat{\mathscr{C}}$. Of the $j = O(1)$ terms, the $j = 0$ term dominates near the self-dual point where $v \approx 1/4$, once again leaving us with a lone coherent state of $|\bs{\chi}|^{2\Df k}$ on the RHS of~(\ref{hypergeometricsummand}). While the latter argument is not completely rigorous, we have checked our main results for the $L\to \infty$ and $\Df$ fixed limit by other means, including direct moment computations for the classical measure and numerical computation of the rate function, and have found perfect agreement with the analytic results we present here.

Using our previous analysis for the classical measures associated with coherent state correlators, we find the classical measures for these correlators at fixed $L$ and $\Df \to \infty$ to be
\ba
d\nu^{(L)}(\bs{\eta}) = \binom{2L}{L}^{-1} \sum_k^L  \binom{L}{k}^2\delta^{(2)}\left(\bs{\eta}- \frac{\sqrt{2} k}{L}\bs{1} \right) d\bs{\eta}.
\ea
Since this measure is supported only on the diagonal subspace $\mathcal{D}^{\langle\mathcal{C}\rangle}$, it is in bijection with the same measure restricted to $\bs{\eta} = (\eta,0)$ obtained by integrating out $\bar{\eta}$. The resulting pushforward measure is a discrete measure on $[0,\sqrt{2}]$ given by
\ba
d\nu^{(L)}(\eta)=\binom{2L}{L}^{-1} \sum_k^L  \binom{L}{k}^2\delta\left(\eta- \frac{\sqrt{2} k}{L} \right) d\eta,
\ea
with the full measure on $\mathcal{D}$ recovered by multiplying each delta mass by its image under chiral reflection. 

This gives us a simple way to compute the $L \to \infty$ limit of the measure, where the symmetry factors tend towards a continuous normalized Gaussian weight function over $k$
\ba
\label{binomweightapprox}
\binom{2L}{L}^{-1} \binom{L}{k}^2 \simeq \frac{1}{\sqrt{2 \pi} \varsigma} e^{-\frac{1}{2} \left(\frac{k - L/2}{\varsigma} \right)^2} \quad \text{with}\quad \varsigma \equiv \frac{\sqrt{L}}{2\sqrt{2}}.
\ea
This is simply a normalized Gaussian centered around $k = L/2$ with width $\varsigma \equiv \frac{\sqrt{L}}{2\sqrt{2}}$. Since the delta masses are located at $\eta = \frac{\sqrt{2}k}{L}$, we can convert the above Gaussian into a distribution over $\eta$ by setting $ k = L \eta /\sqrt{2}$ and modifying the Jacobian by the appropriate scale factor coming from the delta mass, i.e.~$\delta(\eta - \sqrt{2}k/L) = |L/\sqrt{2}| \delta(k- L \eta /2)$. Doing so gives the restricted measure of
\ba
d\nu^{(\infty)}(\eta) = \lim_{L \to \infty} \sqrt{\frac{2 L}{\pi}} e^{-2L ( \eta - \sqrt{2}/2)^2} d\eta = \delta(\eta - \sqrt{2}/2) d\eta.
\ea
Multiplying this result by its image under chiral reflection recovers the full classical measure over $\bs{\eta} \in \mathcal{D}$ in the $L \to \infty$ limit as
\ba
d\nu^{(\infty)} (\bs{\eta}) = \delta^{(2)}\left(\bs{\eta} - \frac{\sqrt{2}}{2} \bs{1}\right).
\ea
The corrections from finite $\Df$ do not affect this final result by the same argument as \cite{Poland:2024hvb}: since the width of the total classical measure is always $O(1/\sqrt{L})$, and the width of an individual coherent state wavepacket in $\eta$ space with $k \propto L$ is also $O(1/\sqrt{L})$, the widths of the individual wavepackets giving dominant contributions to the total measure do not grow faster than the width controlled by the binomial factor.

This result can be alternatively obtained by setting $\bs{\rho} = (\bs{\rho}_* e^{\bs{c}})$ and taking derivatives of the correlator with respect to $c,\bar{c}$ and evaluating at $c = \bar{c} = 0$; checking to a high derivative order, we find perfect agreement with the appropriately rescaled moments for the above  measure in the limit of $L \to \infty$ with finite $\Df$. 

To recap, we have computed the following maximally heavy measures for $\mathcal{G}_L$ in both the $\Df \to \infty$ limit for fixed $L$, and the $L \to \infty$ limit for fixed and infinite $\Df$:
\ba
\label{classicalmeasuresGFF}
d\nu^{(L)}(\bs{\eta}) &= \binom{2L}{L}^{-1} \sum_k^L  \binom{L}{k}^2\delta^{(2)}\left(\bs{\eta}- \frac{\sqrt{2} k}{L}\bs{1} \right) d\bs{\eta},\\
d\nu^{(\infty)} (\bs{\eta}) &= \delta^{(2)}\left(\bs{\eta} - \frac{\sqrt{2}}{2} \bs{1}\right).
\ea
Classifying these measures into isotropy subgroups, we find that $\nu^{(L)}$ for finite $L$ is purely Type 1, and $\nu^{(\infty)}$ is Type 0. 

 We will now compute the rate functions of these correlators. For the $\Df\to\infty$ and $L$ fixed case, the rate function is uniquely determined by $\nu^{(L)}$, and is given by  $\tilde{\lambda}_1^{(0)}$ as defined in section \ref{phasetype}. To remind readers, this is simply
\ba
\tilde{\lambda}^{(0)}_1(\bs{x}) = \max(0,x+\bar{x}).
\ea

The rate functions for the $L\to\infty$ and $\Df$ fixed case are a bit trickier to compute, as they are not globally determined by a coherent state rate function. To do so, we will employ a saddle point approximation for the sum over $k$ in eq.~(\ref{cohstateGFF}). As $L \to \infty$, this sum becomes sharply peaked around a critical value $k = k^*$ when
\ba
\binom{L}{k^*}^2|\bs{\chi}|^{2\Df k^*} \approx \binom{L}{k^*+1}^2|\bs{\chi}|^{2\Df (k^*+1)}.
\ea
Expanding both sides of the equation around large $L$ and $k^*$, assuming that $k^*/L$ is held fixed, and solving for $k^*$ gives
\ba
k^{*} = \frac{L|\bs{\chi}|^{\Df} }{1+|\bs{\chi}|^{\Df} }.
\ea
Plugging this value into the summand and taking the $L\to\infty$ limit gives
\ba
\label{saddlepointapproxofGFFcorr}
\frac{\widehat{\mathcal{G}}_\infty(\bs{\chi})}{\widehat{\mathcal{G}}_\infty(\bs{1}) } \simeq \frac{2^{-2 L-1} \left(|\bs{\chi}|^{\Df}+1\right)^{2 L+2}}{\sqrt{\pi } \sqrt{L} |\bs{\chi}|^{\Df}},
\ea
where we have normalized by the asymptotic value of the correlator at the self-dual point in the $L\to\infty$ limit. Writing $|\bs{\chi}| = e^{(x+\bar{x})/2}$ and computing the rate function in the $L\to\infty$ limit gives
\ba
\lambda^{(\Df)}(\bs{x}) = \frac{2}{\Df} \log\left(\frac{1}{2} \left( 1 + e^{\frac{\Df}{2}(x + \bar{x}) } \right) \right) \quad \forall \bs{x} \in \mathbb{R}^2.
\ea

In figure \ref{fig:ratefuncplot}, we plot these GFF rate functions in the diagonal limit $\bar{x} = x$ for both the case of $L\to\infty$ with $\Df$ fixed (for various values of $\Df$), as well as the universal rate function in the case of $\Df\to\infty$ with $L$ fixed. We also include dashed lines corresponding to the lower bound of proposition \ref{universalratefunctionbound} for different values of $\Sigma = \log(4)/\Df$. One observes that in the OPE limits of $x \to \pm \infty$, these bounds are saturated, and that the lower bound coming from Jensen's inequality lies tangent to the rate function at the self-dual point. We note that the case of $\Df =1$ and $L\to \infty$ in the diagonal limit is precisely the correlator of four infinite products of identical $1/2$-BPS displacement operators in the Wilson line defect CFT of $\mathcal{N}=4$ SYM in the limit of large 't Hooft coupling studied in \cite{ferrero2}. These operators are generally viewed as sourcing a free gas of single-particle states in the bulk which do not localize to a geodesic, hence their rate function on $\mathscr{C}$ is smooth and reveals a cross-over rather than a phase transition.

\begin{figure}[t] 
    \centering 
    \includegraphics[width=\textwidth]{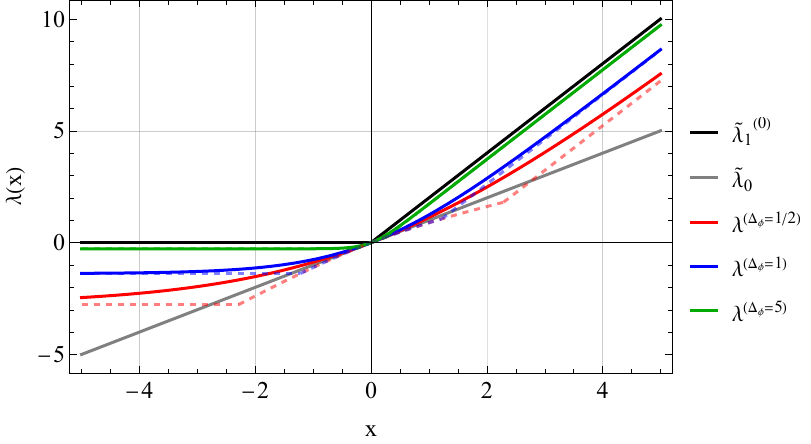} 
    \caption{Rate functions $\lambda(\mathrm{x})$ of GFF correlators in the diagonal limit $\bar{x} =  x$. Black: $\tilde{\lambda}_1^{(0)}$ is the universal rate function of $\mathcal{G}_L$ correlators with $L$ fixed and $\Df\to\infty$. Gray: $\tilde{\lambda}_0$ is the coherent state rate function of a Type 0 classical measure. Red, Blue, Green: $\lambda^{(\Df)}$ are the exact rate functions in the $L\to\infty$ limit for $\Df = \{ 1/2,1,5\}$. The light dashed lines are the lower bounds from proposition \ref{universalratefunctionbound} with $\Sigma = \frac{\log(4^{L})}{L\Df} + O(1/L) = \frac{\log(4)}{\Df}$.} 
    \label{fig:ratefuncplot} 
\end{figure}

\subsubsection{Quasi-local ``fragmentons"}
\label{fragmentons}
We have discussed two maximally heavy limits of the $\mathcal{G}_L$ correlator, one with $\Df \to \infty$ and $L$ fixed, and one with $L\to\infty$ and $\Df$ fixed. The former is ``local" in flavor: we can characterize the external operators by their fixed values of $L$ using local observables such as the classical measure and characteristic function, while they all give rise to the same global picture via their universal rate function. On the other hand, the $L\to\infty$ and $\Df$ fixed operators are ``global" in flavor: they all share the same classical measure, which does not discern their fixed value of $\Df$, while their rate function does. We now want to consider an intermediate regime with a kind of ``quasi-local" external operator, which is neither uniquely characterized by a local nor global regularization. Indeed, the regularization that can discern differences between these states is the $\alpha$-local rate function, for which theorem \ref{matchinglocalratefunctions} is an important tool in diagnosing for which critical value of $\alpha = \alpha_c$ the correlator can be uniquely characterized by its $\alpha_c$-local rate function.

The intuitive physical picture can be understood in the following way: fix the AdS radius to $R =1$ and consider a dual single-particle state with mass $M$ that we want to ``localize" by sending $M \to \infty$, so that only its free geodesic path contributes to the world-line effective action in eq.~(\ref{wleffectiveaction}). If this particle is sufficiently stable, then we can safely take $M\to\infty$ to recover the case when $L = 1$ and $M = \Df \to \infty$, and similarly for a stable multi-particle state with $L>1$. Now, we want to introduce some fragmentation mechanism so that, as we try to increase $M$ beyond some threshold, the initial single-particle state is forced to break into two daughter particles, each with mass $M/2$. As we take $M \to \infty$, we produce an infinite cascade of these fragments whose total mass sums to $M$. If the threshold mass at which fragmentation occurs is a constant, say $M_*$, then taking $M \to \infty$ results in a cascade of $L \sim M/M_*$ particles each with fixed mass $M_*$. This is precisely the case where $L \to \infty$ and $\Df = M_*$ is fixed, producing a delocalized gas cloud in the bulk. More precisely, every time the parent particle fragments, it contributes another world-line effective action term $\sim M_*  \ell(\mathcal{P}_{ij})$, which does not localize onto a geodesic in the path integral since $M_*$ never becomes large. 

In order to produce an excitation that neither completely localizes to a geodesic nor fully fragments into a gas cloud, we can require that the stable mass $M_*$ grows with some power of $M$, say $M_* = \kappa M^{\alpha_c}$, with $\kappa> 0$ and $\alpha_c \in (0,1)$. Since $M = M_* L$, this implies that $L \sim M^{1-\alpha_c}$ as we take $M \to \infty$. We can model this process with our GFF correlator by setting $\Df = M_* = \kappa M^{\alpha_c}$ and $L = \lfloor M^{1-\alpha_c}/\kappa\rfloor$. As we take $M \to \infty$, the total dimension is $\De_{\phi^L} \simeq M$. Since $\alpha_c <1$, taking $M\to\infty$ results in $L\to\infty$, so our approximation of the binomial weights in eq.~(\ref{cohstateGFF}) still holds, and the classical measure still concentrates around $\bs{\eta} = \frac{1}{\sqrt{2}} \bs{1}$. Thus, the classical measure in this fragmentation limit is simply given by $\nu^{(\infty)}$ in eq.~(\ref{classicalmeasuresGFF}). 

The first non-trivial task is to compute the critical speed at which the mass around this point falls off super-exponentially, as described by the condition in theorem \ref{matchinglocalratefunctions}. Using the approximation for the binomial weights in eq.~(\ref{binomweightapprox}) pulled back to the $\bs{\eta}$ variables, and setting $L=L_{\alpha_c} \equiv \lfloor M^{1-\alpha_c}/\kappa \rfloor$, we find
\ba
\nu_{M}^{(L_{\alpha_c})}(K^c_\delta) =O( e^{-\delta^2 M^{1-\alpha_c} /\kappa}).
\ea
Plugging this into the condition of theorem \ref{matchinglocalratefunctions}, we find the condition holds for all $\beta > \alpha_c \sim \frac{\log(M_*)}{\log(M)}$. That is,
\ba
\lim_{M\to\infty} \frac{\log(\nu_{M}^{(L_{\alpha_c})}(K^c_\delta))}{M^{1-\beta}} \sim \lim_{M\to\infty} -\delta^2 M^{\beta - \alpha_c}/\kappa = -\infty
\ea
if and only if $\beta > \alpha_c$. As a result, the coherent state rate function computed from the classical measure only predicts the $\alpha$-local rate function for $\alpha > \alpha_c $. 

More generally, using the saddle point approximation of eq.~(\ref{saddlepointapproxofGFFcorr}), we find that the $\alpha$-local rate function is given by
\ba
\lambda^{(M_*)}(\alpha;\bs{x}) = \lim_{M \to \infty} \frac{2}{\kappa} M^{\alpha-\alpha_c} \log\left( \frac{1 + e^{\kappa M^{\alpha_c - \alpha} (x+\bar{x})/2 }}{2} \right).
\ea
Evaluating the limit explicitly gives rise to three regimes in the $\alpha$-parameter
\ba
\lambda^{(M_*)}(\alpha;\bs{x})  =  \begin{cases}
\displaystyle \frac{x+\bar{x}}{2} & \alpha > \alpha_c \quad \text{(coherent state prediction)}, \\[14pt]
\displaystyle \frac{2}{\kappa}\log \left( \frac{1 + e^{\kappa (x+\bar{x}) / 2}}{2} \right) & \alpha = \alpha_c \quad \text{(critical regime)}, \\[14pt]
\max(0,x+\bar{x}) & \alpha < \alpha_c \quad \text{(sharp phase transition)}.
\end{cases}
\ea
Clearly, it is only in the critical regime of $\alpha = \alpha_c$ that our $\alpha$-local rate function is able to detect the microscopic details of the setup, namely the $\kappa$ parameter that enters the stable mass cutoff $M_* = \kappa M^{\alpha_c}$. Away from this critical regime, the $\alpha$-local rate function is given by either the coherent state rate function prediction from the classical measure ($\alpha > \alpha_c$), or the sharp universal phase transition governed by different pairs of disconnected geodesics $(\alpha < \alpha_c)$. Thus, this simple model of a free ``fragmenton" correlator characterized by the threshold mass $M_*$ shows how the $\alpha$-local rate function can be used to discern the details of certain quasi-local states: ones which neither fully localize to bulk geodesic exchange nor de-localize into a gas of particle states with fixed mass.

\subsection{Chiral product correlators}
\label{chiralproduct}

Another interesting family of unitary solutions to crossing are ones which factorize into a ``chiral" product of correlators separately dependent on $z,\bar{z}$, such as the correlator of energy density operators $\varepsilon$ in the 2d Ising CFT. These correlators are given by a product of generalized free fermion correlators, each with $\Df = 1/2$, which gives a solution to the crossing equation with total external dimension $\De = 1$:
\ba
\frac{\langle \varepsilon\varepsilon\varepsilon\varepsilon\rangle }{\langle \varepsilon\varepsilon\rangle\langle\varepsilon\varepsilon\rangle} &= \widehat{\mathcal{G}}^F_{1/2}(z) \widehat{\mathcal{G}}^F_{1/2}(\bar{z})\\
&= \left(1 - z + \frac{z}{1-z}\right)\left(1 - \bar{z} + \frac{\bar{z}}{1-\bar{z}}\right).
\ea
One can consider a generalization of this correlator of the form $\mathcal{G}_\De^{\mathrm{B/F}}(z)\mathcal{G}_\De^{\mathrm{B/F}}(\bar{z})$, where, for certain values of $\De$, this gives a correlator which is single-valued in the Euclidean section and admits a decomposition into even-spin conformal blocks. These correlators may be physically realized in a theory that decomposes into non-interacting holomorphic and anti-holomorphic sectors. A further generalization is given by correlators of the form $\widehat{\mathcal{H}}_L(\bs{z}) = \widehat{\mathcal{G}}_L(z)\widehat{\mathcal{G}}_L(\bar{z})$, where $\widehat{\mathcal{G}}_L(z)$ is defined by the previous example after setting $(u,v) =(z,1-z)$. This example was constructed using a bosonic field $\phi$, but it is important to note that the terms in the correlator which distinguish fermionic and bosonic fields decay exponentially in the causal diamond as $\De\to\infty$ and therefore are not different from each other for the maximally heavy analysis. 

Focusing on the $\widehat{\mathcal{H}}_L$ generalization, we write
\ba
\widehat{\mathcal{H}}_L(\bs{z}) &= \widehat{\mathcal{G}}_L(z,1-z) \widehat{\mathcal{G}}_L(\bar{z},1-\bar{z})\\
&= \sum_{j,k}^L \binom{L}{k}^2\binom{L}{j}^2 (z^k \bar{z}^j)^{\Df} {}_2 F_1(-k,-k;1;(1-z)^{-\Df}) {}_2 F_1(-j,-j;1;(1-\bar{z})^{-\Df}).
\ea
Restricting to the kinematics within the causal diamond, as $\Df \to \infty$, we are left only with the coherent state decomposition of
\ba
\widehat{\mathcal{H}}_L(\bs{\chi}) \simeq  \sum_{j,k}^L \binom{L}{k}^2 \binom{L}{j}^2 \left( \chi^k \bar{\chi}^j\right)^{\Df}, 
\ea
which corresponds to a maximally heavy measure of 
\ba
d\nu_{\mathcal{H}}^{(L)}(\bs{\eta}) = \binom{2L}{L}^{-2} \sum_{j,k}^L \binom{L}{k}^2 \binom{L}{j}^2 \delta\left(\eta -\frac{\sqrt{2} k}{L} \right)\delta\left(\bar{\eta} -\frac{\sqrt{2} j}{L} \right)d\bs{\eta}.
\ea
Unlike the $\nu^{(L)}$ measure described previously, $\nu_{\mathcal{H}}^{(L)}$ is purely Type 3, and has off diagonal modes which can give rise to a classical vortex phase when the saddles indexed by $(j,k) = (0,L)$ and $(j,k) = (L,0)$ dominate the coherent state decomposition. In this case one has a rate function
\ba
\lambda_{\mathcal{H}}(\bs{x}) = \max(0,x) + \max(0,\bar{x}),
\ea
which is equivalent to $\lambda_{\eta_0}(\bs{x})$ with $\eta_0 = 0$ as defined in eq.~(\ref{ratetwistgap}).

Using the same technique as in the previous example applied separately to the sums over $j,k$, one finds that taking the limit $L \to \infty$ yields the same universal Type 0 measure of
\ba
d\nu^{(\infty)}(\bs{\eta}) = d\nu^{(\infty)}_{\mathcal{H}}(\bs{\eta}) = \delta^{(2)}\left(\bs{\eta} - \frac{\sqrt{2}}{2} \bs{1} \right).
\ea
The rate function of these correlators when $L \to \infty$ and $\Df$ is held fixed can also be computed using the saddle point approximation of eq.~(\ref{saddlepointapproxofGFFcorr}) applied separately to each correlator factor. The result is 
\ba
\lambda_{\mathcal{H}}^{(\Df)}(\bs{x}) = \frac{2}{\Df} \left( \log\left(\frac{1 + e^{\Df x/2}}{2}\right) + \log\left(\frac{1 + e^{\Df \bar{x}/2}}{2}\right) \right).
\ea

Lastly, we can compute the $\alpha$-local rate function for the ``fragmenton" limit where $\Df = M_*= \kappa M^{\alpha_c}$ and $L =\lfloor M^{1-\alpha_c}/\kappa \rfloor$ with $\alpha_c \in (0,1)$ and $M\to\infty$. Using the same techniques as before, we find
\ba
\label{eq:fragmenton_alpha_local}
\lambda^{(M_*)}_{\mathcal{H}}(\alpha;\bs{x})  =  \begin{cases}
\displaystyle \frac{x+\bar{x}}{2} & \alpha > \alpha_c, \quad \\
\displaystyle \frac{2}{\kappa}\left( \log \left( \frac{1 + e^{\kappa x / 2}}{2} \right)+\log \left( \frac{1 + e^{\kappa \bar{x} / 2}}{2} \right) \right) & \alpha = \alpha_c, \, \\
\max(0,x)+\max(0,\bar{x})  & \alpha < \alpha_c. \quad 
\end{cases}
\ea

While $\mathcal{H}_L$ admits a decomposition into a positive sum of radial monomials for all $\Df$, and is thus compatible with the constraints from unitarity in our setup, it does not admit a decomposition into even spin conformal blocks for generic $\Df$. Moreover, for general $\Df$, the correlator does not admit a decomposition into the path-independent quantities given in eq.~(\ref{eq:path-ind}). To see this, one first observes that each factor $\mathcal{G}_L$ admits a positive expansion in radial monomials that contains an identity and the first non-identity monomial $\rho^{\Df}$. Thus, the full decomposition of the correlator contains the term $\rho^{\Df} + \rhob^{\Df}$ (up to an overall coefficent). We write the path-independent quantity as
\ba
\left(\frac{\rho}{\rhob}\right)^{k/2} + \left(\frac{\rhob}{\rho}\right)^{k/2}  = (\rho \rhob)^{-k/2} \left( \rho^k + \rhob^k \right).
\ea
Noting that $(\rho \rhob)^{-k/2}$ is path independent, we see that $\rho^{\Df} + \rhob^{\Df}$ is only path independent for $\Df \in \mathbb{Z}^+$. This can be equivalently viewed as the constraint that only integer spin operators may show up in the OPE. 

An additional constraint arises by imposing that the correlator admits a decomposition into only even spin conformal blocks, which is a consequence of Bose symmetry. This follows from the fact that a three-point function of two identical operators $\cO(x_1),\cO(x_2)$ and an exchanged spinning operator $\cO_\ell(x_3)$ has a unique tensor structure that picks up a factor of $(-1)^\ell$ when $x_1 \leftrightarrow x_2$. Since $\cO(x_1),\cO(x_2)$ are identical bosons, the three-point function must be symmetric under the exchange of $x_1 \leftrightarrow x_2$, thus the correlator is only non-vanishing when $\ell$ is even. Imposing this constraint restricts the heavying sequence to $\Df \in 2\mathbb{Z}^+$. This can be explicitly shown by verifying that the global 2d conformal block decomposition of the correlator only receives contributions from even spin blocks. 

We can check this explicitly with the following procedure. First, perform the change of variables back to the $u,v$ cross ratios:
\ba
z &= \frac{1}{2} \left(-\sqrt{(u-v+1)^2-4 u}+u-v+1\right),\\
\bar{z} &=\frac{1}{2} \left(\sqrt{(u-v+1)^2-4 u}+u-v+1\right).
\ea
For $\Df \in \mathbb{Z}$, $\widehat{\mathcal{H}}_L$ then reduces to the form
\ba
\widehat{\mathcal{H}}_L(u,v) = \sum_{p,q \in \mathbb{Z}} c^{(n)}_{p,q} \frac{u^p}{v^q},
\label{uvexp}
\ea
where the sum is finite for finite $n$ and $\Df$, and $c_{p,q} \in \mathbb{Z}$ are possibly negative coefficients. 

Symmetrizing the $\alpha$-space identity of \cite{Hogervorst:2017sfd} in $z \leftrightarrow \bar{z}$, we obtain the following identity for 2d blocks
\ba
\frac{u^p}{v^q} = \sum_{j,k}^\infty A_j(p,q) A_k(p,q) G_{2 p + j+k , j-k}(u,v),
\ea
where
\ba
A_j(p,q) = \frac{(p)_j^2 }{j! (j+2 p-1)_j} \, _3F_2(-j,j+2 p-1,p-q;p,p;1),
\ea
and
\ba
G_{\De,\ell}(u,v) = \frac{1}{2}\left( k_{\frac{\De +\ell}{2}}(z)k_{\frac{\De -\ell}{2}}(\bar{z}) + (z\leftrightarrow\bar{z})\right)
\ea
denotes the 2d conformal block (with slightly unconventional normalization). Here $k_h(z) = z^h {}_2 F_1(h,h,2h;z)$ is the standard 1d conformal block. It is now just a matter of plugging this identity into eq.~(\ref{uvexp}) and verifying that only conformal blocks with even spins show up in the decomposition (i.e.~that blocks with odd $j-k$ vanish), and that they all have positive coefficients. After checking this explicitly in \texttt{Mathematica} for a large range of $L, \Df \in \mathbb{Z}^+$, we find that only for $\Df \in 2\mathbb{Z}^+$ are both of the conditions satisfied, with odd integer values of $\Df$ resulting in contributions from odd spin blocks. This brute-force approach is crude, but demonstrates the necessary condition of $\Df \in 2\mathbb{Z}^+$ for the $\mathcal{H}_L$ correlator to be compatible with Bose symmetry constraints arising in its conformal block decomposition.

\subsection{Maximal giant gravitons at tree level}
\label{giant graviton example}

The previous examples of maximally heavy correlators arise in generalized free theories, and only exhibit non-trivial phases in their respective rate functions when the dimension of a single elementary field $\phi$ is taken to be infinite. Since the mechanism behind non-trivial phase structure arises from the localization of heavy operators to bulk geodesics, one would also expect sharp transitions to emerge in four-point correlators of maximal giant gravitons, where they have been previously shown to admit a description as a bulk defect \cite{Chen:2025yxg,Chen:2026ium}. These operators have protected dimension $\De = N$, where $N$ is the rank of the $U(N)$ gauge group in $\mathcal{N} = 4$ SYM. These operators are built from the six elementary real scalar fields in the theory, which have a classical dimension of $1$. Therefore, unlike the previous examples where one needs to tune the dimension of a primary scalar operator to be infinite, determinant operators give an example of a heavy state that both localizes to a bulk geodesic, and is constructed from elementary fields with a fixed dimension.

Let us describe the structure of these operators in a bit more detail. Determinant operators are a class of Schur polynomial operators which are taken to be in the totally antisymmetric representation of the permutation group, with a corresponding Young diagram of a single column. Via the Cayley-Hamilton theorem, the length of this column is bounded by $N$, with the maximal case saturating the bound. We construct a determinant operator out of the six real scalar fields $\Phi^I$ in the adjoint representation of $\U(N)$ with $I = 1,...,6$ as:
\ba
\mathcal{D}_i = \det\left(y_i \cdot\Phi(x_i) \right),
\ea
where $y_i \cdot \Phi(x_i) = y_i^I \Phi_I(x_i)$ and $y_i^I$ is a six-dimensional null polarization vector $(y_i \cdot y_i = 0)$ that controls the $SU(4)$ R-symmetry polarization of the operator. In computing a four-point function of these operators, we can shift to a conformal frame in both $SO(2,4)$ and $SU(4)$. In this frame the connected correlator, given by
\ba
\mathrm{GG}(z,\bar{z}; r,\bar{r}) = \frac{\langle \mathcal{D}_1 \mathcal{D}_2 \mathcal{D}_3\mathcal{D}_4 \rangle}{\langle \mathcal{D}_1 \mathcal{D}_2 \rangle \langle \mathcal{D}_3\mathcal{D}_4 \rangle },
\ea
is a function of the standard cross ratios $ u = z \bar{z}$ and $v = (1-z)(1-\bar{z})$, in addition to the R-symmetry cross ratios
\ba
r \bar{r} = \frac{y_{12} y_{34}}{y_{14}y_{23}} \quad \text{and} \quad (1-r)(1- \bar{r} )= \frac{y_{13} y_{24}}{y_{14}y_{23}},
\ea
where $y_{ij} = y_i \cdot y_j$. At tree level, $\mathrm{GG}$ is purely a function of the superconformal cross ratios, given by
\ba
\mathfrak{u} = \frac{z \bar{z}}{r \bar{r}} \quad \text{and} \quad \mathfrak{v} =  \frac{(1-z)(1- \bar{z})}{(1-r)(1- \bar{r})}.
\ea

At weak coupling, we expand $\mathrm{GG}$ as
\ba
\mathrm{GG}(z,\bar{z};r,\bar{r}) = \mathrm{GG}^{(0)}(\mathfrak{u},\mathfrak{v}) + g^2 \mathrm{GG}^{(1)}(z,\bar{z};r,\bar{r})  + O(g^4),
\ea
where $g^2 = g_{\mathrm{YM}}^2 N /(16 \pi^2)$ and $g_{\mathrm{YM}}$ is the Yang-mills coupling. At finite $N$, the tree-level result can be computed using Wick contractions to obtain the closed form result~\cite{Vescovi:2021fjf}
\ba
\mathrm{GG}^{(0)}(\mathfrak{u},\mathfrak{v}) = \sum_{n=0}^N \sum_{m=0}^{N-n} c_{nm} \mathfrak{u}^{N-n}\mathfrak{v}^{n+m-N},
\ea
with coefficients
\ba
c_{nm}  = \frac{1}{N!} \sum_{p,q,s} \frac{1}{(2n + p + q - N)!  s!  (q - p + s)!} \left(
\frac{n!  m!  (N - n - m)!}{(m - p)!  (N - n - m - q)!  (p - s)!} \right)^2,
\ea
where the sum runs over the ranges
\ba
p \in [0, m], \quad 
q \in [\max(0, N - 2n - p),  N - n - m], \quad 
s \in [\max(0, p - q),  p].
\ea
Due to the R-symmetry polarizations of the operators, the s-t crossing equation is modified to
\ba
\mathrm{GG}(z, \bar{z};r,\bar{r}) = \left( \frac{\mathfrak{u}}{\mathfrak{v}}\right)^N \mathrm{GG}(1-z,1- \bar{z};1-r,1-\bar{r}).
\ea
In turn, we have a notion of \textit{super-coherent} states, defined as 
\ba
\mathfrak{X} = \frac{z (1-r)}{(1-z)r} = e^{\mathfrak{x}} \quad\text{and}\quad \bar{\mathfrak{X}} = \frac{\bar{z} (1-\bar{r})}{(1-\bar{z})\bar{r}} = e^{\bar{\mathfrak{x}}}. 
\ea
Note that we could have also defined these by swapping $r \leftrightarrow \bar{r}$ due to an equivalent notion of chiral symmetry in R-symmetry cross-ratio space. The self-dual configuration in this set up is defined as $\mathfrak{x} = \bar{\mathfrak{x}} = 0$, which is also the base point where we define the maximally heavy measure. There are additionally two more conjugate variables corresponding to R-symmetry charge that label the spectral measure. However, at tree-level, the parameter space reduces and we need only consider the conjugate variables to the superconformal cross ratios.

In order to directly compare our computations of maximally heavy observables for supersymmetric correlators to our previous results for non-supersymmetric correlators, we will restrict the R-symmetry cross ratios to the Euclidean self-dual line, defined by the condition $r \bar{r} = (1-r)(1-\bar{r})$ with $r = \bar{r}^*$. Writing $r = x+ i y$ and solving this condition directly, we see this locus is parametrized by
\ba
r = \tfrac{1}{2} + i y, \quad \bar{r} = \tfrac{1}{2} - iy, \quad y \in\mathbb{R},
\ea
for which $r\bar{r} = (1-r)(1-\bar{r}) = \tfrac{1}{4} + y^2$. The superconformal cross ratios on this locus take the form $\mathfrak{u} = \sigma u$ and $\mathfrak{v} = \sigma v$ with $\sigma(y) = ({\tfrac{1}{4} + y^2})^{-1}$, so that $\mathfrak{u}/\mathfrak{v} = u/v$ for all $y$. There is also the special case of $y = \sqrt{3}/2$, corresponding to $r = e^{i\pi/3}$, where $\sigma = 1$ and $(\mathfrak{u},\mathfrak{v}) = (u,v)$, which completely removes dependence on the R-symmetry cross ratios from the problem.

At large $N$, the tree-level correlator on this locus can be computed using the effective field theory description of~\cite{Vescovi:2021fjf}. At leading order around the planar limit, it takes the form
\ba
\mathrm{GG}^{(0)}(\mathfrak{u},\mathfrak{v}) \simeq f(\mathfrak{u},\mathfrak{v}) + f(1/\mathfrak{u}, \mathfrak{v}/\mathfrak{u})\,\mathfrak{u}^N + f(\mathfrak{v},\mathfrak{u}) \left(\frac{\mathfrak{u}}{\mathfrak{v}}\right)^N,
\ea
with $f(a,b) = {b}/\big({(1-a)(b-a)}\big) + O(N^{-1})$. On the R-symmetry self-dual line, $\mathfrak{u} = \sigma u$ and $\mathfrak{v} = \sigma v$, so near the self-dual point where $u = v = 1/4$, the middle saddle scales as $\mathfrak{u}^N \sim (\sigma/4)^N$. For any $y \neq 0$ we have $\sigma < 4$, so this saddle is exponentially suppressed and can be dropped. 

In the remaining two saddles, $\mathfrak{u}/\mathfrak{v} = u/v$ is manifestly $\sigma$-independent, and the only $\sigma$-dependent factors are $1/(1-\sigma u)$ and $1/(1-\sigma v)$ from the function $f$, which both reduce to $1/(1-\sigma/4)$ at the self-dual point. Combining the first and third saddles gives
\ba
\label{eq:GGsaddles}
f(\sigma u, \sigma v) + f(\sigma v, \sigma u)\,(u/v)^N \;\approx\; \frac{1}{4(1-\sigma/4)} \frac{(u/v)^N - 1}{u - v}.
\ea
The prefactor $1/(1-\sigma/4)$ also appears when the correlator is taken to the self-dual point, where $((u/v)^N - 1)/(u - v) \to 4N$, giving $\mathrm{GG}^{(0)}(\sigma/4,\sigma/4) \approx N/(1-\sigma/4)$. Therefore, all $\sigma$ dependence cancels for correlators normalized to $1$ at the self-dual point in conformal cross-ratio space, giving
\ba
 \frac{\mathrm{GG}^{(0)}(\sigma u,\sigma v)}{\mathrm{GG}^{(0)}(\sigma/4,\sigma/4) } \;\approx\; \frac{1}{4N} \frac{(u/v)^N - 1}{u - v}.
\ea

Now, we write this normalized correlator in radial coordinates and evaluate at $\bs{\rho} = \left(\rho_* e^{i2\pi \bs{t}/N} \right)$, where $\rho_* = 3-2 \sqrt{2}$. Taking the $N\to\infty$ limit of the result gives the following characteristic function of the classical measure:
\ba
\phi^{(\mathrm{GG})}(\bs{t}) = \frac{e^{i 2\sqrt{2}\pi(t + \bar{t})} - 1}{i\, 2\sqrt{2}\pi(t + \bar{t})}.
\ea
Since $\phi^{(\mathrm{GG})}$ depends only on $t + \bar{t}$, the classical measure is supported on the diagonal $\bar{\eta} = \eta$. The Fourier inverse of the characteristic function gives the classical measure
\ba
\label{GGclassicalmeasure}
d\nu^{(\mathrm{GG})}(\bs{\eta}) = \frac{1}{\sqrt{2}}\,\delta(\bar{\eta} - \eta) \;\mathbbm{1}_{\eta \in [0,\sqrt{2}]}\, d\bs{\eta},
\ea
where $\mathbbm{1}_{\eta \in [0,\sqrt{2}]} = \Theta(\eta) - \Theta(\eta - \sqrt{2})$ denotes the indicator function. This is simply a uniform distribution on the full diagonal of $\mathcal{D}^{\langle \mathcal{C} \rangle}$.

As a cross-check, this result can also be obtained exactly at finite $N$ in the $y \to \infty$ limit, where $\sigma \to 0$ and the tree-level correlator reduces to
\ba
\mathrm{GG}^{(0)}\big|_{y \to \infty} = \sum_{n = 0}^N (\chi \bar{\chi})^{N-n} = \frac{(\chi \bar{\chi})^{N+1} - 1 }{\chi \bar{\chi} -1},
\ea
using the fact that $c_{n0} = 1$. Normalizing by $N+1$ and taking $N \to \infty$ in the radial parametrization reproduces $\phi^{(\mathrm{GG})}(\bs{t})$ directly. Since $\mathrm{supp}(\nu^{(\mathrm{GG})}) = \mathcal{D}^{\langle \mathcal{C} \rangle}$, this measure is purely Type 1, with a corresponding rate function of $\tilde{\lambda}_1^{(0)}$.

\subsection{Summary}

The examples studied in this section illustrate how the three classes of localized, quasi-localized, and delocalized operators arise concretely in different heavy limits. We collect the main results below, organized by class, and plot the cumulative distribution functions for different classical measures projected onto the $\eta$ parameter in figure \ref{fig:CDFplots}.

\textbf{Localized operators} are those for which the global rate function exhibits a sharp phase transition around the self-dual point, governed entirely by the coherent state rate function constructed from the classical measure. The results for this class are summarized in table~\ref{tab:localized}.
\begin{table}[t]\footnotesize
\centering
\renewcommand{\arraystretch}{2.5}
\begin{tabular}{c|c|c|c}
Correlator & Heavy Limit & Classical Measure $d\nu(\bs{\eta})$ & Type \\
\hline
\hline
$\mathcal{G}_L$ (GFF) & $\Df \to \infty$, $L$ fixed & $\displaystyle\binom{2L}{L}^{-1} \sum_{k=0}^L  \binom{L}{k}^2\delta^{(2)}\!\left(\bs{\eta}- \frac{\sqrt{2} k}{L}\bs{1} \right) d\bs{\eta}$ & 1 \\
\hline
$\mathcal{H}_L$ (Chiral) & $\Df \to \infty$, $L$ fixed & $\displaystyle\binom{2L}{L}^{-2} \sum_{j,k=0}^L \binom{L}{k}^2 \binom{L}{j}^2 \delta\!\left(\eta -\frac{\sqrt{2} k}{L} \right)\delta\!\left(\bar{\eta} -\frac{\sqrt{2} j}{L} \right)d\bs{\eta}$ & 3 \\
\hline
$\mathrm{GG}^{(0)}$  & $N \to \infty$ & $\displaystyle\frac{1}{\sqrt{2}}\delta(\bar{\eta} - \eta)  \mathbbm{1}_{\eta \in [0,\sqrt{2}]} d\bs{\eta}$ & 1 \\
\end{tabular}
\caption{Classical measures for localized operators. The GG result holds uniformly on the Euclidean R-symmetry self-dual line of $r\bar{r} = (1-r)(1-\bar{r})$ with $r = \bar{r}^*$.}
\label{tab:localized}
\end{table}
The corresponding rate functions are given by:
\ba
\lambda_{\mathcal{G}_L}(\bs{x}) = \lambda_{\mathrm{GG}}(\bs{x}) = \tilde{\lambda}_1^{(0)}(\bs{x}) &= \max(0, x + \bar{x}), \\
\lambda_{\mathcal{H}_L}(\bs{x}) &= \max(0,x) + \max(0,\bar{x}).
\ea
The GFF and giant graviton correlators are Type 1, with a mutual information phase transition along the diagonal reflecting a change in causal connectivity between bulk geodesics. The chiral product correlators are Type 3, and support vortex phases sourced by off-diagonal saddles with non-zero classical spin. The chiral product correlators satisfy additional Bose symmetry constraints (decomposition into even spin conformal blocks) only when the heavying sequence is restricted to $\Df \in 2\mathbb{Z}^+$.

\begin{figure}[p]
    \centering
\includegraphics[width=0.95\linewidth]{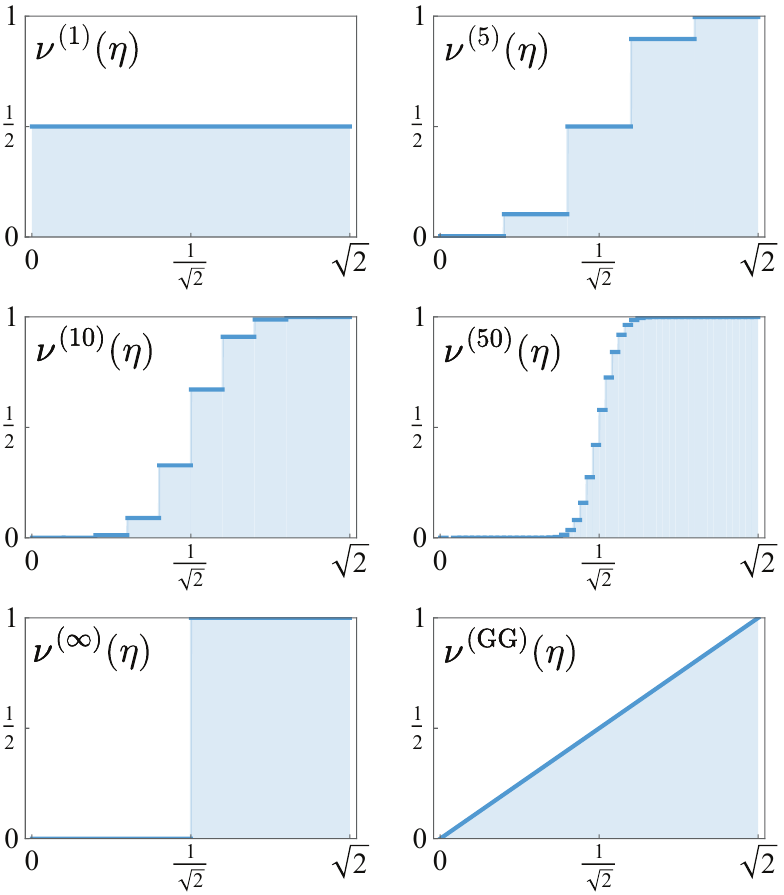}
    \caption{Classical cumulative distribution functions projected onto the $\eta$ parameter (bottom axis) for $\nu^{(L)}$ (the GFF classical measure with $\Df \to\infty$ and fixed $L = 1,5,10,$ and $50$), the universal Type 0 measure $\nu^{(\infty)}$ (the GFF classical measure with $\Df$ fixed and $L\to\infty$), and the classical measure of the maximal giant graviton at tree level $\nu^{(GG)}$. For $\eta > \sqrt{2}$, the CDF is 1 since $\mathrm{supp}(\nu) \subset [0,\sqrt{2}]^2$. }
    \label{fig:CDFplots}
\end{figure}

The giant graviton correlator provides a physically distinct realization of the localized phase diagram. Unlike the GFF and chiral product examples, where one must send $\Df \to \infty$ to achieve localization, giant gravitons are built from elementary fields with fixed classical dimension and localize because their total dimension $\De = N$ is built into the composite structure of the operator. The tree-level classical measure is a uniform distribution on $\mathcal{D}^{\langle \mathcal{C}\rangle}$ at every point on the Euclidean R-symmetry self-dual line $r = \bar{r}^* = \frac{1}{2} + i y$. In these configurations, $(\mathfrak{u},\mathfrak{v}) = (\sigma u ,\sigma v)$ with $\sigma = \left(\frac{1}{4} + y^2\right)^{-1}$ and $\mathfrak{u}/\mathfrak{v} = u/v$, so the singular structure of the EFT saddles is $\sigma$-independent, and the regular $\sigma$-dependent prefactors cancel upon normalization. This universality reflects the insensitivity of the localization mechanism to the choice of self-dual R-symmetry polarization at tree level.

\textbf{Delocalized operators} are those for which the global rate function is smooth around self-duality, exhibiting a crossover rather than a sharp phase transition. The classical measure collapses to the universal Type 0 form in all cases:
\ba
d\nu^{(\infty)}(\bs{\eta}) = \delta^{(2)}\!\left(\bs{\eta} - \frac{\sqrt{2}}{2}\bs{1}\right) d\bs{\eta}.
\ea
This collapse occurs for both $\mathcal{G}_L$ and $\mathcal{H}_L$ in the $L \to \infty$ limit with $\Df$ held fixed. Physically, the external operators fragment into an infinite number of constituents each with fixed mass, producing a delocalized gas in the bulk that does not localize along any geodesic. 

While the classical measure is universal across all delocalized operators, the global rate functions retain sensitivity to the microscopic details. The results are:
\ba
\lambda^{(\Df)}_{\mathcal{G}}(\bs{x}) &= \frac{2}{\Df} \log\left(\frac{1 + e^{\Df(x + \bar{x})/2}}{2}\right), \\
\lambda^{(\Df)}_{\mathcal{H}}(\bs{x}) &= \frac{2}{\Df}\left( \log\left(\frac{1 + e^{\Df x/2}}{2}\right) + \log\left(\frac{1 + e^{\Df \bar{x}/2}}{2}\right)\right).
\ea
These rate functions interpolate smoothly between the s-channel and t-channel regimes, with the width of the crossover region controlled by $1/\Df$.

\textbf{Quasi-localized operators} interpolate between the localized and delocalized classes. They arise in the ``fragmenton" limit, where $\Df = M_* = \kappa M^{\alpha_c}$ and $L = \lfloor M^{1-\alpha_c}/\kappa \rfloor$ with $\alpha_c \in (0,1)$ and $M \to \infty$. The classical measure collapses to the same universal Type 0 form as in the delocalized case, so the local picture cannot distinguish these operators from fully delocalized ones. At the same time, the global rate function ($\alpha = 0$) exhibits a sharp phase transition, so the global picture cannot distinguish them from fully localized ones. The resolution requires the $\alpha$-local rate function, which exhibits three distinct regimes as a function of $\alpha$:
\ba
\lambda^{(M_*)}_{\mathcal{G}}(\alpha;\bs{x}) &= \begin{cases}
\displaystyle \frac{x+\bar{x}}{2} & \alpha > \alpha_c, \\[10pt]
\displaystyle \frac{2}{\kappa}\log \left( \frac{1 + e^{\kappa (x+\bar{x}) / 2}}{2} \right) & \alpha = \alpha_c, \\[10pt]
\max(0,x+\bar{x}) & \alpha < \alpha_c,
\end{cases}\\[16pt]
\lambda^{(M_*)}_{\mathcal{H}}(\alpha;\bs{x}) &= \begin{cases}
\displaystyle \frac{x+\bar{x}}{2} & \alpha > \alpha_c, \\[10pt]
\displaystyle \frac{2}{\kappa}\left( \log \left( \frac{1 + e^{\kappa x / 2}}{2} \right)+\log \left( \frac{1 + e^{\kappa \bar{x} / 2}}{2} \right) \right) & \alpha = \alpha_c, \\[10pt]
\max(0,x)+\max(0,\bar{x}) & \alpha < \alpha_c.
\end{cases}
\ea
For $\alpha > \alpha_c$, the rate function matches the smooth coherent state prediction from the classical measure. For $\alpha < \alpha_c$, it reduces to the sharp phase transition characteristic of a localized operator. Only at the critical value $\alpha = \alpha_c$ does the rate function resolve the non-universal microscopic details of the state, namely the stable mass parameter $\kappa$ entering the fragmentation threshold $M_* = \kappa M^{\alpha_c}$. The critical exponent $\alpha_c$ thus serves as a precise quantitative measure of the degree of localization of the external operator.

\section{Discussion}
\label{disc}

In this work, we have developed a systematic framework for studying the space of CFT four-point functions of identical scalar operators in the heavy limit. This space of correlators captures a wide range of physical phenomena, including classical holography, black hole dynamics, planar gauge theory, flat-space scattering amplitudes, and more.

Our formalism revolves around the notion of maximally heavy observables, which are accumulation points in the image of good regularizations acting on ``heavying" sequences of correlation functions, whose external dimensions form an unbounded sequence in $\mathbb{R}_+$. These are akin to intrinsic quantities that describe statistical systems in the thermodynamic limit. Different maximally heavy observables characterize heavying sequences in different ways. Rate functions give a global picture of a logarithmically regulated correlator sequence on the entire causal diamond, while classical measures give a local picture as the Fourier transform of a characteristic function that probes deviations of size $\sim1/\De$ around the self-dual point. Using the classical measure, we can canonically construct a family of solutions to the crossing equation for an arbitrary external dimension using the coherent state decomposition, and the rate function of the resulting coherent state correlator is solely determined by the convex hull of the support of the classical measure. In order to interpolate between the global and local pictures, we introduce an  ``$\alpha$-local" rate function, which probes asymptotic deviations of size $\sim1/\De^{\alpha}$ around the self-dual point. For $\alpha \to 0$, we recover the standard global rate function, while taking $\alpha\to1$ returns the cumulant generating function for the classical measure. 

In addition to being well-defined for any given heavying correlator sequence, these observables are strongly constrained by crossing symmetry, chiral symmetry, and unitarity. This results in universal bounds on both the value and gradient of rate functions over the entire causal diamond. In the local picture, crossing and chiral symmetry imply that classical measures lie in the space of compactly supported measures on $\mathcal{D} = [0,\sqrt{2}]^2$, which are invariant under a discrete group action isomorphic to the Klein four-group. This establishes a setting where the approximate ``reflection symmetry" result of \cite{Kim2015-sg} is exact. Under the assumptions of a finite identity operator contribution to the classical measure and a twist gap for non-identity contributions, we establish a connection between bounds on the non-universal regions of rate functions and those derived for the free energy of a 2d torus partition function at large central charge in \cite{Hartman:2014oaa,Dey:2024nje}.

Maximally heavy observables also give novel insights into the holographic interpretations of heavy boundary operators. They allow us to develop a refined classification of external states based on the degree to which they localize around bulk geodesics, as described through their world-line effective actions. Localized operators exhibit sharp dynamical phase transitions in their rate functions around the self-dual point, as predicted by the non-trivial convex hull of the corresponding classical measure. Delocalized operators have trivial classical measures that degenerate around a single point, and exhibit smooth crossovers in their rate functions around the self-dual point. Quasi-localized operators have trivial classical measures, yet still produce rate functions that demonstrate a phase transition around the self-dual point. In order to resolve the non-universal microscopic structure of quasi-localized operators, one can tune the $\alpha$-local rate function to a critical value $\alpha =\alpha_c$, where $\alpha_c$ is a critical exponent characterizing the extent to which the state localizes in the bulk. The classification of classical measures also allows for the possible emergence of classical ``vortex" phases in correlators along the self-dual line $\bar{z} = 1-z$, where large-spin saddles dominate the OPE decomposition. We realize these different possibilities in a number of examples, including generalized free theories, chiral product theories in 2d, and giant graviton correlators in planar $\mathcal{N} = 4$ SYM at weak coupling.

A number of compelling questions and future directions arise from this analysis. As an immediate next step, one can compute maximally heavy observables in a wider range of examples, including higher-order perturbative computations in planar $\mathcal{N} = 4$ SYM and other holographic theories at both weak and strong coupling, in order to see how interactions modify the leading-order results. One might expect correlators of delocalized operators in one coupling regime to become localized in another, analogous to a dilute gas coalescing into a single bound state as an attractive interaction is switched on. It might be especially interesting to compute a four-point function of huge operators, such as dual giant gravitons, in order to analyze black hole scattering through the lens of our heavy observables. One might expect very different rate functions for these states, perhaps with more than one phase transition away from the self-dual point. The reason for this is as follows: when operator pairs are taken to be well separated, we expect the correlator to be described by a disconnected banana geometry introduced in~\cite{Abajian2023-xw}, while at the self-dual point these bananas would merge into a single ``quad-nana" encoding the event that a pair of colliding black holes first merge and then separate at a later time. In between these disconnected banana and fully merged geometries, we might expect a phase transition when the event horizons of the black holes begin to interact. Would this be a first or second order transition in the rate function? Would it remain a smooth crossover? These questions can be answered through a direct calculation.

\begin{figure}
    \centering
    \includegraphics[width=\linewidth]{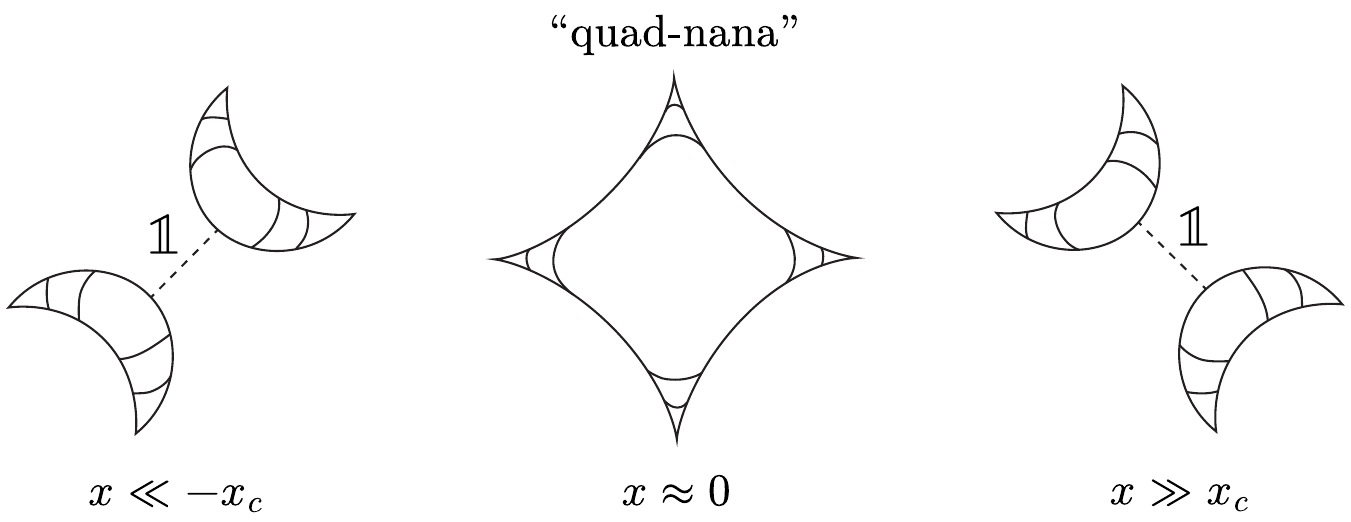}
    \caption{An illustration of proposed phase dynamics for a four-point function of huge operators. Deep in the OPE limits of $|x| \gg x_c$, the bulk effective action is dominated by a disconnected double banana geometry, or ``bi-nanas," associated with the exchange of the s and t-channel identity operators. Near the self-dual point, the bi-nanas merge into a single ``quad-nana," analogous to the ``tri-nana" pictured in figure 12 of \cite{Abajian2023-xw}. Sharp transitions may occur around $x = \pm x_c$ when the event horizons of the disconnected bananas begin to interact.}
    \label{fig:quadnana}
\end{figure}

With this in mind, one hope is that, given that a single maximally heavy observable does not require knowing all the analytic structure of the correlator, computing a maximally heavy observable directly in perturbation theory may be easier to do than computing the entire heavying sequence of correlators (say at increasing values of finite $N$) and then extracting the heavy observables by taking a limit. For example, in the case of free localized operators, we know that the rate function is universally predicted from the free geodesic action of the disconnected pieces. It may be straightforward to compute perturbations of this free rate function using the well-developed toolbox of world-line effective field theory and Witten diagrams~\cite{Strassler:1992zr,Maxfield:2017rkn,Kulkarni:2024ghc,Hijano:2015zsa}. It would also be fruitful to develop a systematic
$1/N$ expansion for dynamical free energy densities directly, complementing the standard perturbative expansion of correlators around large $N$. Since these functions are locally bounded and Lipschitz on the entire causal diamond, they provide a natural starting point for perturbation theory in which the leading-order result is known classically, and subleading corrections would smooth out the sharp phase transitions. This smoothing would reflect the manner in which quantum corrections smear a classical world-line into a fuzzy wave packet at finite $N$.

In holographic theories that admit descriptions in terms of integrable spin chains, such as planar $\mathcal{N} = 4$ SYM~\cite{Beisert:2003tq,Beisert:2010jr}, one can ask how these alternate descriptions translate into maximally heavy observables. In this setting, the number of Bethe roots scales with the system size in the heavy limit, and the thermodynamic Bethe ansatz provides exact finite-coupling predictions for anomalous dimensions. Recent progress in computing structure constants involving determinant operators at finite coupling via the hexagon formalism and $g$-function methods~\cite{Jiang:2019xdz} suggests that integrable descriptions of heavy correlators may be within reach. This may lead to exciting results such as exact rate functions and classical measures at finite 't Hooft coupling, giving us a window into the scattering of huge bodies in non-perturbative quantum gravity. 

Even without gravity, such as in the context of the rigid AdS/CT correspondence, one could develop a connection between our picture of maximally heavy observables and flat-space scattering amplitudes~\cite{vanRees:2022zmr,Paulos2016-mh,vanRees:2023fcf,Penedones:2010ue}. Here, the limit of infinite external scaling dimension is manifest, and one might expect to find heavy observables that can be related directly to properties of the S-matrix. This raises the sharp structural question of whether the S-matrix is itself a maximally heavy observable in our precise sense, and whether the computation of the on-shell T-matrix in \cite{vanRees:2022zmr} is a generically good regulation of a heavying sequence. If so, this would provide a rigorous bridge between the conformal bootstrap and the S-matrix bootstrap \cite{Paulos:2016fap,Paulos:2017fhb}, and would illuminate the sense in which analyticity, crossing, and unitarity of amplitudes emerge from properties of heavy CFT correlators.

While general heavying sequences may correspond to a family of CFTs, one can also focus on heavying sequences which belong to a single CFT. Since the scalar spectrum of any given CFT has an unbounded sequence of dimensions~\cite{Qiao:2017xif,Mukhametzhanov:2018zja}, one can ask about the complete set of maximally heavy observables that a given CFT produces, and their dependence on properties such as central charge, internal global symmetries, and degree of supersymmetry. It is possible that these observables are a fundamental property of a given theory, similar to the spectrum of relevant operators, and thus carry an imprint of the low-lying fusion algebra or underlying equations of motion. It would be interesting to make this precise by studying heavying sequences of quasi-primaries in a known minimal model or rational CFT. Another question is whether a generic CFT at finite central charge contains only one class of excitation, or if a single CFT can furnish localized, delocalized, and quasi-localized operators, as one might expect for holographic CFTs with all different kinds of asymptotic classical states, each having different scalings with $N$. 

While we have only studied four-point functions of identical scalar operators in this work, it is natural to consider correlation functions of non-identical maximally heavy operators. One can ask whether our regularizations remain well-defined for such heavying sequences, or which regularizations are better suited to these cases. This analysis could also be extended to maximally heavy operators in the presence of defects, such as the setup of~\cite{Linardopoulos:2026mut}, where two-point functions of heavy operators were computed in the presence of a co-dimension 1 defect dual to a probe brane in AdS. One expects a dynamical phase transition to emerge in this setting through a similar mechanism: when the operators are close together relative to their distance from the defect, the bulk effective action is dominated by a single geodesic connecting them, and as the operators are taken far apart, it becomes dominated by the configuration where each operator connects independently to the probe brane via a minimal length geodesic. Studying higher-point functions of maximally heavy identical operators could also be fruitful, granting access to additional dynamical phases in a configuration space with more cross ratios. In this vein, there is the more bootstrap-oriented question of how crossing and chiral symmetry can be generally understood for multi-point functions as a generalization of reflection symmetry for the corresponding space of OPE measures.

In our setup, we introduced the causal diamond as a statistical manifold fibered by OPE measures defined at each point. It would be interesting to develop this interpretation further, and to better understand the information-theoretic aspects of this problem. In this setting, the Hessian matrix of the rate function is the Fisher information metric~\cite{Ruppeiner:1995zz,Amari:2000ig}, which can be used to detect phase transitions when the scalar curvature diverges. One can also consider higher curvature invariants, which may be used to detect second order phase transitions in the rate function. There are additional information-theoretic properties of the classical measure to consider as well, such as the spectral entropy and Kullback--Leibler divergences relating different classical measures. There is then a natural question of how these quantities encode bulk physics, which may be answered similarly to the work of~\cite{Bohra:2021zyw}, generalized to four-point functions.

Another exciting direction would be to clarify the physical origins of the quasi-local fragmenton limit described in section~\ref{fragmentons}. Namely, we would like to find a dynamical mechanism that would produce such a fragmentation cascade in a known holographic theory. One real-world mechanism which produces a similar effect is that of higher-order soliton fission in nonlinear dispersive media. In this setting, a soliton propagating in an optical fiber is characterized by a soliton number $N \propto \sqrt{P_0}$, where $P_0$ is the peak power~\cite{Agrawal:2019nfo,Dudley:2006rmp}. When perturbed by higher-order dispersion or Raman scattering, this higher-order soliton fragments into $N$ fundamental solitons, each centered at a different frequency, and the resulting cascade generates a broadband supercontinuum of excitations as it propagates through the medium~\cite{Herrmann:2001prl,Husakou:2002josab}. This may be viewed as a parallel to the fragmenton limit with $\alpha_c = 1/2$, where we take the external dimension to be proportional to the peak power. 

A potentially closer holographic parallel may be provided by grey galaxy solutions in AdS~\cite{Kim:2023sig,Bajaj:2024utv}, which have been proposed as endpoints of the superradiant instability of Kerr-AdS black holes. In these configurations, the initial rotating black hole progressively sheds energy and angular momentum into co-rotating bulk modes until it reaches a critical angular velocity, at which point the instability saturates and the system settles into a composite state: a localized black hole at the center of AdS surrounded by a delocalized disk of gravitons revolving at the speed of light. This is a fragmentation cascade with a definite dynamical endpoint, producing a state that appears quasi-localized in the sense of our classification. However, the grey galaxy fragmentation pattern is qualitatively different from the uniform fragmenton model of section~\ref{fragmentons}: rather than producing $\sim M^{1-\alpha_c}$ identical fragments each of mass $M_* \sim M^{\alpha_c}$, the superradiant cascade produces a bimodal mass spectrum consisting of one heavy remnant carrying an $O(M)$ fraction of the total energy and $O(M)$ light quanta each carrying $O(1)$ energy in AdS units. If this picture extends to correlators, it suggests that no single critical exponent $\alpha_c$ would characterize such a state. In this case, the $\alpha$-local rate function may exhibit a richer interpolation between the localized and delocalized regimes than the three-regime structure of eq.~(\ref{eq:fragmenton_alpha_local}), with the black hole and gas components entering at different characteristic scales as $\alpha$ is varied. It would be interesting to make this precise by computing maximally heavy observables for correlators of CFT operators dual to grey galaxy states. 

One can ask a similar question for the emergence of classical vortex phases in heavy correlators. Since the chiral product correlator example we considered in this work is rather tautological, it would be extremely interesting to find a more physical example of a heavying sequence of correlators in $d>2$ which produces a rate function that exhibits this phase. Developing this picture further could also help clarify the physical mechanism behind the non-universal phases observed along the self-dual line in~\cite{Dey:2024nje}.

More generally, it would be exciting to extend our construction of maximally heavy observables to large-$c$ observables in 2d CFT torus partition functions. These setups are enticingly similar, but with modular invariance notably more constraining than crossing symmetry. This is evident by the fact that the bound on the non-universal regions of the free energy of the torus partition function in $(\beta_L,\beta_R)$ space is much tighter than what we obtained in the heavy correlator setup (see figure \ref{fig:nonuniversalregions}). It would be very interesting to use $\alpha$-local rate functions to more systematically probe the fine structure of the partition function near the self-dual point, as well as explicitly compute the classical measure for known holographic theories such as a symmetric product orbifold or Narain CFTs \cite{Dijkgraaf:1996xw,AfkhamiJeddi:2020hde, Benjamin:2021wzr,Maloney:2020nni}. These could give a much more refined characterization of the set of CFT operators whose quantum numbers grow proportional to the central charge that the free energy alone is insensitive to. There are likely other maximally heavy observables that are especially suited for torus partition functions, such as those which encode the universal Schwartzian sector discussed in \cite{Ghosh:2019rcj}.

In conclusion, the space of maximally heavy dynamics in the causal diamond is a vast, largely unexplored landscape which captures a huge\footnote{No pun intended.} variety of physics both in holographic settings and beyond. Very few examples of heavying correlator sequences have been computed beyond tree level, and there is a great deal of rich structure to be uncovered through direct computation in perturbative and non-perturbative settings. It is our hope that this work provides a robust and mathematically rigorous foundation for these fascinating dynamics to be explored further, shedding light on key problems in both holography and the conformal bootstrap, as well as inspiring novel questions about quantum gravity, gauge theory, and scale-invariant phenomena.

\acknowledgments

We thank Carlos Bercini, Li-Yuan Chiang, Liam Fitzpatrick, Tom Hartman, Nafiz Ishtiaque, Petr Kravchuk, Yue-Zhou Li, George Linardopoulos, Tony Liu, Jeremy Mann, Matthew Mitchell, Ian Moult, Raman Sundrum, Balt van Rees, and Yuan Xin for discussions. We thank Anshul Adve, Sridip Pal, and Allic Sivaramakrishnan for comments on the draft. The authors were supported by DOE grant DE-SC0017660.

\begin{appendix}
\label{appendix}

\section{Asymptotics}
Let $(\rho,\rhob)  = (e^{-t-s_0}, e^{-\bar{t} - \bar{s}_0}) $ and $(\rho_0,\rhob_0) = (e^{-s_0},e^{-\bar{s}_0})$. In this appendix, we will present the asymptotics of various factors involving radial monomials around $t = 0$, computed using \texttt{Mathematica}. 

\begin{lem}

Let $h \in (0,\infty)$, $\mathrm{\mathbf{t}} = t \csch(s_0/2)$, and $N \in \mathbb{N}_{\geq2}$, then we have the following asymptotic expansion around $t = 0$:
\ba
\left(\frac{1-\sqrt{\rho}}{1+\sqrt{\rho}}\right)^{2h} =& \left( \frac{1-\sqrt{\rho_0}}{1+\sqrt{\rho_0}}\right)^{2h}e^{h \mathrm{\mathbf{t}}}\\
&\times \left(1 + \sum_{n=2}^{N-1} \mathrm{\mathbf{t}}^n \sum_{m=1}^{\lfloor\frac{n}{2}\rfloor}h^m f_{m,n}(s_0) + O\left(\mathrm{\mathbf{t}}^N \sum_{m=1}^{\lfloor\frac{N}{2}\rfloor}h^m f_{m,N}(s_0)\right)  \right),
\ea
where
\ba
f_{m,n}(s_0) = \sum_{k = (m-n)/2}^{(n-m)/2} a^{(m,n)}_k e^{k s_0},
\ea
for some symmetric coefficients $a^{(m,n)}_k=a^{(m,n)}_{-k} \in \mathbb{R}$. 
\label{qtermasym}
\end{lem}
\paragraph{Note:}As a finite sum of entire functions in $s_0$, $f_{m,n}(s_0)$ is an entire function in $s_0$ for all $n,m <\infty$. On the other hand, $\csch(s_0/2)$ has a $\sim1/s_0$ divergence around $s_0 \sim 0$, so we lose control of the asymptotic expansion in this limit. 

\begin{lem}
\label{cohstateasymexp}
Let $\De \in (0,\infty)$, $\mathrm{\mathbf{t}} = t \csch(s_0/2)$, and $N \in \mathbb{N}_{\geq2}$, then we have the following asymptotic expansion around $t = 0$:
\ba
\left(\frac{4 \rho }{(1-\rho)^2}\right)^{\De} =& \left( \frac{4 \rho_0}{(1-\rho_0)^2}\right)^{\De}e^{-t \De \coth(s_0/2)}\\
&\times\left(1 + \sum_{n=2}^{N-1} \mathrm{\mathbf{t}}^n \sum_{m=1}^{\lfloor\frac{n}{2}\rfloor}\De^m g_{m,n}(s_0) + O\left(\mathrm{\mathbf{t}}^N \sum_{m=1}^{\lfloor\frac{N}{2}\rfloor}\De^m g_{m,N}(s_0)\right)  \right),
\ea
where
\ba
g_{m,n}(s_0) = \sum_{k = (m-n+2)/2}^{(n-m-2)/2} b^{(m,n)}_k e^{k s_0},
\ea
for some symmetric coefficients $b^{(m,n)}_k  =b^{(m,n)}_{-k} \in \mathbb{R}$. 
\label{phitermasym}
\end{lem}
\section{OPE convergence and moment bounds}
\label{convergence}
In this appendix we present and prove a number of lemmas concerning convergence and moment bounds which will be necessary for proving our main results.

\begin{lem}[Convergence estimate]
\label{convestimate}
Let $\mu(\bs{\rho};\bs{h})$ denote the normalized OPE measure at $\bs{\rho}$. The integrated high energy tails of $h,\bh \gg \De$ satisfy the asymptotic bounds
\ba
\mathcal{G}(\bs{\rho})\int_{ h,\bh \geq \De^{2\alpha}} h^j \bh^k d\mu(\bs{\rho}; \bs{h}) \lesssim\De^{2\alpha(j+k)} |\bs{\rho}|^{\De^{2\alpha} ( 1 + g(\De) ) },
\ea
where $\alpha >1/2$ and $g(\De) = O(\De^{1-2\alpha}\log(\De))$.
\end{lem}
\begin{proof}
    The proof makes use of Tauberian theory, which is a well studied topic in both the physics~\cite{Qiao:2017xif,Mukhametzhanov:2018zja,vanRees:2024xkb} and mathematics~\cite{Korevaar2004} literature. We will simply quote the main results here. Consider the Laplace transform
    \ba
\mathcal{L}(s) = \int_0^\infty dF(t) e^{-s t},
    \ea
    where $dF(t)$ denotes a positive spectral density over $t$. If the Laplace transform has the asymptotic behavior as $s \to 0$
    \ba
    \mathcal{L}(s) \sim C(k) s^{-k},
    \ea
    for some constant $C(k)$ (which may depend on $k$), then the integrated spectral density has the large $T$ asymptotic
    \ba
    F(T) = \int_0^T dF(t) = \frac{C(k) T^k}{\Gamma(1+k)}\left( 1 + O\left( \frac{1}{\log(T)}\right) \right),
    \ea
    where the relative error term is vanishing for $T \gg k$. We can use this result to estimate the tail of the Laplace transform given by
    \ba
\mathcal{L}(T_0,s) = \int_{T_0}^\infty dF(t) e^{-s t}.
\ea
Integrating by parts, we have
\ba
\mathcal{L}(T_0,s) &=\lim_{t\to\infty} \left[e^{-s t}F(t)\right] - e^{-s T_0} F(T_0) + s \int_{T_0}^
\infty F(t) e^{-s t} dt\\
&\leq\lim_{t\to\infty} \left[e^{-s t}F(t)\right] + s \int_{T_0}^
\infty F(t) e^{-s t} dt,
\ea
where we used the positivity of the integrated spectral density to drop the second term in the first line, and the first term in the second line vanishes for all $s>0$ since $F(t)$ grows only as fast as a power in $t$. 

We can now compute the remaining term explicitly and expand around large $T_0 \gg k/s$, to find the asymptotic bound of
\ba
\mathcal{L}(T_0,s)  \lesssim \frac{C(k) e^{-s T_0} T_0^k}{\Gamma(k+1)} \quad \text{for  } T_0 \gg k/s.
\ea
To produce the lemma, we apply these estimates to the following Laplace transform:
\ba
\mathcal{G}(e^{-\bs{s}}) = \sum_{\bs{h}} a_{\bs{h}} e^{-\bs{s}\cdot \bs{h}},
\ea
which is dominated by the t-channel OPE singularity in the $\bs{s} \to \bs{0}$ limit. More generally, we want to take derivatives with respect to $(-s,-\bar{s})$ before expanding around small $\bs{s}$ to obtain convergence estimates for the tails of the OPE measure weighted by powers of $(h,\bar{h})$. 

This gives the asymptotic of 
\ba
\partial_{-s}^j \partial_{-\bar{s}}^k \mathcal{G}(e^{-\bs{s}\cdot \bs{h}}) \sim \frac{C_{j,k}(\De)}{ s^{2\De + j} \bar{s}^{2\De + k}}  , 
\ea
with
\ba
C_{j,k}(\De) = \frac{ 16^{\De} (-1)^{j+k} \Gamma^2 (1-2 \Delta )}{\Gamma (-j-2 \Delta +1)\Gamma (-k-2 \Delta +1)}.
\ea
Applying the Tauberian theorem then gives the integrated tail estimate of
\ba
\sum_{\bs{h} \geq \De^{2\alpha} } a_{\bs{h}} e^{-\bs{s}\cdot \bs{h}} h^j \bh^k \lesssim \frac{C_{j,k}(\De) e^{-(s+\bar{s})\De^{2\alpha}} \De^{2\alpha( 4\De + j+k)} }{ \Gamma(2\De+j +1) \Gamma(2\De + k+ 1)},
\ea
where $\alpha > 1/2$ so that $\De^{2\alpha} \gg \De$ as $\De \to\infty$. Fixing $(j,k) \in \mathbb{Z}_+^2$ and expanding the $\Gamma$-functions around $\De \to \infty$ gives
\ba
\frac{C_{j,k}(\De) e^{-(s+\bar{s})\De^{2\alpha}} \De^{2\alpha( 4\De + j+k)} }{ \Gamma(2\De+j +1) \Gamma(2\De + k+ 1)} = \Delta ^{2 \alpha  (4 \Delta +j+k)-4 \Delta -1} e^{4 \Delta -\Delta ^{2 \alpha } (s+\bar{s})}\left(\frac{1}{4\pi} + O\left(\frac{j+k}{\De} \right) \right).
\ea
Factoring out the $\De^{2\alpha( j +k)}/4\pi$ term, we bound the remaining factors by taking a logarithm and re-exponentiating:
\ba
\log( \Delta ^{4 (2 \alpha -1) \Delta -1 } e^{4 \Delta -\Delta ^{2 \alpha } (s+\bar{s})} ) &=(8 \alpha  \Delta -4 \Delta -1) \log (\Delta )+4 \Delta -\left(\Delta ^{2 \alpha } (s+\bar{s})\right)\\
& = -\De^{2\alpha}(s+\bar{s} + O(\De^{1-2\alpha}\log(\De))
\ea
\ba
\implies \Delta ^{4 (2 \alpha -1) \Delta -1  } e^{4 \Delta -\Delta ^{2 \alpha } (s+\bar{s})}  = |\bs{\rho}|^{\De^{2\alpha}(1 + g(\De))},
\label{expdecayest}
\ea
where $|\bs{\rho}| = e^{-(s+\bar{s})}$ with $s,\bar{s}>0$ and $g(\De) = O(\De^{1-2\alpha} \log(\De))$. Multiplying the estimate of eq.~(\ref{expdecayest}) by the $\De^{2\alpha(j+k)} /4\pi$ term which we had previously factored out gives
\ba
\sum_{\bs{h} \geq \De^{2\alpha} } a_{\bs{h}} e^{-\bs{s}\cdot \bs{h}} h^j \bh^k \lesssim \De^{2\alpha(j+k)} |\bs{\rho}|^{\De^{2\alpha}(1+g(\De))},
\ea
which, after restoring our standard notation, is the statement of the lemma.
\end{proof}

\begin{lem}[Coarse moment bounds in the causal diamond]
We denote the normalized moments of the OPE measure at base point $\bs{\rho}$ as
\ba
\omega_{j,k}(\bs{\rho}) = \int h^j \bh^k d\mu(\bs{h};\bs{\rho}).
\ea
For any $\epsilon > 0$ and $\bs{\rho} \in (0,1)^2$ we have
\ba
\omega_{j,k}(\bs{\rho}) = O(\De^{j+k + \epsilon}).
\ea
\end{lem}
\begin{proof}
    Split the sum over $\bs{h}$ into high and low energy sectors as
    \ba
\sum_{\bs{h}} a_{\bs{h}}\bs{\rho}^{\bs{h}} h^j \bh^k = \sum_{\bs{h} < \De^{2\alpha} } a_{\bs{h}}  \bs{\rho}^{\bs{h}} h^j \bh^k +\sum_{\bs{h} \geq \De^{2\alpha}} a_{\bs{h}} \bs{\rho}^{\bs{h}} h^j \bh^k .
    \ea
    The first term is bounded as
    \ba
\sum_{\bs{h} < \De^{2\alpha} } a_{\bs{h}} \bs{\rho}^{\bs{h}} h^j \bh^k <  \De^{2\alpha(j+k)}\sum_{\bs{h} < \De^{2\alpha} } a_{\bs{h}} \bs{\rho}^{\bs{h}} 
< \De^{2\alpha(j+k)} \mathcal{G}(\bs{\rho}).
    \ea
    Thus,
    \ba
\sum_{\bs{h}} a_{\bs{h}}\bs{\rho}^{\bs{h}} h^j \bh^k &< \De^{2\alpha(j+k)} \mathcal{G}(\bs{\rho}) \left( 1 + \frac{1}{\De^{2\alpha(j+k)} \mathcal{G}(\bs{\rho})}\sum_{\bs{h} \geq \De^{2\alpha}} a_{\bs{h}} \bs{\rho}^{\bs{h}} h^j \bh^k  \right)\\
&= \De^{2\alpha(j+k)} \mathcal{G}(\bs{\rho}) \left( 1 + O\left( |\bs{\rho}|^{\De^{2\alpha} (1+g(\De))}\right) \right),
    \ea
    where we have used lemma \ref{convestimate} and $\mathcal{G}(\bs{\rho}) > 1$ to bound the high energy sector. For $\alpha > 1/2$ and $|\bs{\rho}|<1$, the error term decays faster than $\De^{-N} \; \forall N \in\mathbb{N}$, so we can asymptotically bound it by an arbitrary positive constant. Dividing the result by $\mathcal{G}(\bs{\rho})$ to normalize and setting $\alpha = 1/2 + \frac{\epsilon}{2(j+k)}>1/2$ with $\epsilon > 0$ gives the lemma.
\end{proof}

\begin{lem}[Sharp bound on the first moment in the causal diamond]
\label{sharpmomentbound}
The first normalized moments satisfy the sharp bounds
\ba
\omega_{1,0}(\bs{\rho}) \leq \De \frac{1+\rho}{1-\rho} \quad \text{and} \quad \omega_{0,1}(\bs{\rho}) \leq \De \frac{1+\rhob}{1-\rhob} \quad \forall \bs{\rho} \in (0,1)^2.
\ea
\end{lem}
\begin{proof}
The moments satisfy $\omega_{j,k}(\bs{\rho}) = \omega_{k,j}(\bar{\bs{\rho}})$, so it suffices to check for $(j,k) = (1,0)$ as long as the bound holds for all $\bs{\rho} \in (0,1)^2$. We write $\bs{\rho} = (e^{-\bs{s}}) = (e^{-s},e^{-\bar{s}})$ and take the logarithm of the crossing equation
\ba
\log(\mathcal{G}(\bs{s})) = \De \log(|\bs{\chi}(\bs{s})|) + \log( \mathcal{G}( \hat{\bs{s}} )),
\ea
where 
\ba
\hat{\bs{s}} = \left(- \log(\hat{\bs{\rho}}(\bs{s}))\right) = \left(- 2 \log\left(\frac{1-\sqrt{e^{-s}}}{1+\sqrt{e^{-s}}}\right) ,- 2 \log\left(\frac{1-\sqrt{e^{-\bar{s}}}}{1+\sqrt{e^{-\bar{s}}}}\right) \right)
\ea
and 
\ba
|\bs{\chi}(\bs{s})| = \left(  \frac{16 e^{-(s+\bar{s})}}{(1-e^{-s})^2 (1-e^{-\bar{s}})^2}\right).
\ea
We now take a derivative of both sides with respect to $-s$ and evaluate at $\bs{s} = \left(-\log(\bs{\rho}) \right)$ to compute the first moment $\omega_{1,0}(\bs{\rho})$ on the LHS
\ba
\omega_{1,0}(\bs{\rho}) = \De \partial_{-s} \log(|\bs{\chi}(\bs{s})|)|_{\bs{s} = -\log(\bs{\rho})} + \left(\partial_{-s} \hat{\rho}(s) \right)|_{s = -\log(\rho)} \omega_{1,0}(\hat{\bs{\rho}}).
\ea
By the positivity of the support of the OPE measure, $\omega_{1,0}(\bs{\rho})$ is positive for all points in the causal diamond, as is $\omega_{1,0}(\hat{\bs{\rho}})$. It is also easy to check that $\partial_{-s} \hat{\rho}(s) = \frac{2 \sqrt{e^{-s}}}{e^{-s} -1} \leq 0$ for all points in the causal diamond, so the second term on the RHS is always negative. Thus, we have
\ba
\omega_{1,0}(\bs{\rho}) &= \De \partial_{-s} \log(|\bs{\chi}(\bs{s})|)|_{\bs{s} = -\log(\bs{\rho})} + \left(\partial_{-s} \hat{\rho}(s) \right)|_{s = -\log(\rho)} \omega_{1,0}(\hat{\bs{\rho}})\\& \leq \De \partial_{-s} \log(|\bs{\chi}(\bs{s})|)|_{\bs{s} = -\log(\bs{\rho})}  = \De \frac{1+\rho}{1-\rho}.
\ea
\end{proof}

\begin{lem}[Uniform gradient bound]
\label{uniformgradbound}
For all $\bs{x} \in \mathbb{R}^2$ and $n \in \mathbb{N}$,
\ba
\bs{0}\leq \nabla\lambda_n(\bs{x}) \leq \bs{1} .
\ea
\end{lem}
\begin{proof}
Since $\lambda_n(\bs{x}) = \lambda_n(\bar{\bs{x}})$ by chiral symmetry, it suffices to bound $\partial_x \lambda_n(\bs{x})$. We compute the following:
\ba
\partial_x \lambda_n(\bs{x}) = \frac{1}{\De_n} \partial_x \log(\mathcal{G}_n(\bs{x})) &= \frac{1}{\De_n} \frac{1- \rho}{1+\rho} \omega_{1,0}(\bs{\rho})\\
&\leq \frac{1}{\De_n} \frac{1- \rho}{1+\rho} \De_n \frac{1+\rho}{1-\rho} = 1,
\ea
where we obtained the second line by applying lemma \ref{sharpmomentbound}. Since the bound is independent of $\bs{x}$, we have both $ \partial_{x} \lambda_n(\bs{x}) , \partial_{\bar{x}} \lambda_n(\bs{x}) \leq 1$. Defining $\left(\partial_{x} \lambda_n(\bs{x}) , \partial_{\bar{x}} \lambda_n(\bs{x})  \right) \equiv \nabla\lambda_n(\bs{x})$  gives the upper bound of the lemma. For the lower bound, we simply have that $\mathcal{G}(\bs{x})$ increases monotonically in $x,\bar{x}$ independently, thus its logarithm does as well, so its gradient is positive. 
\end{proof}

\begin{lem}[Second moment bound at the self-dual point]
\label{sdmomentbounds}
Let $\omega_n$ denote the pure moments of the OPE measure at the self-dual point with external dimension $\De > 0$. Then
\ba
\omega_1 = \De/\sqrt{2}, \qquad \frac{\De^2}{2}\leq \omega_2 \leq \omega_2^{(+)} = \frac{1}{8}\left(6\De^2 + 3\De + \sqrt{\De^2(4\De(\De+3)+13)}\right).
\ea
In particular, $\omega_2^{(+)}/\De^2 = 1 + \frac{3}{4\De} + O(1/\De^2)$ as $\De \to \infty$.
\end{lem}
 
\begin{proof}
The support of the OPE measure lies in $\mathbb{R}^2_+$ and the radial monomial decomposition is exponentially convergent~\cite{Pappadopulo:2012jk}, so the moment problem is determinate and the Stieltjes conditions $H^{(k)}_l\equiv \left(\omega_{k+i+j}\right)_{0 \leq i,j \leq l} \succeq 0$ for $k = 0,1$ and all $l \in\mathbb{N}$ hold. Write the crossing equation with $\rho$ free and $\rhob = \rho_* = 3-2\sqrt{2}$:
\ba
\int (\rho/\rho_*)^h \,d\mu(\rho_*;h) = \left(\frac{4\rho}{(1-\rho)^2}\right)^{\!\De} \int \left(\hat{\rho}/\rho_* \right)^h d\mu(\rho_*;h),
\ea
where $\hat{\rho} = \left(\frac{1-\sqrt{\rho}}{1+\sqrt{\rho}}\right)^2$. We now differentiate the crossing equation $\Lambda$ times and evaluate at $\rho = \rho_*$. Since $\hat{\rho}(\rho_*) = \rho_*$ and $\hat{\rho}'(\rho_*) = -1$, only odd-$\Lambda$ derivatives yield nontrivial constraints. $\Lambda = 1$ gives $\omega_1 = \De/\sqrt{2}$, and $\omega_2 \geq \omega_1^2 = \De^2/2$ by Jensen's inequality. The $\Lambda = 3$ constraint gives
\ba
\omega_3 = \frac{\De - 2\De^2(4\De+3)}{8\sqrt{2}} + \frac{3(2\De+1)}{2\sqrt{2}}\,\omega_2.
\ea
Combining with $\det H^{(1)}_1 \geq 0$, i.e.\ $\omega_3 \geq \omega_2^2/\omega_1$, yields $-16(\omega_2 - \omega_2^{(-)})(\omega_2 - \omega_2^{(+)}) \geq 0$. Since $\omega_2^{(-)} < \De^2/2$ violates Jensen's inequality for all $\De > 0$, $(\omega_2 - \omega_2^{(-)}) > 0$, so the effective constraint is $\omega_2 \leq \omega_2^{(+)}$.
\end{proof}

\subsection{Uniform MGF bound and coherent state approximation}
\label{mgfsection}
 
The gradient bound on the dynamical free energy density (lemma~\ref{uniformgradbound}) directly implies a uniform bound on the moment generating function (MGF) of the rescaled OPE measure. Define $s: [0,\infty) \to [0, s_{\max})$ by $s(x) = \log(\rho(x)/\rho_*)$, where $\rho(x)$ is the radial coordinate as a function of $x = \log(z/(1-z))$ and $s_{\max} = \log(1/\rho_*) < \infty$. Then $s$ is smooth and strictly increasing with $s(0) = 0$, and has a smooth inverse:
\ba
s^{-1}(u) &= \log \left(4 \sqrt{2} \sinh (u)+6 \cosh (u)-2\right)-2 \log (3-\cosh (u))\quad \forall u < s_{\max}.
\ea

\begin{prop}[Uniform MGF bound]
\label{mgfbound}
Let $\{\nu_n\}$ denote the rescaled OPE measure sequence at the self-dual point. For all $n \in \mathbb{N}$ and $c,\bar{c} \in [0, \De_n\, s_{\max})$,
\ba
\int e^{\bs{c} \cdot \bs{\eta}}\, d\nu_n(\bs{\eta}) \leq e^{\De_n\left(s^{-1}(c/\De_n) + s^{-1}(\bar{c}/\De_n)\right)}.
\ea
\end{prop}
 
\begin{proof}
The gradient bound of lemma \ref{uniformgradbound}, integrated from $\bs{0}$, gives $\log\mathcal{G}_n(x,\bar{x}) \leq \De_n(x + \bar{x})$ for $x, \bar{x} \geq 0$. Exponentiating the inequality and changing variables $\bs{h} = \De_n\bs{\eta}$ gives
\ba
\int e^{\De_n(\eta\, s(x) + \bar{\eta}\, s(\bar{x}))}\, d\nu_n(\bs{\eta}) \leq e^{\De_n(x+\bar{x})} \quad \forall\, x, \bar{x} \geq 0.
\ea
Setting  $\bs{x} = ( s^{-1}(c/\De_n),s^{-1}(\bar{c}/\De_n)) $  gives the desired bound. 
\end{proof}
 \begin{cor}
     Expanding $s^{-1}(u)$ around small $u$ gives the following asymptotic bound for any fixed $c,\bar{c} \geq 0$ and $\De_n \gg 0$:
     \ba
        \int e^{\bs{c} \cdot \bs{\eta}} d\nu_n(\bs{\eta}) \leq e^{\sqrt{2}(c+\bar{c}) + O(1/\De_n) }.
     \ea
 \end{cor}
\begin{rem}
The coefficient $\sqrt{2}$ is sharp: it agrees with the support bound $\mathrm{supp}(\nu) \subseteq [0,\sqrt{2}]^2$ derived below in appendix \ref{tightness}.
\end{rem}

\begin{lem}[Coherent state approximation]\label{cohstateapprox}
    For any $\alpha>0$, $\bs{x} \in \mathbb{R}$, and $s_n = \De_n^{1-\alpha}$,
    \ba
\frac{1}{s_n} \log \mathcal{G}_n(\bs{x}/\De_n^\alpha) = \frac{1}{s_n}\log M_{\nu_n}(s_n \bs{\theta}) + O(\De_n^{-\alpha}),
    \ea
    where $\bs{\theta} = \bs{x}/\sqrt{2}$ and $M_{\nu_n}(\bs{c}) = \int e^{\bs{c} \cdot \bs{\eta}} d\nu_n(\bs{\eta})$.
\end{lem}
\begin{proof}
    The radial monomial decomposition with $\bs{h} = \De_n \bs{\eta}$ gives the identity
    \ba
 \log M_{\nu_n}(\bs{c}) = \De_n \lambda_n\left(s^{-1}(c/\De_n),s^{-1}(\bar{c}/\De_n) \right),
    \ea
    where $s(x) = \log(\rho(x)/\rho_*)$. Setting $\bs{c} = s_n \bs{\theta}$ and using $s^{-1}(u) = \sqrt{2} u + O(u^2)$ gives
    \ba
\log M_{\nu_n}(s_n \bs{\theta}) = \De_n \lambda_n\left( x/\De_n^\alpha + O(x^2/\De_n^{2\alpha}), \bar{x}/\De_n^\alpha + O(\bar{x}^2/\De_n^{2\alpha}) \right).
    \ea
    Additionally,  $\log \mathcal{G}_n(\bs{x}/\De_n^\alpha) = \De_n \lambda_n(\bs{x}/\De_n^\alpha)$. By the mean value theorem, there exists a $\bs{\xi} \in \mathbb{R}^2$ such that
    \ba
\log \mathcal{G}_n(\bs{x}/\De_n^\alpha) - \log M_{\nu_n}(s_n \bs{\theta})  = \De_n \nabla \lambda_n\left(\bs{\xi}) \cdot ( \bs{x}/\De_n^\alpha - s^{-1}(\bs{\theta} \De_n^{1-2\alpha}) \right) ,
    \ea
where $||( \bs{x}/\De_n^\alpha - s^{-1}(\bs{\theta} \De_n^{1-2\alpha}) )|| = O(\bs{x}^2/\De_n^{2\alpha})$. Taking an absolute value of both sides, and bounding the RHS with Cauchy-Schwarz and the gradient bound of lemma \ref{uniformgradbound} gives
\ba
|\log \mathcal{G}_n(\bs{x}/\De_n^\alpha) - \log M_{\nu_n}(s_n \bs{\theta})| &\leq \De_n || \nabla \lambda_n(\bs{\xi}) ||  \; || ( \bs{x}/\De_n^\alpha - s^{-1}(\bs{\theta} \De_n^{1-2\alpha}) )||
\\&=O(\bs{x}^2 \De_n^{1-2\alpha}).
\ea
Dividing both sides of this inequality by $s_n$ gives the lemma.
\end{proof}

\section{Tightness and compactness}
\label{tightness}

In this appendix, we prove that the rescaled measure sequences $\{\nu_n(\bs{\eta}) = \mu_n(\De_n \bs{\eta})\}$ defined at the self-dual point are tight, and that any weak limit $\nu = \lim_{n\to\infty} \nu_n$ is compactly supported on $\mathcal{D} = [0,\sqrt{2}]^2$.

\begin{prop}[Tightness of rescaled measure sequences]
\label{prop:tight}
Let $\{\nu_n\}_{n\in\mathbb{N}}$ denote the rescaled measure sequence at the self-dual point. This measure sequence is tight. That is, for any $\epsilon > 0$, there exists a compact set $K_\epsilon \subset \mathbb{R}_+^2$ such that
\ba
\nu_n(K_\epsilon^c) < \epsilon\quad \forall n \in \mathbb{N}.
\ea
\end{prop}
\begin{proof}
Lemma \ref{sdmomentbounds} gives $\omega_2 \leq \omega_2^{(+)}$, where $\omega_2^{(+)}/\De_n^2 = 1 + O(1/\De_n)$ as $\De_n \to \infty$. In particular, there exists a constant $C > 0$ such that $\omega_2^{(+)}/\De_n^2 \leq C$ for all $n\in\mathbb{N}$. Since $\nu_n(\bs{\eta}) = \mu_n(\De_n\bs{\eta})$, the second moment of $\nu_n$ satisfies
\ba
\int \eta^2\, d\nu_n(\bs{\eta}) = \frac{\omega_2}{\De_n^2} \leq \frac{\omega_2^{(+)}}{\De_n^2} \leq C,
\ea
and by chiral symmetry, $\int \bar{\eta}^2\, d\nu_n(\bs{\eta}) \leq C$ as well. Now fix $\epsilon > 0$ and let $R = 2\sqrt{C/\epsilon}$. Define the compact set $K_\epsilon = [0,R]^2 \subset \mathbb{R}^2_+$. Since $K_\epsilon^c \cap \mathbb{R}^2_+ \subset \{\eta > R\} \cup \{\bar{\eta} > R\}$, a union bound followed by Markov's inequality gives
\ba
\nu_n(K_\epsilon^c) \leq \nu_n(\eta > R) + \nu_n(\bar{\eta} > R) \leq \frac{1}{R^2}\int \eta^2\, d\nu_n + \frac{1}{R^2}\int \bar{\eta}^2\, d\nu_n \leq \frac{2C}{R^2} = \frac{\epsilon}{2} < \epsilon
\ea
for all $n \in \mathbb{N}$, which establishes tightness.
\end{proof}

\begin{prop}[Compact support of classical measures]
\label{prop:compactsupport}
Let $\nu$ be any weak limit of $\{\nu_n\}$. Then $\mathrm{supp}(\nu) \subset \mathcal{D} = [0,\sqrt{2}]^2$.
\end{prop}
\begin{proof}
Since each $\nu_n$ is supported on $\mathbb{R}^2_+$ and the weak limit of measures supported on a closed set is supported on that set, $\mathrm{supp}(\nu) \subset \mathbb{R}^2_+$. It remains to show $\nu(\{\eta > \sqrt{2}\}) = 0$ and $\nu(\{\bar{\eta} > \sqrt{2}\}) = 0$. By chiral symmetry, it suffices to prove the former.
 
Fix $c > 0$ and $R > 0$, and define the continuous bounded function $f_R(\bs{\eta}) = \min(e^{c\eta}, e^{cR})$. Since $f_R \leq e^{c\eta}$, the uniform MGF bound (proposition~3.3) with $\bar{c} = 0$ gives, for all $n$ sufficiently large,
\ba
\int f_R\, d\nu_n \leq \int e^{c\eta}\, d\nu_n(\bs{\eta}) \leq e^{\sqrt{2}\, c + O(1/\De_n)}.
\ea
Since $f_R$ is continuous and bounded, weak convergence $\nu_n \to \nu$ gives 
\ba
\int f_R\, d\nu = \lim_{n\to\infty}\int f_R\, d\nu_n \leq e^{\sqrt{2}\, c}.
\ea
Sending $R \to \infty$ and applying the monotone convergence theorem, we obtain the MGF bound for the limiting measure:
\ba
\int e^{c\eta}\, d\nu(\bs{\eta}) \leq e^{\sqrt{2}\, c} \quad \forall\, c > 0.
\ea
Now fix $\delta > 0$. Applying the exponential Markov inequality (Chernoff bound), we have
\ba
\nu(\eta \geq \sqrt{2} + \delta) \leq e^{-c(\sqrt{2}+\delta)} \int e^{c\eta}\, d\nu \leq e^{-c\delta} \quad \forall c > 0.
\ea
Sending $c \to \infty$ gives $\nu(\eta \geq \sqrt{2} + \delta) = 0$. As this holds for all $\delta > 0$, we conclude $\nu(\eta > \sqrt{2}) = 0$. By an identical argument applied to $\bar{\eta}$, $\nu(\bar{\eta} > \sqrt{2}) = 0$, hence $\mathrm{supp}(\nu) \subset [0,\sqrt{2}]^2$.
\end{proof}

\section{Proof of proposition \ref{preaa}}
\label{proofofpreaa}

We use lemma \ref{uniformgradbound} to prove compact boundedness and uniform equicontinuity as follows:

\begin{lem}[Compact boundedness]\label{compactbound}
Let $A \subset \mathbb{R}^2$ be compact, then
\ba
|\lambda_n(\bs{x})| \leq \sup_{\bs{x} \in A} ||\bs{x}||_1 \equiv M_A,
\ea
where $M_A $ is bounded and does not depend on $n$ or $\bs{x} \in A$.
\end{lem}
\begin{proof}
Since $\mathcal{G}_n$ is normalized to $1$ at $\bs{x} = 0$, $\lambda_n(\bs{0}) = 0$ for all $n$. Applying the mean value theorem (MVT) along the $x,\bar{x}$ axes we have

\ba
|\lambda_n(\bs{x})| &= |\lambda_n(\bs{x}) - \lambda_n(0,\bar{x}) + \lambda_n(0,\bar{x}) - \lambda_n(\bs{0}) | \\
&\leq |\lambda_n(\bs{x}) - \lambda_n(0,\bar{x}) | + |\lambda_n(0,\bar{x}) - \lambda_n(\bs{0})|\\
&\leq \partial_x \lambda_n(\xi,\bar{x}) |x| + \partial_{\bar{x}} \lambda_n(0,\bar{\xi}) |\bar{x}|\\
& \leq |x| + |\bar{x}| \leq \sup_{\bs{x}\in A} || \bs{x}||_1 .
\ea
Since $|| \bs{x}||_1$ is continuous in $\bs{x}$, it is bounded on any compact subset $A$. Thus, we have $\sup_{\bs{x}\in A} || \bs{x}||_1 \equiv M_A<\infty$.
\end{proof}

\begin{lem}[Uniform equicontinuity] \label{uniformequi}
For any $\bs{x},\bs{y} \in \mathbb{R}^2$ and $\epsilon > 0$, we have
\ba
|\lambda_n(\bs{x}) - \lambda_n(\bs{y})|  < \epsilon
\ea
when $||\bs{x} - \bs{y}|| < \delta \equiv \epsilon /\sqrt{2} $.
\end{lem}
\begin{proof}
Once again, we apply MVT:
\ba
|\lambda_n(\bs{x}) - \lambda_n(\bs{y})|  &= \nabla\lambda_n(\bs{\eta})\cdot(\bs{x}-\bs{y})\\
&\leq |\partial_x \lambda_n| |x-y| + |\partial_{\bar{x}}\lambda_n| | \bar{x} - \bar{y}|\\
& \leq |x-y| + |\bar{x}-\bar{y}| \leq \sqrt{2} || \bs{x}-\bs{y}||,
\ea
where we used Cauchy-Schwarz for the last line. Setting $\delta = \epsilon / \sqrt{2}$ completes the proof.
\end{proof}
Combining lemma \ref{compactbound} and \ref{uniformequi} gives the content of proposition \ref{preaa}.

\section{Proof of proposition \ref{universalratefunctionbound}}
\label{proofofuniversalratefuncbound}

For the upper bound, we bound each quadrant for $\bs{x} \in \mathbb{R}^2$. Starting with $x,\bar{x} \leq0$ (or $\bs{x} \in \mathbb{R}^{--}$), we notice that $(\rho/\rho_*), (\rhob/\rho_*) \leq 1$. Thus, we have
\ba
\mathcal{G}(\bs{\rho}(\bs{x})) \leq \int d\mu(\bs{\rho}_*,\bs{h}) =1,
\ea
so that $\lambda(\bs{x}) \leq 0$ for $\bs{x} \in \mathbb{R}^{--}$.

To bound the mixed quadrants $\mathbb{R}^{+-}$ and $\mathbb{R}^{-+}$ it suffices to bound $\mathbb{R}^{+-}$ and then use chiral symmetry to get the bound on $\mathbb{R}^{-+}$. We note that $\lambda(\bs{x}) \leq \lambda( x,0)$ for $\bs{x} \in\mathbb{R}^{+-}$. We write
\ba
d\lambda(x',0) = \partial_{x'}\lambda(x',0) dx'
\ea
and integrate both sides from $x' =0$ to $x$ to obtain
\ba
\lambda(x,0) - \lambda(\bs{0}) = \lambda(x,0) = \int_0^x\partial_{x'}\lambda(x',0) dx' \leq x,
\ea
where we have used the fact that $\lambda(\bs{0}) =0$ via our normalization and our gradient bound to get the far right inequality. In total, this gives $\lambda(\bs{x}) \leq x$ for $\bs{x} \in \mathbb{R}^{+-}$ and $\lambda(\bs{x}) \leq \bar{x}$ for $\bs{x} \in \mathbb{R}^{-+}$. For the final quadrant $\mathbb{R}^{++}$, we write
\ba
d\lambda(\bs{x}) = \partial_x \lambda(\bs{x}) dx +  \partial_{\bar{x}} \lambda(\bs{x}) d\bar{x}
\ea
and perform the same procedure to obtain $\lambda(\bs{x}) \leq x+\bar{x}$ for $\bs{x} \in \mathbb{R}^{++}$. Combining the results for each quadrant gives the upper bound in the proposition.

For the lower bound, consider the radial monomial decomposition of the correlator normalized to $1$ at the self-dual point
\ba
\mathcal{G}(\bs{\rho}) = \int (\rho/\rho_*)^h (\rhob/\rho_*)^{\bh} d\mu(\bs{\rho}_*;\bs{h}).
\ea
For all $\bs{\rho} \in \mathscr{C}$, the monomial factors are convex functions for all $\bs{h} \in \mathbb{R}^2_+$. Thus, Jensen's inequality applies, and we have
\ba
\mathcal{G}(\bs{\rho}) \geq (\rho/\rho_*)^{\omega_{1,0}(\bs{0})} (\rhob/\rho_*)^{\omega_{0,1}(\bs{0})}.
\ea
Plugging in the universal result that $\omega_{1,0}(\bs0)= \omega_{0,1}(\bs0) = \De/\sqrt{2}$ and computing the rate function, we obtain the bound
\ba
\lambda(\bs{x}) \geq f(x) + f(\bar{x}) \quad \forall \bs{x} \in \mathbb{R}^2,
\ea
where $f(x) = \log(\rho(x) /\rho_*)/\sqrt{2}$. 

We can improve this bound by imposing crossing symmetry:
\ba
\lambda(\bs{x}) = (x+\bar{x}) + \lambda(-\bs{x}) \geq (x+\bar{x}) + f(-x) + f(-\bar{x}).
\ea
For the bounds to be compatible, we then have
\ba
\lambda(\bs{x}) \geq \max(f(x) + f(\bar{x}) , (x+\bar{x}) + f(-x) + f(-\bar{x}) ).
\ea
This gives our lower bound in the region around the self-dual point. To derive tighter bounds far away from the self dual point, we use that the unnormalized correlator is bounded from below by $1$ due to the contribution of the s-channel identity operator. This gives
\ba
\lambda(\bs{x}) &= \lim_{n \to \infty} \frac{1}{\De_n} \left( \log\left( \widehat{\mathcal{G}}_n(\bs{x} ) \right) - \log\left(\widehat{\mathcal{G}}_n(\bs{0}) \right)\right)\\
&\geq  - \lim_{n\to\infty} \frac{1}{\De_n}\log\left( \widehat{\mathcal{G}}_n(\bs{0})\right)  = - \Sigma \quad \forall \bs{x} \in \mathbb{R}^2 ,
\ea
where $\widehat{\mathcal{G}}$ denotes the raw, unnormalized correlator. Again, we can strengthen this bound by imposing crossing symmetry:
\ba
\lambda(\bs{x}) = (x+\bar{x}) + \lambda(-\bs{x}) \geq (x+\bar{x}) - \Sigma.
\ea
Taking the maximum of all lower bounds gives the full lower bound $\lambda_-$ in the proposition.

\section{Proof of proposition \ref{convolutioncond}}
\label{convolutioncondproof}

Using that $||f||_{L^1(\mathbb{R}^2)} \geq | \mathcal{F}^{\pm}[f]|$ for any $L^1$ integrable function $f$ on $\mathbb{R}^2$, we have
    \ba
    \| \bs{\Lambda}_{\De_n^{1 - \alpha}}*(\mu_n'-\mu_{\De_n}') \|_{L^1(\mathbb{R}^2)} &\geq | \mathcal{F}^{-1}\left[ \bs{\Lambda}_{\De_n^{1 - \alpha}}*(\mu_n'-\mu_{\De_n}')  \right](\bs{\xi}) |\\
    & = e^{-2 \pi^2 \De_n^{2 -2\alpha} \bs{\xi}^2} |  \phi_{\mu_n}(\bs{\xi}) - \phi_{\mu_{\De_n}}(\bs{\xi}) |\\
    & =  e^{-2 \pi^2 \De_n^{2 - 2\alpha} \bs{\xi}^2} | \phi_{\nu_n}(\De_n \bs{\xi}) - \phi_{\nu}(\De_n\bs{\xi})  |,
    \ea
where we used theorem \ref{convthm} to obtain the second line. Now, write $\bs{\xi} = \tilde{\bs{t}}\De_n^{-1 + \beta} $, which gives
\ba
\| \bs{\Lambda}_{\De_n^{1 - \alpha}}*(\mu_n'-\mu_{\De_n}') \|_{L^1(\mathbb{R}^2)} \geq  e^{-2 \pi^2 \De_n^{2(\beta -\alpha)} \tilde{\bs{t}}^2}|  \phi_{\nu_n}(\De_n^{\beta} \tilde{\bs{t}}) -\phi_{\nu}(\De_n^{\beta} \tilde{\bs{t}})  |.
\ea
For any fixed $\tilde{\bs{t}} \in \mathbb{R}^2$ and $\beta <\alpha$, the Gaussian prefactor monotonically increases to $1$ as $n \to \infty$, so the vanishing of the $L^1(\mathbb{R}^2)$ norm implies the vanishing of $| \phi_{\nu_n}(\De_n^{\beta} \tilde{\bs{t}}) -\phi_{\nu}(\De_n^{\beta} \tilde{\bs{t}})  |$ as $n \to \infty$. Finally, bounding 
\ba
-|\phi_{\nu_n}(\De_n^{\beta} \tilde{\bs{t}}) -\phi_{\nu}(\De_n^{\beta} \tilde{\bs{t}})  | &\leq \left(\phi_{\nu_n}(\De_n^{\beta} \tilde{\bs{t}}) -\phi_{\nu}(\De_n^{\beta} \tilde{\bs{t}}) \right)  \\&\leq |\phi_{\nu_n}(\De_n^{\beta} \tilde{\bs{t}}) -\phi_{\nu}(\De_n^{\beta} \tilde{\bs{t}})  |
\ea 
and applying the squeeze theorem gives eq.~(\ref{pointwiseasymptoticconvergence}).

\section{Proof of theorem \ref{cohstatedecompthm}}
\label{cohstatedecompthmproof}

The first part of the theorem follows from the definition of coherent states. For the second part, we first need to verify pointwise convergence between $\phi_\nu(\De_n^\beta \tilde{\bs{t}})$ and $\tilde{\mathcal{G}}_{\De_n}(\De_n^{\beta}\tilde{\bs{t}})$. We have
\ba
\phi_\nu(\De_n^\beta\tilde{\bs{t}}) = \int e^{ i 2\pi \tilde{\bs{t}} \cdot \bs{\eta} \De_n^\beta} d\nu(\bs{\eta})
\ea
and
\ba
\tilde{\mathcal{G}}_{\De_n}(\De_n^{\beta}\tilde{\bs{t}}) &= \int \chi(\De_n^\beta \tilde{t})^{h/\sqrt{2}}  \chi(\De_n^\beta \tilde{\bar{t}})^{\bar{h}/\sqrt{2}} d\mu_{\De_n}(\bs{h})\\ &= \int \chi(\De_n^\beta \tilde{t})^{\De_n\eta/\sqrt{2}}  \chi(\De_n^\beta \tilde{\bar{t}})^{\De_n \bar{\eta}/\sqrt{2} } d\nu(\bs{\eta}).
\label{cohstatedecomprewrite}
\ea
From lemma \ref{cohstateapprox}, we approximate
\ba
\chi(\De_n^\beta \tilde{t})^{\De_n \eta/\sqrt{2}} = e^{i2\pi \tilde{t} \eta \De_n^{\beta}} \left(1 + O\left( \De_n^{2\beta - 1} \eta \tilde{t}^2\right) \right).
\ea
Plugging this into eq.~(\ref{cohstatedecomprewrite}), and using the compactness of $\nu$, we have
\ba
\tilde{\mathcal{G}}_{\De_n}(\De_n^{\beta}\tilde{\bs{t}}) &= \int e^{ i 2\pi \tilde{\bs{t}} \cdot \bs{\eta} \De_n^\beta} d\nu(\bs{\eta}) + O\left(\De_n^{2\beta -1} \tilde{\bs{t}}^2\right) \\
&= \phi_\nu(\De_n^\beta \tilde{\bs{t}}) + O\left(\De_n^{2\beta -1} \tilde{\bs{t}}^2\right) 
\ea
for all $\tilde{\bs{t}} \in \mathbb{R}^2$. Thus for any fixed $\tilde{\bs{t}} \in \mathbb{R}^2$ and $\beta < 1/2$ we have
\ba
\lim_{n\to \infty} | \tilde{\mathcal{G}}_{\De_n}(\De_n^{\beta}\tilde{\bs{t}}) - \phi_\nu(\De_n^\beta \tilde{\bs{t}})|  = 0.
\label{cohstatevspullback}
\ea
Now, write
\ba
|\mathcal{G}_n(\De_n^\beta \tilde{\bs{t}}) - \tilde{\mathcal{G}}_{\De_n}(\De_n^{\beta}\tilde{\bs{t}}) | &= |\mathcal{G}_n(\De_n^\beta \tilde{\bs{t}}) - \phi_\nu(\De_n^\beta \tilde{\bs{t}}) + \phi_\nu(\De_n^\beta \tilde{\bs{t}}) - \tilde{\mathcal{G}}_{\De_n}(\De_n^{\beta}\tilde{\bs{t}}) |\\
&\leq |\mathcal{G}_n(\De_n^\beta \tilde{\bs{t}}) - \phi_\nu(\De_n^\beta) | + |\phi_\nu(\De_n^\beta) - \tilde{\mathcal{G}}_{\De_n}(\De_n^{\beta}\tilde{\bs{t}}) |.
\ea
Since the convolution condition holds true for $\alpha = 1/2$, the first term on the RHS vanishes as $n \to \infty$ for $\beta < 1/2$, and by eq.~(\ref{cohstatevspullback}) the second term on the RHS also vanishes as $n \to \infty$ for $\beta < 1/2$. Therefore
\ba
\lim_{n \to \infty} |\mathcal{G}_n(\De_n^\beta \tilde{\bs{t}}) - \tilde{\mathcal{G}}_{\De_n}(\De_n^{\beta}\tilde{\bs{t}}) | = 0 \quad \forall \tilde{\bs{t}} \in \mathbb{R}^2 \;\text{and}\; \beta < 1/2.
\ea
Applying the squeeze theorem to 
\ba
-|\mathcal{G}_n(\De_n^\beta \tilde{\bs{t}}) - \tilde{\mathcal{G}}_{\De_n}(\De_n^{\beta}\tilde{\bs{t}})|&\leq\left( \mathcal{G}_n(\De_n^\beta \tilde{\bs{t}}) - \tilde{\mathcal{G}}_{\De_n}(\De_n^{\beta}\tilde{\bs{t}})\right) \\&\leq|\mathcal{G}_n(\De_n^\beta \tilde{\bs{t}}) - \tilde{\mathcal{G}}_{\De_n}(\De_n^{\beta}\tilde{\bs{t}})| 
\ea 
gives eq.~(\ref{cohstatevsexact}).

\section{Proof of theorem \ref{matchinglocalratefunctions}}
\label{proofofmatchinglocalratefunctions}

We note that for $\beta \geq1$, matching is always guaranteed, as we have pointwise convergence of both the true and coherent state correlators in this region. Thus we can just focus on the case where $0<\beta<1$. Let $s_n = \De_n^{1-\alpha}$ and $\bs{\theta} = \bs{x}/\sqrt{2}$. For the forward direction that the $(\star)$ condition implies matching, it suffices to prove
\ba
\lim_{n\to\infty} \frac{1}{s_n} \log \int e^{s_n \bs{\theta}\cdot \bs{\eta}}d\nu_n(\bs{\eta}) = h_K(\bs{\theta}) = \sup_{\bs{\eta}\in K} \bs{\theta} \cdot \bs{\eta},
\label{laplaceintegralasym}
\ea
as the remaining step of passing from the true correlator to the above Laplace integral uses only the coherent state approximation of radial monomials. Lemma~\ref{cohstateapprox} and definition~\ref{localratefuncdef} give
\ba
\frac{1}{s_n}\log\mathcal{G}_n(\bs{x}/\De_n^\alpha) = \frac{1}{s_n}\log\int e^{s_n\bs{\theta}\cdot\bs{\eta}}\,d\nu_n(\bs{\eta}) + O(\De_n^{-\alpha}).
\ea
The error term decays as $n \to \infty$ for any $\alpha > 0$. 

We now just need to compute the limit of the remaining term to prove eq.~(\ref{laplaceintegralasym}). We will do so by deriving an upper and lower bound and applying the squeeze theorem. To obtain the upper bound, fix $\delta> 0$ and decompose
\ba
\int e^{s_n \bs{\theta}\cdot\bs{\eta} }d\nu_n = \int_{K_\delta} e^{s_n \bs{\theta}\cdot \bs{\eta}} d\nu_n + \int_{K_\delta^c} e^{s_n \bs{\theta}\cdot\bs{\eta}} d\nu_n.
\ea
In the first integral on the RHS, we have 
\ba
\bs{\theta} \cdot\bs{\eta} \leq h_K(\bs{\theta}) + ||\bs{\theta}||\delta \quad \forall \bs{\eta} \in K_\delta,
\ea
giving
\ba
\int_{K_\delta} e^{s_n \bs{\theta}\cdot \bs{\eta}} d\nu_n \leq e^{s_n(h_K(\bs{\theta}) + ||\bs{\theta}||\delta)}.
\ea
For the second integral, we use Cauchy-Schwarz to bound
\ba
\int_{K_\delta^c} e^{s_n \bs{\theta}\cdot\bs{\eta}} d\nu_n \leq \left(  \int e^{2 s_n \bs{\theta} \cdot \bs{\eta}} d\nu_n\right)^{1/2} \left( \nu_n(K_\delta^c)\right)^{1/2}.
\ea
Since $\eta,\bar{\eta}\geq 0$, we have the bound $\bs{\theta}\cdot\bs{\eta} \leq |\theta|\eta + |\bar{\theta}|\bar{\eta}$, so the integral in the first factor is bounded by $M_{\nu_n}(2s_n|\theta|,\,2s_n|\bar{\theta}|)$. 

Since $2s_n|\theta|/\De_n = 2|\theta|\De_n^{-\alpha} \to 0$, the arguments lie in $[0,\De_n s_{\max})$ for $\De_n \gg 0$, and proposition~\ref{mgfbound} with $s^{-1}(u) = \sqrt{2}\,u + O(u^2)$ gives
\ba
\frac{1}{2}\log M_{\nu_n}(2s_n|\theta|,\,2s_n|\bar{\theta}|) \leq \sqrt{2}\,s_n||\bs{\theta}||_1 + O(\De_n^{1-2\alpha}).
\ea
Writing the result in terms of $\De_n^{1-\beta}$ and using $\alpha \geq \beta$ gives
\ba
\frac{1}{\De_n^{1-\beta}}\log\int_{K_\delta^c} e^{s_n \bs{\theta}\cdot\bs{\eta}} d\nu_n \leq \sqrt{2}\De_n^{\beta-\alpha}(|\theta| + | \bar{\theta}|) + O(\De_n^{\beta-2\alpha}) +\frac{1}{2} \frac{\log\nu_n(K_\delta^c)}{\De_n^{1-\beta}}.
\ea
The first two terms remain $O(1)$ since $\alpha \geq \beta$, while the $(\star)$ condition sends the last term to $-\infty$, so the entire contribution is negligible in the $\limsup$. Taking the limit supremum of the non-vanishing term, we have
\ba
\limsup_{n\to\infty}\frac{1}{s_n}\int e^{s_n \bs{\theta}\cdot\bs{\eta} }d\nu_n \leq h_K(\bs{\theta}) + ||\bs{\theta}||\delta .
\ea
Since the bound holds for all $\delta>0$ and the LHS is independent of $\delta$, we can take $\delta \to 0^+$ and obtain $(\mathrm{LHS})\leq h_K(\bs{\theta})$.

To compute a lower bound, fix $\bs{\theta}$ and let $\bs{\eta}_* \in \text{supp}(\nu)$ achieve $\bs{\theta}\cdot\bs{\eta}_* = h_K(\bs{\theta})$, which exists by the compactness of $\nu$. Fix $\epsilon>0$ and consider the open ball $B = B(\bs{\eta}_*,\epsilon)$ in $\mathbb{R}^2$. Since $\bs{\eta}_*\in\text{supp}(\nu)$, we have the Portmanteau theorem applied to $\{\nu_n\}$, which weakly converges to $\nu$, to obtain the lower bound
\ba
\liminf_{n\to\infty} \nu_n(B) \geq \nu(B) >0.
\ea
In particular, $\nu_n(B) \geq c_\epsilon$ for all large enough $n$ and some positive constant $c_\epsilon$. We can restrict the integral to $B$ to find
\ba
\int e^{s_n \bs{\theta}\cdot \bs{\eta}} d\nu_n \geq \int_B e^{s_n \bs{\theta}\cdot \bs{\eta}} d\nu_n \geq e^{s_n(h_K(\bs{\theta}) - ||\bs{\theta}||\epsilon)}\nu_n(B).
\ea
Taking a $\log$ of both sides and dividing by $s_n$, we have

\ba
\frac{1}{s_n}\log \int e^{s_n \bs{\theta}\cdot \bs{\eta}} d\nu_n  \geq h_K(\bs{\theta}) -|| \bs{\theta}||\epsilon + \log(c_\epsilon)/s_n.
\ea
Since $s_n \to \infty$ for any $\alpha<1$, the last term vanishes. Taking the $\liminf$ and sending $\epsilon \to 0$ gives
\ba
\liminf_{n\to\infty}\frac{1}{s_n}\log \int e^{s_n \bs{\theta}\cdot \bs{\eta}} d\nu_n  \geq h_K(\bs{\theta}),
\ea
where combining with the upper bound and applying the squeeze theorem gives eq.~(\ref{laplaceintegralasym}) as desired.

We now prove the reverse statement.  Let $s_n = \De_n^{1-\beta}$ and $\bs{\theta} = \bs{x}/\sqrt{2}$. Assume that for all $\alpha \geq \beta$ and $\bs{x} \in \mathbb{R}^2$, matching holds:
\ba
\lambda(\alpha;\bs{x}) = \tilde{\lambda}_\nu(\bs{x}) = h_K(\bs{\theta}) \equiv \sup_{\bs{\eta}\in K}\bs{\theta}\cdot\bs{\eta},
\ea
where $K = \mathrm{conv}(\mathrm{supp}(\nu))$. We need to show that for every $\delta>0$,
\ba
\lim_{n\to\infty}\frac{\log(\nu_n(K_\delta^c))}{s_n} = -\infty.
\ea
The key observation is that matching identifies $h_K(\bs{\theta})$ as the limiting rescaled cumulant generating function (CGF) of $\{\nu_n\}$ at speed $s_n$. To see this, recall the coherent state approximation from the forward direction: for any $\beta>0$,
\ba
\frac{1}{s_n}\log|\mathcal{G}_n(\bs{x}/\De_n^\beta)| = \frac{1}{s_n}\log\int e^{s_n \bs{\theta}\cdot\bs{\eta}}d\nu_n(\bs{\eta}) + O(\De_n^{-\beta}).
\ea
The error term vanishes as $n \to \infty$ for any $\beta > 0$, so matching at $\alpha = \beta$ gives convergence of the rescaled CGF:
\ba
\label{rescaledCGFconvergence}
\frac{1}{s_n}\log\int e^{s_n\bs{\theta}\cdot\bs{\eta}}d\nu_n(\bs{\eta}) \to h_K(\bs{\theta}) \quad \forall \bs{\theta}\in\mathbb{R}^2.
\ea
Note that $h_K$ is the support function of the compact convex set $K$, so it is convex, finite on all of $\mathbb{R}^2$, and continuous. In particular, $h_K$ is finite in a neighborhood of $\bs{0}$. This places us in the setting of the G\"artner-Ellis theorem \cite{DemboZeitouni1998}, which provides the following unconditional upper bound: for any closed set $F\subset\mathbb{R}^2$,
\ba
\limsup_{n\to\infty}\frac{1}{s_n}\log\nu_n(F) \leq - \inf_{\bs{\eta}\in F} h^*(\bs{\eta}),
\ea
where $h_K^*$ is the Legendre-Fenchel transform of $h_K$. A standard result in convex analysis gives $h_K^* = \iota_K$, the convex indicator function of $K$:
\ba
h_K^*(\bs{\eta}) = \sup_{\bs{\theta}\in\mathbb{R}^2}\left\{\bs{\theta}\cdot\bs{\eta} - \sup_{\bs{\zeta}\in K}\bs{\theta}\cdot\bs{\zeta}\right\} = \begin{cases}0 & \bs{\eta}\in K \\ +\infty & \bs{\eta}\notin K\end{cases}.
\ea
Now fix $\delta > 0$ and let $F_\delta = \{\bs{\eta}\in\mathbb{R}^2\mid d(\bs{\eta},K)\geq \delta\}$, which is closed with $K_\delta^c \subset F_\delta$. Since $F_\delta$ is disjoint from $K$, we have $\inf_{\bs{\eta}\in F_\delta}\iota_K(\bs{\eta}) = +\infty$. Applying the G\"artner-Ellis upper bound to $F_\delta$ and using $\nu_n(K_\delta^c)\leq \nu_n(F_\delta)$:
\ba
\limsup_{n\to\infty} \frac{1}{s_n}\log\nu_n(K_\delta^c) \leq \limsup_{n\to\infty} \frac{1}{s_n}\log\nu_n(F_\delta) \leq -\infty.
\ea
Since $\delta > 0$ was arbitrary, this establishes the $(\star)$ condition.

\begin{rem}
Two features of this argument are worth highlighting. First, the proof uses only the G\"artner-Ellis \textit{upper bound}, which requires no regularity assumptions on the limiting CGF beyond pointwise convergence and finiteness in a neighborhood of the origin. In particular, the essential smoothness condition required for the full G\"artner-Ellis LDP lower bound is not needed. This is crucial, since $h_K$ fails to be differentiable whenever $K$ has flat faces (e.g.~when $K$ is a polytope). Second, the Legendre-Fenchel dual of a support function is always a convex indicator function, which takes the value $+\infty$ everywhere outside $K$. This is a much stronger statement than the ``good rate function" condition that the level sets of the Legendre-Fenchel dual are compact: it says that deviations outside $K$ are not merely exponentially rare, but \textit{super-exponentially} suppressed at any sub-leading speed relative to the CGF convergence.
\end{rem}

\end{appendix}

\newpage
\bibliography{biblio}{}

@article{Maldacena:1997re,
    author = "Maldacena, Juan Martin",
    title = "{The Large N limit of superconformal field theories and supergravity}",
    eprint = "hep-th/9711200",
    archivePrefix = "arXiv",
    reportNumber = "HUTP-97-A097, HUTP-98-A097",
    doi = "10.4310/ATMP.1998.v2.n2.a1",
    journal = "Adv. Theor. Math. Phys.",
    volume = "2",
    pages = "231--252",
    year = "1998"
}

@article{Gubser:1998bc,
    author = "Gubser, S. S. and Klebanov, Igor R. and Polyakov, Alexander M.",
    title = "{Gauge theory correlators from noncritical string theory}",
    eprint = "hep-th/9802109",
    archivePrefix = "arXiv",
    reportNumber = "PUPT-1767",
    doi = "10.1016/S0370-2693(98)00377-3",
    journal = "Phys. Lett. B",
    volume = "428",
    pages = "105--114",
    year = "1998"
}

@ARTICLE{Simmons-Duffin2016-nn,
  title         = "{TASI} Lectures on the Conformal Bootstrap",
  author        = "Simmons-Duffin, David",
  month         =  feb,
  year          =  2016,
  copyright     = "http://arxiv.org/licenses/nonexclusive-distrib/1.0/",
  archivePrefix = "arXiv",
  primaryClass  = "hep-th",
  eprint        = "1602.07982"
}

@article{Benjamin2023-qo,
    author = "Benjamin, Nathan and Lee, Jaeha and Ooguri, Hirosi and Simmons-Duffin, David",
    title = "{Universal asymptotics for high energy CFT data}",
    eprint = "2306.08031",
    archivePrefix = "arXiv",
    primaryClass = "hep-th",
    reportNumber = "CALT-TH 2023-014, IPMU 23-0020",
    doi = "10.1007/JHEP03(2024)115",
    journal = "JHEP",
    volume = "03",
    pages = "115",
    year = "2024"
}

@article{Abajian2023-xw,
    author = "Abajian, Jacob and Aprile, Francesco and Myers, Robert C. and Vieira, Pedro",
    title = "{Holography and correlation functions of huge operators: spacetime bananas}",
    eprint = "2306.15105",
    archivePrefix = "arXiv",
    primaryClass = "hep-th",
    doi = "10.1007/JHEP12(2023)058",
    journal = "JHEP",
    volume = "12",
    pages = "058",
    year = "2023"
}

@article{Paulos2016-mh,
    author = "Paulos, Miguel F. and Penedones, Joao and Toledo, Jonathan and van Rees, Balt C. and Vieira, Pedro",
    title = "{The S-matrix bootstrap. Part I: QFT in AdS}",
    eprint = "1607.06109",
    archivePrefix = "arXiv",
    primaryClass = "hep-th",
    reportNumber = "CERN-TH-2016-162",
    doi = "10.1007/JHEP11(2017)133",
    journal = "JHEP",
    volume = "11",
    pages = "133",
    year = "2017"
}

@article{Kim2015-sg,
    author = "Kim, Hyungrok and Kravchuk, Petr and Ooguri, Hirosi",
    title = "{Reflections on Conformal Spectra}",
    eprint = "1510.08772",
    archivePrefix = "arXiv",
    primaryClass = "hep-th",
    reportNumber = "CALT-TH-2015-053, IPMU-15-0184",
    doi = "10.1007/JHEP04(2016)184",
    journal = "JHEP",
    volume = "04",
    pages = "184",
    year = "2016"
}

@article{Witten:1998qj,
    author = "Witten, Edward",
    title = "{Anti-de Sitter space and holography}",
    eprint = "hep-th/9802150",
    archivePrefix = "arXiv",
    reportNumber = "IASSNS-HEP-98-15",
    doi = "10.4310/ATMP.1998.v2.n2.a2",
    journal = "Adv. Theor. Math. Phys.",
    volume = "2",
    pages = "253--291",
    year = "1998"
}

@article{churchill1,
	author = {Edmund Churchill},
	doi = {10.1214/aoms/1177730987},
	journal = {The Annals of Mathematical Statistics},
	number = {2},
	pages = {244 -- 246},
	publisher = {Institute of Mathematical Statistics},
	title = {{Information Given by Odd Moments}},
	url = {https://doi.org/10.1214/aoms/1177730987},
	volume = {17},
	year = {1946}}

@article{ferrero2,
    author = "Ferrero, Pietro and Meneghelli, Carlo",
    title = "{Unmixing the Wilson line defect CFT. Part II. Analytic bootstrap}",
    eprint = "2312.12551",
    archivePrefix = "arXiv",
    primaryClass = "hep-th",
    doi = "10.1007/JHEP06(2024)010",
    journal = "JHEP",
    volume = "06",
    pages = "010",
    year = "2024"
}

@article{vanRees:2024xkb,
    author = "van Rees, Balt C.",
    title = "{Theorems for the lightcone bootstrap}",
    eprint = "2412.06907",
    archivePrefix = "arXiv",
    primaryClass = "hep-th",
    doi = "10.21468/SciPostPhys.18.6.207",
    journal = "SciPost Phys.",
    volume = "18",
    number = "6",
    pages = "207",
    year = "2025"
}

@article{Hogervorst:2017sfd,
    author = "Hogervorst, Matthijs and van Rees, Balt C.",
    title = "{Crossing symmetry in alpha space}",
    eprint = "1702.08471",
    archivePrefix = "arXiv",
    primaryClass = "hep-th",
    reportNumber = "YITP-SB-17-7",
    doi = "10.1007/JHEP11(2017)193",
    journal = "JHEP",
    volume = "11",
    pages = "193",
    year = "2017"
}

@article{Hogervorst:2013sma,
    author = "Hogervorst, Matthijs and Rychkov, Slava",
    title = "{Radial Coordinates for Conformal Blocks}",
    eprint = "1303.1111",
    archivePrefix = "arXiv",
    primaryClass = "hep-th",
    reportNumber = "CERN-PH-TH-2013-043, LPTENS-13-05",
    doi = "10.1103/PhysRevD.87.106004",
    journal = "Phys. Rev. D",
    volume = "87",
    pages = "106004",
    year = "2013"
}

@article{Pappadopulo:2012jk,
    author = "Pappadopulo, Duccio and Rychkov, Slava and Espin, Johnny and Rattazzi, Riccardo",
    title = "{OPE Convergence in Conformal Field Theory}",
    eprint = "1208.6449",
    archivePrefix = "arXiv",
    primaryClass = "hep-th",
    reportNumber = "LPTENS-12-31",
    doi = "10.1103/PhysRevD.86.105043",
    journal = "Phys. Rev. D",
    volume = "86",
    pages = "105043",
    year = "2012"
}

@article{Hijano:2015zsa,
    author = "Hijano, Eliot and Kraus, Per and Perlmutter, Eric and Snively, River",
    title = "{Witten Diagrams Revisited: The AdS Geometry of Conformal Blocks}",
    eprint = "1508.00501",
    archivePrefix = "arXiv",
    primaryClass = "hep-th",
    doi = "10.1007/JHEP01(2016)146",
    journal = "JHEP",
    volume = "01",
    pages = "146",
    year = "2016"
}

@article{Mukhametzhanov:2018zja,
    author = "Mukhametzhanov, Baur and Zhiboedov, Alexander",
    title = "{Analytic Euclidean Bootstrap}",
    eprint = "1808.03212",
    archivePrefix = "arXiv",
    primaryClass = "hep-th",
    doi = "10.1007/JHEP10(2019)270",
    journal = "JHEP",
    volume = "10",
    pages = "270",
    year = "2019"
}

@article{Qiao:2017xif,
    author = "Qiao, Jiaxin and Rychkov, Slava",
    title = "{A tauberian theorem for the conformal bootstrap}",
    eprint = "1709.00008",
    archivePrefix = "arXiv",
    primaryClass = "hep-th",
    reportNumber = "CERN-TH-2017-176",
    doi = "10.1007/JHEP12(2017)119",
    journal = "JHEP",
    volume = "12",
    pages = "119",
    year = "2017"
}

@article{Asplund:2014coa,
    author = "Asplund, Curtis T. and Bernamonti, Alice and Galli, Federico and Hartman, Thomas",
    title = "{Holographic Entanglement Entropy from 2d CFT: Heavy States and Local Quenches}",
    eprint = "1410.1392",
    archivePrefix = "arXiv",
    primaryClass = "hep-th",
    doi = "10.1007/JHEP02(2015)171",
    journal = "JHEP",
    volume = "02",
    pages = "171",
    year = "2015"
}

@article{Ceyhan:2025qrj,
    author = "Ceyhan, Fikret and Faulkner, Thomas",
    title = "{Bounds on CFT correlations from the thermal partition function}",
    eprint = "2510.24042",
    archivePrefix = "arXiv",
    primaryClass = "hep-th",
    month = "10",
    year = "2025"
}

@article{Karlsson:2021duj,
    author = "Karlsson, Robin and Parnachev, Andrei and Tadi{\'c}, Petar",
    title = "{Thermalization in large-N CFTs}",
    eprint = "2102.04953",
    archivePrefix = "arXiv",
    primaryClass = "hep-th",
    doi = "10.1007/JHEP09(2021)205",
    journal = "JHEP",
    volume = "09",
    pages = "205",
    year = "2021"
}

@article{Alkalaev:2024knk,
    author = "Alkalaev, Konstantin and Litvinov, Pavel",
    title = "{A note on the large-c conformal block asymptotics and {\ensuremath{\alpha}}-heavy operators}",
    eprint = "2407.12986",
    archivePrefix = "arXiv",
    primaryClass = "hep-th",
    doi = "10.1016/j.nuclphysb.2024.116741",
    journal = "Nucl. Phys. B",
    volume = "1009",
    pages = "116741",
    year = "2024"
}

@article{Arutyunov:2003ad,
    author = "Arutyunov, G. and Penati, S. and Santambrogio, A. and Sokatchev, E.",
    title = "{Four point correlators of BPS operators in N=4 SYM at order g**4}",
    eprint = "hep-th/0305060",
    archivePrefix = "arXiv",
    reportNumber = "AEI-2003-039, BICOCCA-FT-03-13, IFUM-758-FT, LAPTH-980-03",
    doi = "10.1016/j.nuclphysb.2003.07.027",
    journal = "Nucl. Phys. B",
    volume = "670",
    pages = "103--147",
    year = "2003"
}

@article{Eden:2012fe,
    author = "Eden, Burkhard and Heslop, Paul and Korchemsky, Gregory P. and Smirnov, Vladimir A. and Sokatchev, Emery",
    title = "{Five-loop Konishi in N=4 SYM}",
    eprint = "1202.5733",
    archivePrefix = "arXiv",
    primaryClass = "hep-th",
    reportNumber = "CERN-PH-TH-2012-053, DCPT-12-11, HU-EP-12-07, HU-MATH-2012-04, LAPTH-011-12, IPHT-T12-014",
    doi = "10.1016/j.nuclphysb.2012.04.015",
    journal = "Nucl. Phys. B",
    volume = "862",
    pages = "123--166",
    year = "2012"
}

@article{Chicherin:2018avq,
    author = "Chicherin, Dmitry and Georgoudis, Alessandro and Gon{\c{c}}alves, Vasco and Pereira, Raul",
    title = "{All five-loop planar four-point functions of half-BPS operators in $\mathcal N=4$ SYM}",
    eprint = "1809.00551",
    archivePrefix = "arXiv",
    primaryClass = "hep-th",
    doi = "10.1007/JHEP11(2018)069",
    journal = "JHEP",
    volume = "11",
    pages = "069",
    year = "2018"
}

@article{Chicherin_2016,
   title={All three-loop four-point correlators of half-BPS operators in planar $ \mathcal{N}  = 4$ SYM},
   volume={2016},
   ISSN={1029-8479},
   url={http://dx.doi.org/10.1007/JHEP08(2016)053},
   DOI={10.1007/jhep08(2016)053},
   number={8},
   journal={Journal of High Energy Physics},
   publisher={Springer Science and Business Media LLC},
   author={Chicherin, Dmitry and Drummond, James and Heslop, Paul and Sokatchev, Emery},
   year={2016},
   month=aug }

@article{Balasubramanian:2001nh,
    author = "Balasubramanian, Vijay and Berkooz, Micha and Naqvi, Asad and Strassler, Matthew J.",
    title = "{Giant gravitons in conformal field theory}",
    eprint = "hep-th/0107119",
    archivePrefix = "arXiv",
    reportNumber = "UPR-T-943, WIS-15-01-DPP",
    doi = "10.1088/1126-6708/2002/04/034",
    journal = "JHEP",
    volume = "04",
    pages = "034",
    year = "2002"
}

@article{Corley:2001zk,
    author = "Corley, Steve and Jevicki, Antal and Ramgoolam, Sanjaye",
    title = "{Exact correlators of giant gravitons from dual N=4 SYM theory}",
    eprint = "hep-th/0111222",
    archivePrefix = "arXiv",
    reportNumber = "BROWN-HET-1292",
    doi = "10.4310/ATMP.2001.v5.n4.a6",
    journal = "Adv. Theor. Math. Phys.",
    volume = "5",
    pages = "809--839",
    year = "2002"
}

@article{McGreevy:2000cw,
    author = "McGreevy, John and Susskind, Leonard and Toumbas, Nicolaos",
    title = "{Invasion of the giant gravitons from Anti-de Sitter space}",
    eprint = "hep-th/0003075",
    archivePrefix = "arXiv",
    reportNumber = "SU-ITP-00-09",
    doi = "10.1088/1126-6708/2000/06/008",
    journal = "JHEP",
    volume = "06",
    pages = "008",
    year = "2000"
}

@article{Chen:2025yxg,
    author = "Chen, Junding and Jiang, Yunfeng and Zhou, Xinan",
    title = "{Giant Graviton Correlators as Defect Systems}",
    eprint = "2503.22987",
    archivePrefix = "arXiv",
    primaryClass = "hep-th",
    reportNumber = "USTC-ICTS/PCFT-25-14",
    doi = "10.1103/hg9p-hblr",
    journal = "Phys. Rev. Lett.",
    volume = "135",
    number = "8",
    pages = "081602",
    year = "2025"
}

@article{Kazakov:2024ald,
    author = "Kazakov, Vladimir and Murali, Harish and Vieira, Pedro",
    title = "{Huge BPS operators and fluid dynamics in $ \mathcal{N}=4 $ SYM}",
    eprint = "2406.01798",
    archivePrefix = "arXiv",
    primaryClass = "hep-th",
    doi = "10.1007/JHEP09(2025)142",
    journal = "JHEP",
    volume = "09",
    pages = "142",
    year = "2025"
}

@article{Ivanovskiy:2024vel,
    author = "Ivanovskiy, Vyacheslav and Komatsu, Shota and Mishnyakov, Victor and Terziev, Nikolay and Zaigraev, Nikita and Zarembo, Konstantin",
    title = "{Vacuum Condensates on the Coulomb Branch}",
    eprint = "2405.19043",
    archivePrefix = "arXiv",
    primaryClass = "hep-th",
    month = "5",
    year = "2024"
}

@article{Grisaru:2000zn,
    author = "Grisaru, Marcus T. and Myers, Robert C. and Tafjord, Oyvind",
    title = "{SUSY and goliath}",
    eprint = "hep-th/0008015",
    archivePrefix = "arXiv",
    reportNumber = "MCGILL-00-21, BRX-TH-472",
    doi = "10.1088/1126-6708/2000/08/040",
    journal = "JHEP",
    volume = "08",
    pages = "040",
    year = "2000"
}

@article{Kraus:2018pax,
    author = "Kraus, Per and Sivaramakrishnan, Allic",
    title = "{Light-state Dominance from the Conformal Bootstrap}",
    eprint = "1812.02226",
    archivePrefix = "arXiv",
    primaryClass = "hep-th",
    doi = "10.1007/JHEP08(2019)013",
    journal = "JHEP",
    volume = "08",
    pages = "013",
    year = "2019"
}

@article{Poland:2024hvb,
    author = "Poland, David and Rogelberg, Gordon",
    title = "{Moments and saddles of heavy CFT correlators}",
    eprint = "2501.00092",
    archivePrefix = "arXiv",
    primaryClass = "hep-th",
    doi = "10.1007/JHEP10(2025)100",
    journal = "JHEP",
    volume = "10",
    pages = "100",
    year = "2025"
}

@article{Dey:2024nje,
    author = "Dey, Indranil and Pal, Sridip and Qiao, Jiaxin",
    title = "{A universal inequality on the unitary 2D CFT partition function}",
    eprint = "2410.18174",
    archivePrefix = "arXiv",
    primaryClass = "hep-th",
    reportNumber = "CALT-TH 2024-039",
    doi = "10.1007/JHEP07(2025)163",
    journal = "JHEP",
    volume = "07",
    pages = "163",
    year = "2025"
}

@article{Hartman:2014oaa,
    author = "Hartman, Thomas and Keller, Christoph A. and Stoica, Bogdan",
    title = "{Universal Spectrum of 2d Conformal Field Theory in the Large c Limit}",
    eprint = "1405.5137",
    archivePrefix = "arXiv",
    primaryClass = "hep-th",
    reportNumber = "CALT-68-2889, RUNHETC-2014-07",
    doi = "10.1007/JHEP09(2014)118",
    journal = "JHEP",
    volume = "09",
    pages = "118",
    year = "2014"
}

@article{Qiao:2020bcs,
    author = "Qiao, Jiaxin",
    title = "{Classification of Convergent OPE Channels for Lorentzian CFT Four-Point Functions}",
    eprint = "2005.09105",
    archivePrefix = "arXiv",
    primaryClass = "hep-th",
    doi = "10.21468/SciPostPhys.13.4.093",
    journal = "SciPost Phys.",
    volume = "13",
    number = "4",
    pages = "093",
    year = "2022"
}

@article{Caron-Huot:2017vep,
    author = "Caron-Huot, Simon",
    title = "{Analyticity in Spin in Conformal Theories}",
    eprint = "1703.00278",
    archivePrefix = "arXiv",
    primaryClass = "hep-th",
    doi = "10.1007/JHEP09(2017)078",
    journal = "JHEP",
    volume = "09",
    pages = "078",
    year = "2017"
}

@article{Vescovi:2021fjf,
    author = "Vescovi, Edoardo",
    title = "{Four-point function of determinant operators in $\mathcal{N}=4$ SYM}",
    eprint = "2101.05117",
    archivePrefix = "arXiv",
    primaryClass = "hep-th",
    doi = "10.1103/PhysRevD.103.106001",
    journal = "Phys. Rev. D",
    volume = "103",
    number = "10",
    pages = "106001",
    year = "2021"
}

@article{Hartman:2015lfa,
    author = "Hartman, Thomas and Jain, Sachin and Kundu, Sandipan",
    title = "{Causality Constraints in Conformal Field Theory}",
    eprint = "1509.00014",
    archivePrefix = "arXiv",
    primaryClass = "hep-th",
    doi = "10.1007/JHEP05(2016)099",
    journal = "JHEP",
    volume = "05",
    pages = "099",
    year = "2016"
}

@article{Kravchuk:2020scc,
    author = "Kravchuk, Petr and Qiao, Jiaxin and Rychkov, Slava",
    title = "{Distributions in CFT. Part I. Cross-ratio space}",
    eprint = "2001.08778",
    archivePrefix = "arXiv",
    primaryClass = "hep-th",
    doi = "10.1007/JHEP05(2020)137",
    journal = "JHEP",
    volume = "05",
    pages = "137",
    year = "2020"
}

@book{DemboZeitouni1998,
  author    = {Dembo, Amir and Zeitouni, Ofer},
  title     = {Large Deviations Techniques and Applications},
  publisher = {Springer},
  year      = {1998},
  edition   = {2nd},
  series    = {Applications of Mathematics},
  volume    = {38},
  doi       = {10.1007/978-1-4612-5320-4}
}

@article{Chen:2026ium,
    author = "Chen, Junding and Jiang, Yunfeng and Zhou, Xinan",
    title = "{Defect Approach to Giant Graviton Dynamics}",
    eprint = "2602.13570",
    archivePrefix = "arXiv",
    primaryClass = "hep-th",
    reportNumber = "USTC-ICTS/PCFT-26-14",
    month = "2",
    year = "2026"
}

@book{Chung2001,
  author    = {Chung, Kai Lai},
  title     = {A Course in Probability Theory},
  edition   = {3},
  publisher = {Academic Press},
  address   = {San Diego, CA},
  year      = {2001},
}

@book{Durrett2019,
  author    = {Durrett, Rick},
  title     = {Probability: Theory and Examples},
  edition   = {5},
  publisher = {Cambridge University Press},
  address   = {Cambridge},
  year      = {2019},
  series    = {Cambridge Series in Statistical and Probabilistic Mathematics},
}

@inbook{McGillemCooper1984,
  author    = {McGillem, Clare D. and Cooper, George R.},
  title     = {Continuous and Discrete Signal and System Analysis},
  edition   = {2},
  publisher = {Holt, Rinehart and Winston},
  year      = {1984},
  pages     = {118},
  isbn      = {0-03-061703-0},
}

@article{Heyl:2017blm,
    author = "Heyl, Markus",
    title = "{Dynamical quantum phase transitions: a review}",
    eprint = "1709.07461",
    archivePrefix = "arXiv",
    primaryClass = "cond-mat.stat-mech",
    doi = "10.1088/1361-6633/aaaf9a",
    journal = "Rept. Prog. Phys.",
    volume = "81",
    number = "5",
    pages = "054001",
    year = "2018"
}

@article{vanRees:2022zmr,
    author = "van Rees, Balt C. and Zhao, Xiang",
    title = "{Quantum Field Theory in AdS Space instead of Lehmann-Symanzik-Zimmerman Axioms}",
    eprint = "2210.15683",
    archivePrefix = "arXiv",
    primaryClass = "hep-th",
    doi = "10.1103/PhysRevLett.130.191601",
    journal = "Phys. Rev. Lett.",
    volume = "130",
    number = "19",
    pages = "191601",
    year = "2023"
}

@article{vanRees:2023fcf,
    author = "van Rees, Balt C. and Zhao, Xiang",
    title = "{Flat-space Partial Waves From Conformal OPE Densities}",
    eprint = "2312.02273",
    archivePrefix = "arXiv",
    primaryClass = "hep-th",
    month = "12",
    year = "2023"
}

@book{Korevaar2004,
  author = {Korevaar, Jacob},
  title = {Tauberian Theory: {A} Century of Developments},
  series = {Grundlehren der mathematischen Wissenschaften},
  volume = {329},
  publisher = {Springer},
  address = {Berlin, Heidelberg},
  year = {2004},
  isbn = {978-3-540-21058-0},
  doi = {10.1007/978-3-662-10225-1}
}

@article{Strassler:1992zr,
    author = "Strassler, Matthew J.",
    title = "{Field theory without Feynman diagrams: One loop effective actions}",
    eprint = "hep-ph/9205205",
    archivePrefix = "arXiv",
    reportNumber = "SLAC-PUB-5757",
    doi = "10.1016/0550-3213(92)90098-V",
    journal = "Nucl. Phys. B",
    volume = "385",
    pages = "145--184",
    year = "1992"
}

@article{Maxfield:2017rkn,
    author = "Maxfield, Henry",
    title = "{A view of the bulk from the worldline}",
    eprint = "1712.00885",
    archivePrefix = "arXiv",
    primaryClass = "hep-th",
    month = "12",
    year = "2017"
}

@article{Kulkarni:2024ghc,
    author = "Kulkarni, Raghotham A. and Rahul and Bhattacharyya, Soham and Kothawala, Dawood",
    title = "{Worldline EFT treatment of quadratic and cubic gravity theories}",
    eprint = "2410.01266",
    archivePrefix = "arXiv",
    primaryClass = "gr-qc",
    doi = "10.1103/yx1z-1wzs",
    journal = "Phys. Rev. D",
    volume = "112",
    number = "12",
    pages = "124028",
    year = "2025"
}

@article{Bohra:2021zyw,
    author = "Bohra, Hardik and Kakkar, Ashish and Sivaramakrishnan, Allic",
    title = "{Information geometry and holographic correlators}",
    eprint = "2110.15359",
    archivePrefix = "arXiv",
    primaryClass = "hep-th",
    doi = "10.1007/JHEP04(2022)037",
    journal = "JHEP",
    volume = "04",
    pages = "037",
    year = "2022"
}

@article{Beisert:2003tq,
    author = "Beisert, Niklas and Staudacher, Matthias",
    title = "{The N=4 SYM integrable super spin chain}",
    eprint = "hep-th/0307042",
    archivePrefix = "arXiv",
    doi = "10.1016/j.nuclphysb.2003.08.015",
    journal = "Nucl. Phys. B",
    volume = "670",
    pages = "439--463",
    year = "2003"
}

@article{Beisert:2010jr,
    author = "Beisert, Niklas and others",
    title = "{Review of AdS/CFT Integrability: An Overview}",
    eprint = "1012.3982",
    archivePrefix = "arXiv",
    primaryClass = "hep-th",
    doi = "10.1007/s11005-011-0529-2",
    journal = "Lett. Math. Phys.",
    volume = "99",
    pages = "3--32",
    year = "2012"
}

@article{Jiang:2019xdz,
    author = "Jiang, Yunfeng and Komatsu, Shota and Vescovi, Edoardo",
    title = "{Structure constants in $\mathcal{N}=4$ SYM at finite coupling as worldsheet g-function}",
    eprint = "1906.07733",
    archivePrefix = "arXiv",
    primaryClass = "hep-th",
    doi = "10.1007/JHEP07(2020)037",
    journal = "JHEP",
    volume = "07",
    pages = "037",
    year = "2020"
}

@article{Penedones:2010ue,
    author = "Penedones, Joao",
    title = "{Writing CFT correlation functions as AdS scattering amplitudes}",
    eprint = "1011.1485",
    archivePrefix = "arXiv",
    primaryClass = "hep-th",
    doi = "10.1007/JHEP03(2011)025",
    journal = "JHEP",
    volume = "03",
    pages = "025",
    year = "2011"
}

@article{Paulos:2016fap,
    author = "Paulos, Miguel F. and Penedones, Joao and Toledo, Jonathan and van Rees, Balt C. and Vieira, Pedro",
    title = "{The S-matrix bootstrap. Part I: QFT in AdS}",
    eprint = "1607.06109",
    archivePrefix = "arXiv",
    primaryClass = "hep-th",
    doi = "10.1007/JHEP11(2017)133",
    journal = "JHEP",
    volume = "11",
    pages = "133",
    year = "2017"
}

@article{Paulos:2017fhb,
    author = "Paulos, Miguel F. and Penedones, Joao and Toledo, Jonathan and van Rees, Balt C. and Vieira, Pedro",
    title = "{The S-matrix bootstrap II: two dimensional amplitudes}",
    eprint = "1607.06110",
    archivePrefix = "arXiv",
    primaryClass = "hep-th",
    doi = "10.1007/JHEP11(2017)143",
    journal = "JHEP",
    volume = "11",
    pages = "143",
    year = "2017"
}

@article{Ruppeiner:1995zz,
    author = "Ruppeiner, George",
    title = "{Riemannian geometry in thermodynamic fluctuation theory}",
    doi = "10.1103/RevModPhys.67.605",
    journal = "Rev. Mod. Phys.",
    volume = "67",
    pages = "605--659",
    year = "1995",
    note = "[Erratum: Rev.Mod.Phys. 68, 313 (1996)]"
}

@book{Amari:2000ig,
    author = "Amari, Shun-ichi and Nagaoka, Hiroshi",
    title = "{Methods of Information Geometry}",
    publisher = "American Mathematical Society",
    series = "Translations of Mathematical Monographs",
    volume = "191",
    year = "2000"
}

@article{Herrmann:2001prl,
    author = "Herrmann, J. and Griebner, U. and Zhavoronkov, N. and Husakou, A. and Nickel, D. and Knight, J. C. and Wadsworth, W. J. and Russell, P. St. J. and Korn, G.",
    title = "{Experimental Evidence for Supercontinuum Generation by Fission of Higher-Order Solitons in Photonic Fibers}",
    doi = "10.1103/PhysRevLett.87.203901",
    journal = "Phys. Rev. Lett.",
    volume = "87",
    pages = "203901",
    year = "2001"
}

@article{Husakou:2002josab,
    author = "Husakou, A. V. and Herrmann, J.",
    title = "{Supercontinuum generation, four-wave mixing, and fission of higher-order solitons in photonic-crystal fibers}",
    doi = "10.1364/JOSAB.19.002171",
    journal = "J. Opt. Soc. Am. B",
    volume = "19",
    number = "9",
    pages = "2171--2182",
    year = "2002"
}

@article{Dudley:2006rmp,
    author ="Dudley, John M. and Genty, Goery and Coen, St\'ephane",

    title = "{Supercontinuum generation in photonic crystal fiber}",
    doi = "10.1103/RevModPhys.78.1135",
    journal = "Rev. Mod. Phys.",
    volume = "78",
    pages = "1135--1184",
    year = "2006"
}

@book{Agrawal:2019nfo,
    author = "Agrawal, Govind P.",
    title = "{Nonlinear Fiber Optics}",
    edition = "6th",
    publisher = "Academic Press",
    year = "2019"
}

@article{Linardopoulos:2026mut,
    author = "Linardopoulos, Georgios and Park, Chanyong",
    title = "{Heavy holographic correlators in defect conformal field theories}",
    eprint = "2601.15736",
    archivePrefix = "arXiv",
    primaryClass = "hep-th",
    reportNumber = "APCTP Pre2025-009",
    month = "1",
    year = "2026"
}

@book{Williams1991,
    author    = {Williams, David},
    title     = {Probability with Martingales},
    publisher = {Cambridge University Press},
    year      = {1991},
    series    = {Cambridge Mathematical Textbooks},
    address   = {Cambridge},
    isbn      = {978-0-521-40605-5}
}

@book{Billingsley1999,
    author    = {Billingsley, Patrick},
    title     = {Convergence of Probability Measures},
    edition   = {2nd},
    publisher = {John Wiley \& Sons},
    year      = {1999},
    series    = {Wiley Series in Probability and Statistics},
    address   = {New York},
    isbn      = {978-0-471-19745-4}
}

@article{Prokhorov1956,
    author    = {Prokhorov, Yu. V.},
    title     = {Convergence of Random Processes and Limit Theorems in Probability Theory},
    journal   = {Theory of Probability \& Its Applications},
    volume    = {1},
    number    = {2},
    pages     = {157--214},
    year      = {1956},
    doi       = {10.1137/1101016}
}

@book{Munkres2000,
    author    = {Munkres, James R.},
    title     = {Topology},
    edition   = {2nd},
    publisher = {Prentice Hall},
    year      = {2000},
    address   = {Upper Saddle River, NJ},
    isbn      = {978-0-131-81629-9}
}

@article{Das:2020jhy,
    author = "Das, Diptarka and Kusuki, Yuya and Pal, Sridip",
    title = "{Universality in asymptotic bounds and its saturation in 2D CFT}",
    eprint = "2011.02482",
    archivePrefix = "arXiv",
    primaryClass = "hep-th",
    doi = "10.1007/JHEP04(2021)288",
    journal = "JHEP",
    volume = "04",
    pages = "288",
    year = "2021"
}

@article{Fitzpatrick:2016ive,
    author = "Fitzpatrick, A. Liam and Kaplan, Jared and Li, Daliang and Wang, Junpu",
    title = "{On information loss in AdS$_3$/CFT$_2$}",
    eprint = "1603.08925",
    archivePrefix = "arXiv",
    primaryClass = "hep-th",
    doi = "10.1007/JHEP05(2016)109",
    journal = "JHEP",
    volume = "05",
    pages = "109",
    year = "2016"
}

@article{Kraus:2016nwo,
    author = "Kraus, Per and Maloney, Alexander",
    title = "{A Cardy formula for three-point coefficients or how the black hole got its spots}",
    eprint = "1608.03284",
    archivePrefix = "arXiv",
    primaryClass = "hep-th",
    doi = "10.1007/JHEP05(2017)160",
    journal = "JHEP",
    volume = "05",
    pages = "160",
    year = "2017"
}

@article{Das:2017vej,
    author = "Das, Diptarka and Datta, Shouvik and Pal, Sridip",
    title = "{Charged structure constants from modularity}",
    eprint = "1706.04612",
    archivePrefix = "arXiv",
    primaryClass = "hep-th",
    doi = "10.1007/JHEP11(2017)183",
    journal = "JHEP",
    volume = "11",
    pages = "183",
    year = "2017"
}

@article{Ghosh:2019rcj,
    author = "Ghosh, Animik and Maxfield, Henry and Turiaci, Gustavo J.",
    title = "{A universal Schwarzian sector in two-dimensional conformal field theories}",
    eprint = "1912.07654",
    archivePrefix = "arXiv",
    primaryClass = "hep-th",
    doi = "10.1007/JHEP05(2020)104",
    journal = "JHEP",
    volume = "05",
    pages = "104",
    year = "2020"
}

@article{Pal:2023cgk,
    author = "Pal, Sridip and Qiao, Jiaxin",
    title = "{Lightcone modular bootstrap and Tauberian theory: a Cardy-like formula for near-extremal black holes}",
    eprint = "2307.02587",
    archivePrefix = "arXiv",
    primaryClass = "hep-th",
    doi = "10.1007/s00023-024-01503-z",
    journal = "Annales Henri Poincare",
    volume = "26",
    pages = "787",
    year = "2025"
}

@article{Fitzpatrick:2015zha,
    author = "Fitzpatrick, A. Liam and Kaplan, Jared and Walters, Matthew T.",
    title = "{Virasoro conformal blocks and thermality from classical background fields}",
    eprint = "1501.05315",
    archivePrefix = "arXiv",
    primaryClass = "hep-th",
    doi = "10.1007/JHEP11(2015)200",
    journal = "JHEP",
    volume = "11",
    pages = "200",
    year = "2015"
}

@article{Collier:2019weq,
    author = "Collier, Scott and Maloney, Alexander and Maxfield, Henry and Tsiares, Ioannis",
    title = "{Universal dynamics of heavy operators in CFT$_2$}",
    eprint = "1912.00222",
    archivePrefix = "arXiv",
    primaryClass = "hep-th",
    doi = "10.1007/JHEP07(2020)074",
    journal = "JHEP",
    volume = "07",
    pages = "074",
    year = "2020"
}

@article{LIGOScientific:2016aoc,
    author = "{LIGO Scientific Collaboration and Virgo Collaboration}",
    collaboration = "LIGO Scientific, Virgo",
    title = "{Observation of Gravitational Waves from a Binary Black Hole Merger}",
    eprint = "1602.03837",
    archivePrefix = "arXiv",
    primaryClass = "gr-qc",
    doi = "10.1103/PhysRevLett.116.061102",
    journal = "Phys. Rev. Lett.",
    volume = "116",
    number = "6",
    pages = "061102",
    year = "2016"
}

@article{EventHorizonTelescope:2019dse,
    author = "{Event Horizon Telescope Collaboration}",
    collaboration = "Event Horizon Telescope",
    title = "{First M87 Event Horizon Telescope Results. I. The Shadow of the Supermassive Black Hole}",
    eprint = "1906.11238",
    archivePrefix = "arXiv",
    primaryClass = "astro-ph.GA",
    doi = "10.3847/2041-8213/ab0ec7",
    journal = "Astrophys. J. Lett.",
    volume = "875",
    pages = "L1",
    year = "2019"
}

@article{Bambusi1999,
  author    = {Bambusi, D. and Graffi, S. and Paul, T.},
  title     = {Long time semiclassical approximation of quantum flows:
               {A} proof of the {E}hrenfest time},
  journal   = {Asymptotic Analysis},
  volume    = {21},
  pages     = {149--160},
  year      = {1999}
}

@article{Kim:2023sig,
    author = "Kim, Seok and Kundu, Suman and Lee, Eunwoo and Lee, Jaeha and Minwalla, Shiraz and Patel, Chintan",
    title = "{`Grey Galaxies' as an endpoint of the Kerr-AdS superradiant instability}",
    journal = "JHEP",
    volume = "11",
    year = "2023",
    pages = "024",
    doi = "10.1007/JHEP11(2023)024",
    eprint = "2305.08922",
    archivePrefix = "arXiv",
    primaryClass = "hep-th"
}

@article{Bajaj:2024utv,
    author = "Bajaj, Kabir and Kundu, Suman and Lee, Eunwoo and Minwalla, Shiraz and Patel, Chintan",
    title = "{Grey Galaxies in $AdS_5$}",
    eprint = "2412.06904",
    archivePrefix = "arXiv",
    primaryClass = "hep-th",
    year = "2024"
}

@article{Brunello:2025rhh,
    author = "Brunello, Giacomo and Caron-Huot, Simon and Crisanti, Giulio and Giroux, Mathieu and Smith, Sid",
    title = "{High-energy evolution in planar QCD to three loops: the non-conformal contribution}",
    eprint = "2508.03794",
    archivePrefix = "arXiv",
    primaryClass = "hep-ph",
    doi = "10.1007/JHEP11(2025)055",
    journal = "JHEP",
    volume = "11",
    pages = "055",
    year = "2025"
}

@article{Mueller:2018llt,
    author = "Mueller, Alfred H.",
    title = "{Conformal spacelike-timelike correspondence in QCD}",
    eprint = "1804.07249",
    archivePrefix = "arXiv",
    primaryClass = "hep-th",
    doi = "10.1007/JHEP08(2018)139",
    journal = "JHEP",
    volume = "08",
    pages = "139",
    year = "2018"
}

@article{Chiang:2026nmd,
    author = "Chiang, Li-Yuan and Poland, David and Rogelberg, Gordon",
    title = "{Moments in the CFT Landscape}",
    eprint = "2603.18140",
    archivePrefix = "arXiv",
    primaryClass = "hep-th",
    month = "3",
    year = "2026"
}

@article{Dijkgraaf:1996xw,
    author = "Dijkgraaf, Robbert and Moore, Gregory and Verlinde, Erik and Verlinde, Herman",
    title = "{Elliptic genera of symmetric products and second quantized strings}",
    journal = "Commun. Math. Phys.",
    volume = "185",
    year = "1997",
    pages = "197--209",
    doi = "10.1007/s002200050087",
    eprint = "hep-th/9608096",
    archivePrefix = "arXiv",
    primaryClass = "hep-th"
}

@article{AfkhamiJeddi:2020hde,
    author = "Afkhami-Jeddi, Nima and Cohn, Henry and Hartman, Thomas and Tajdini, Amirhossein",
    title = "{Free partition functions and an averaged holographic duality}",
    journal = "JHEP",
    volume = "01",
    year = "2021",
    pages = "130",
    doi = "10.1007/JHEP01(2021)130",
    eprint = "2006.04839",
    archivePrefix = "arXiv",
    primaryClass = "hep-th"
}

@article{Maloney:2020nni,
    author = "Maloney, Alexander and Witten, Edward",
    title = "{Averaging over Narain moduli space}",
    journal = "JHEP",
    volume = "10",
    year = "2020",
    pages = "187",
    doi = "10.1007/JHEP10(2020)187",
    eprint = "2006.04855",
    archivePrefix = "arXiv",
    primaryClass = "hep-th"
}

@article{Benjamin:2021wzr,
    author = "Benjamin, Nathan and Collier, Scott and Fitzpatrick, A. Liam and Maloney, Alexander and Perlmutter, Eric",
    title = "{Harmonic analysis of 2d CFT partition functions}",
    journal = "JHEP",
    volume = "09",
    year = "2021",
    pages = "174",
    doi = "10.1007/JHEP09(2021)174",
    eprint = "2107.10744",
    archivePrefix = "arXiv",
    primaryClass = "hep-th"
}

@article{Gromov:2013pga,
    author = "Gromov, Nikolay and Kazakov, Vladimir and Leurent, Sebastien and Volin, Dmytro",
    title = "{Quantum Spectral Curve for Planar $\mathcal{N} = 4$ Super-Yang-Mills Theory}",
    eprint = "1305.1939",
    archivePrefix = "arXiv",
    primaryClass = "hep-th",
    reportNumber = "IMPERIAL-TP-13-SL-02",
    doi = "10.1103/PhysRevLett.112.011602",
    journal = "Phys. Rev. Lett.",
    volume = "112",
    number = "1",
    pages = "011602",
    year = "2014"
}

@article{Berenstein:2002adc,
    author = "Berenstein, David Eliecer and Maldacena, Juan Martin and Nastase, Horatiu Stefan",
    editor = "Elias, V. and Epp, R. J. and Myers, Robert C.",
    title = "{Strings in flat space and pp waves from N=4 Super Yang Mills}",
    doi = "10.1063/1.1524550",
    journal = "AIP Conf. Proc.",
    volume = "646",
    number = "1",
    pages = "3--14",
    year = "2002"
}

@article{Hofman:2006xt,
    author = "Hofman, Diego M. and Maldacena, Juan Martin",
    title = "{Giant Magnons}",
    eprint = "hep-th/0604135",
    archivePrefix = "arXiv",
    doi = "10.1088/0305-4470/39/41/S17",
    journal = "J. Phys. A",
    volume = "39",
    pages = "13095--13118",
    year = "2006"
}

@article{deMelloKoch:2024sdf,
    author = "de Mello Koch, Robert and Kim, Minkyoo and Mahu, Augustine Larweh",
    title = "{A pedagogical introduction to restricted Schur polynomials with applications to heavy operators}",
    eprint = "2409.15751",
    archivePrefix = "arXiv",
    primaryClass = "hep-th",
    doi = "10.1142/S0217751X24300035",
    journal = "Int. J. Mod. Phys. A",
    volume = "39",
    number = "31",
    pages = "2430003",
    year = "2024"
}

@article{Lin:2004nb,
    author = "Lin, Hai and Lunin, Oleg and Maldacena, Juan Martin",
    title = "{Bubbling AdS space and 1/2 BPS geometries}",
    eprint = "hep-th/0409174",
    archivePrefix = "arXiv",
    reportNumber = "PUPT-2136",
    doi = "10.1088/1126-6708/2004/10/025",
    journal = "JHEP",
    volume = "10",
    pages = "025",
    year = "2004"
}

@article{FollandSitaram1997,
  author    = {Folland, Gerald B. and Sitaram, Alladi},
  title     = {The uncertainty principle: a mathematical survey},
  journal   = {J. Fourier Anal. Appl.},
  volume    = {3},
  number    = {3},
  pages     = {207--238},
  year      = {1997},
  doi       = {10.1007/BF02649110}
}

@article{Horowitz:2000fm,
    author = "Horowitz, Gary T. and Hubeny, Veronika E.",
    title = "{CFT description of small objects in AdS}",
    eprint = "hep-th/0009051",
    archivePrefix = "arXiv",
    doi = "10.1088/1126-6708/2000/10/027",
    journal = "JHEP",
    volume = "10",
    pages = "027",
    year = "2000"
}

@article{Hoyos:2021uff,
    author = "Hoyos, Carlos and Jokela, Niko and Vuorinen, Aleksi",
    title = "{Holographic approach to compact stars and their binary mergers}",
    eprint = "2112.08422",
    archivePrefix = "arXiv",
    primaryClass = "hep-th",
    reportNumber = "HIP-2021-49/TH",
    doi = "10.1016/j.ppnp.2022.103972",
    journal = "Prog. Part. Nucl. Phys.",
    volume = "126",
    pages = "103972",
    year = "2022"
}

@article{Abajian:2023bqv,
    author = "Abajian, Jacob and Aprile, Francesco and Myers, Robert C. and Vieira, Pedro",
    title = "{Correlation functions of huge operators in AdS$_{3}$/CFT$_{2}$: domes, doors and book pages}",
    eprint = "2307.13188",
    archivePrefix = "arXiv",
    primaryClass = "hep-th",
    doi = "10.1007/JHEP03(2024)118",
    journal = "JHEP",
    volume = "03",
    pages = "118",
    year = "2024"
}
\bibliographystyle{JHEP}

\end{document}